\documentclass[11pt]{article}

\usepackage{latexsym}
\usepackage{fullpage}

\usepackage{amssymb,amsmath,amsfonts,bm}
\usepackage{mathtools}
\usepackage{subfigure}
\usepackage{epic}
\usepackage{eepic}
\usepackage{enumerate}
\usepackage{amsthm}
\usepackage{vinayak}
\usepackage{graphicx}
\usepackage{verbatim}
\usepackage[mathscr]{eucal}
\usepackage{xspace}
\usepackage{stmaryrd}
\usepackage{paralist}
\usepackage[boxed, nokwfunc]{algorithm2e}
 \SetAlgoInsideSkip{smallskip}
\SetAlCapSkip{10pt}
 \SetAlFnt{\normalsize\rm}

\makeatletter
\let\oldabs\abs
\def\abs{\@ifstar{\oldabs}{\oldabs*}}
\let\oldnorm\norm
\def\norm{\@ifstar{\oldnorm}{\oldnorm*}}
\makeatother

\title{
Computing the Skorokhod Distance between Polygonal Traces
}
\author{
Rupak Majumdar \qquad
Vinayak S. Prabhu \\
\normalsize MPI-SWS\\
{\tt $\{$rupak,vinayak$\}$@mpi-sws.org}
}
 \date{}

\begin{document}
\maketitle

\thispagestyle{empty}

\begin{abstract}
The \emph{Skorokhod distance} is a  natural metric on traces of continuous and
hybrid systems.
For two traces, from $[0,T]$ to values in a metric space $\myO$,
it measures the best match between the traces when allowed continuous bijective timing
distortions.
Formally, it computes the infimum, over all timing distortions, of the
maximum of two components: the first component quantifies
the {\em timing discrepancy} of the timing distortion, and the second quantifies
the mismatch (in the metric space $\myO$) of the values under the timing distortion.
%
Skorokhod distances appear in various fundamental hybrid systems analysis concerns: 
from definitions of  hybrid systems semantics and notions of equivalence,
 to practical problems such as checking the closeness of models or
the quality of simulations.
Despite its popularity and extensive theoretical  use, the \emph{computation} problem for 
the Skorokhod distance
between two finite sampled-time hybrid traces has remained open. 

We address in this work
 the problem of computing the Skorokhod distance between two polygonal traces 
(these traces arise 
when sampled-time traces are completed by linear interpolation between sample points).
We provide the first 
algorithm to compute the exact Skorokhod distance 
when trace values are in $\reals^n$ for the  $L_1$, $L_2$, and $L_{\infty}$ norms.
Our algorithm, based on a reduction to \frechet distances, is fully polynomial-time,
 and incorporates novel polynomial-time procedures 
for a set of geometric primitives in $\reals^n$ over the three norms. 

\end{abstract}

\section{Introduction}

In approximation theories, we aim to quantify the difference between hybrid systems
by defining a metric on traces.
Given finite system traces 
$x: [0,T] \mapsto {\myO}$ mapping the time interval
$[0,T]$ to some metric space ${\myO}$, 
one simple method to obtain a metric
is to take the pointwise trace value difference: the difference between
two traces $x$ and $y$ is then $\sup_{t}\dist\left(x(t), y(t)\right)$ where
$\dist$ is the metric associated with the metric space ${\myO}$.
The restriction that we compare the value of $x$ at time $t$ to the  value of $y$
at the \emph{same}
time $t$ often too restrictive: a trace can have a large distance from its infinitesimally 
time-shifted version.
This motivates the study of the 
\emph{Skorokhod} metric~\cite{skoro}, which
allows for ``wiggle-room'' in both the trace values \emph{and} in the timeline.
The distortion of the timeline is specified by a \emph{retiming} function $\retime$ which
is a continuous bijective strictly increasing function from $\reals_+$ to  $\reals_+$.
Using the retiming function, we obtain the \emph{retimed trace} $x\left(\retime(t)\right)$ from the
original trace $x(t)$.
Intuitively,  in the retimed trace $x\left(\retime(t)\right)$, 
we see exactly the same values as before, in 
exactly the same order,  but the 
time duration between two values might now be different than
the corresponding duration in the original trace.
The amount of distortion for the retiming $\retime$  is given by $\sup_{t\geq 0} \abs{\retime(t)-t}$.
Using retiming functions, the Skorokhod distance between two traces $x$ and $y$ is defined to be
the least value over all possible retimings $\retime$ of:
\[\max\left(\sup_{t\in[0,T]} \abs{\retime(t)-t},\,  \sup_{t\in [0,T]}\dist\big(x\left(\retime(t)\right), y(t)\big)
\right).\]
The Skorokhod distance thus 
incorporates two components: the first component quantifies
the {\em timing discrepancy} of the timing distortion required to ``match'' two traces, 
and the second quantifies the  \emph{value mismatch}  (in the metric space ${\myO}$) 
of the values under the timing distortion.

The Skorokhod metric has been widely used in the semantics and analysis of
continuous, boolean, stochastic, and hybrid systems~\cite{CaspiB02,Broucke98}.
Despite its popularity, the \emph{computation} of the Skorokhod metric has not been studied,
except for some discrete-time cases (where there is a folklore dynamic programming algorithm); even the
exact computability for continuous piecewise-linear  traces (called {\em polygonal traces}), which arise when time-sampled traces are completed by linear interpolation, 
 was open.
We note that even for linear traces, the space of retiming functions is infinite; and that
a linear trace, after retiming, need not remain linear.

\smallskip\noindent\textbf{Our Contributions.}
In this paper, we show a fully polynomial-time algorithm for the computation of the Skorokhod
metric for polygonal traces for $L_1$, $L_2$, and $L_{\infty}$ norms. 
Our algorithm reduces the computation of  Skorokhod distances between continuous traces
in a normed space to computing a related distance, called the \emph{\frechet distance}, 
between curves (in a different normed space).
\frechet distances have been studied extensively in computational 
geometry (see \emph{e.g.}~\cite{BuchinBW08, Chambers2010, MaheshwariSSZ11}).
A celebrated paper by Alt and Godau~\cite{AltG95} gave a sketch of  a polynomial-time
algorithm for polygonal curves in $\reals^2$ in the $L_2$ norm.
We provide a generalization of this algorithm 
parameterized by a set of geometric 
primitives over the underlying metric space of curves.
We  also provide polynomial-time algorithms for these geometric primitives in $\reals^n$ 
for the $L_1$, $L_2$, and $L_{\infty}$ metrics, as well as for two 
derived metrics $L_1^{\skoro}$ and $L_2^{\skoro}$
which are required for computing the Skorokhod distance.
These algorithms  involve techniques from  linear programming (for $L_1$,  $L_{\infty}$
and $L_1^{\skoro}$ norms), and from vector algebra and convex geometry 
(for the $L_2$ and $L_2^{\skoro}$ norms).
Together, we get a fully polynomial time  algorithm for the Skorokhod distance in 
$\reals^n$ 
for the $L_1$, $L_2$, and $L_{\infty}$ metrics.
In addition, our constructions also provide a fully polynomial time algorithm for computing the
\frechet distance between finite polygonal curves in $\reals^n$ for these metrics.
For practical applications where only constant window retimings are of interest
(\emph{i.e.}, where the $k$-th  affine segment of $y$ can only be
matched to the portion of $x$ between the $(k-W)$-th and  $(k+W)$-th
 affine segments, for $W$ a  constant), our algorithm for the 
Skorokhod distance runs in time ${\myO}\left(m\vdot\log(m)\right)$, 
for a constant dimensional space $\reals^n$,
where $m$ is the number of affine line segments
in the polygonal traces.
The corresponding decision problem runs in linear time.

Our treatment in this paper is self-contained -- we do not assume any background
in computational geometry.
We derive, and prove where necessary,  missing results, and generalizations of the
results, that were sketched in~ \cite{AltG95} for $\reals^2$ and
$L_2$,
which are needed in our algorithm for computing the
Skorokhod distances in $\reals^n$ for the $L_1$, $L_2$, and $L_{\infty}$ metrics.
These steps require us to do a careful and detailed analysis of the
the~\cite{AltG95}  algorithm sketch, filling in missing details, and correcting  inaccuracies, in order to generalize to higher dimensions for the various norms.

%

\smallskip\noindent\textbf{Related Work.}
Metrics between traces are the basis for robustness and abstraction for hybrid systems
\cite{Broucke98,CaspiB02,TabuadaBook}.
More recently, distance metrics have been used to guide test generation \cite{STaliro} and
for conformance testing between different models
of a system \cite{FainekosMEMOCODE,Fainekos14}.

Metrics related to the Skorokhod metric have been studied in the context of timed systems
\cite{ChatterjeeP13, ChatterjeeIM14}; our results are orthogonal.
\emph{Dynamic Time Warping}~\cite{DTN, Berndt96} 
is a discrete \emph{sum} measure (it aggregate the discrepancies over the timeline), 
as opposed to the max-measure of
\frechet distances, between  discrete time sequences
and has been used heavily in signal processing and data mining.
Sum measures take the sum of the trace differences after retiming, as opposed to
considering the maximal (retimed) trace difference.
The discrete time  Dynamic Time Warping distance
be computed  using dynamic programming, and
has  approximation algorithms which are efficient~\cite{Salvador2007}.
The continuous analog of the Dynamic Time Warping sum measure between curves is
explored, and an algorithm presented for polygonal curves, 
in~\cite{Efrat_curvematching}.

\smallskip\noindent\textbf{Outline of the Paper.}
In Sections~\ref{section:SkorokhodSetting} and~\ref{section:FrechetSetting} we
introduce the Skorokhod and \frechet metrics, and show how the  computation of the
Skorokhod metric between two continuous traces in a normed space can be
reduced to the computation of the \frechet metric between two corresponding curves
in a related normed space.
In Section~\ref{section:FrechetMain} we present the general 
algorithms for polygonal curves in  $\reals^n$ (for  five  different norms)
  for the \frechet distance \emph{decision problem}; and for
 the computation of the \emph{value} of the  \frechet distance.
The value computation algorithm is parameterized by a set of geometric 
primitives over the underlying metric space of curves.
The computation procedures for these geometric primitives are
obtained in Section~\ref{section:Geometry} for the five norms.
Section~\ref{section:SkorokhodFinal} ties everything together for
the  Skorokhod distance computation problem for the norms $L_1, L_2$, and
$L_{\infty}$.
We conclude with a discussion of the paper in Section~\ref{section:Conclusion}.

\section{Skorokhod Distances between  Traces}
\label{section:SkorokhodSetting}

We begin by defining the Skorokhod distance between finite traces.
A (finite) \emph{trace} $x: [T_i,T_e] \rightarrow \myO $ is a  mapping from a 
finite time interval $[T_i,T_e]$,
with $0 \leq T_i < T_e$,
 to some observable  metric space  $\myO$ 
 with the associated metric $\dist_{\myO}$.
In this work we restrict our attention to traces which are continuous
with respect to time.
A \emph{retiming} $\retime: [T_i,T_e] \rightarrow  [T_i',T_e'] $  is an order preserving continuous 
bijective function from  
$[T_i,T_e]$ to $[T_i',T_e']$; 
thus if $t<t'$ then $\retime(t) < \retime(t')$.
Let the class of retiming functions from $[T_i,T_e]$ to $ [T_i',T_e'] $ 
be denoted as $\retimeclass_{[T_i,T_e] \rightarrow  [T_i',T_e'] }$, and let 
$\iden$ be the identity retiming.
Intuitively, retiming can be thought of as follows: imagine a stretchable and compressible 
rubber bar; a retiming of the bar gives a configuration of the bar where some parts have been
stretched, and some compressed, without the bar having been broken.
Given a trace  $x: [T_i^x,T_e^x] \rightarrow \myO $, and a retiming
$\retime: [T_i,T_e] \rightarrow  [T_i^x,T_e^x] $; the function
$x\circ \retime$ is another trace from $[T_i,T_e]$ to $\myO$.

\begin{definition}[Skorokhod Metric]
Given a retiming $\retime:  [T_i,T_e] \rightarrow  [T_i',T_e'] $, let $||\retime-\iden||_{\sup}$ be defined as 
\[
||\retime-\iden||_{\sup} =\sup_{t\in [T_i,T_e]}|\retime(t)-t|.\]
Given two traces $x: [T_i^x,T_e^x] \rightarrow \myO $ and $y: [T_i^y,T_e^y] \rightarrow \myO $, 
where $\myO$ is a metric space with the associated 
metric $ \dist_{\myO}$, 
 and a retiming $\retime: [T_i^x,T_e^x] \rightarrow  [T_i^y,T_e^y] $, let
$\norm{x\,-\, y\circ \retime}_{\sup}$  be defined as
\[
\norm{x\,-\, y\circ \retime}_{\sup}
=
\sup_{t\in [T_i^x,T_e^x]} \dist_{\myO}\big(\,  x(t)\ ,\ y\left(\retime(t)\right)\,  \big).\]
The \emph{Skorokhod distance} between the traces  $x()$ and $y()$  is defined to be:
\[
\skoro(x,y) = \inf_{r\in  \retimeclass_{[T_i^x,T_e^x] \rightarrow  [T_i^y,T_e^y]}} 
\max(\norm{\retime-\iden}_{\sup} \, ,\,  \norm{x\,-\, y\circ \retime}_{\sup}).\]
\end{definition}

Intuitively, the  Skorokhod distance 
incorporates two components: the first component quantifies
the {\em timing discrepancy} of the timing distortion required to ``match'' two traces, 
and the second quantifies the  \emph{value mismatch}  (in the metric space $\myO$) 
of the values under the timing distortion.
In the retimed trace $ y\circ \retime$, we see exactly the same values as in $y$, in 
exactly the same order, but the times at which the value are seen can be different.

\begin{remark}
 The two components
of the Skorokhod distance (the retiming, and the value difference components) can
be weighed with different weights -- this simply corresponds to a change of scale.
\end{remark}

\begin{example}[Retimed Traces]
We illustrate retimings and retimed traces in the next example.
Let $x$ be a trace $x: [0,1] \mapsto \reals^2$ defined by
$x(t) = \tuple{10\vdot t, \sqrt{t}}$.
Consider a retiming $r: [0,1] \mapsto [0,1]$ defined by:
\[
\retime(t) =
\begin{cases}
t^ 2& \text{ for } 0\leq t\leq 1/2\\
\frac{3}{2}\vdot t - \frac{1}{2} & \text{ for } 1/2\leq t \leq 1
\end{cases}
\]
The maximum timing distortion of $\retime$, \emph{i.e.}, $\sup_{t\in [0,1]} \abs{\retime(t)-t}$,
is $1/4$.
The retimed trace $x\left(\retime(t)\right)$ is given by
$x\left(\retime(t)\right)  = \tuple{x^{\retime}_1(t),\,  x^{\retime}_1(t)}$, where
 \[
x^{\retime}_1(t) =
\begin{cases}
10\vdot t^2 & \text{ for } 0\leq t\leq 1/2\\
10\vdot\left(\frac{3}{2}\vdot t - \frac{1}{2}\right) & \text{ for } 1/2\leq t \leq 1
\end{cases}
\]
and
\[
x^{\retime}_2(t) = 
\begin{cases}
t & \text{ for } 0\leq t\leq 1/2\\
\sqrt{\frac{3}{2}\vdot t - \frac{1}{2}} & \text{ for } 1/2\leq t \leq 1.
\end{cases}
\]
The original and the retimed traces are depicted in
Figures~\ref{figure:ExampleRetimingOrig} and~\ref{figure:ExampleRetimingRetimed}.
\qed
\end{example}

\begin{figure}[t]
\strut\centerline{\setlength{\unitlength}{0.00052493in}
\begingroup\makeatletter\ifx\SetFigFont\undefined%
\gdef\SetFigFont#1#2#3#4#5{%
  \reset@font\fontsize{#1}{#2pt}%
  \fontfamily{#3}\fontseries{#4}\fontshape{#5}%
  \selectfont}%
\fi\endgroup%
{\renewcommand{\dashlinestretch}{30}
\begin{picture}(7362,3108)(0,-10)
\path(465,3057)(465,357)(3165,357)
\path(465,357)(3165,3057)
\path(4650,3057)(4650,357)(7350,357)
\path(4650,357)(4651,360)(4652,365)
	(4655,376)(4659,392)(4665,414)
	(4672,442)(4681,475)(4691,514)
	(4703,557)(4715,602)(4727,649)
	(4740,696)(4753,743)(4766,789)
	(4778,833)(4790,875)(4802,915)
	(4813,952)(4824,988)(4835,1022)
	(4845,1053)(4855,1084)(4865,1112)
	(4875,1140)(4885,1167)(4895,1194)
	(4905,1219)(4915,1245)(4926,1271)
	(4937,1296)(4949,1322)(4961,1348)
	(4973,1374)(4986,1400)(4999,1427)
	(5013,1453)(5028,1480)(5043,1507)
	(5058,1534)(5075,1561)(5091,1587)
	(5108,1614)(5125,1640)(5143,1666)
	(5160,1691)(5178,1716)(5197,1740)
	(5215,1764)(5233,1788)(5252,1810)
	(5271,1832)(5289,1854)(5308,1875)
	(5328,1896)(5348,1917)(5366,1936)
	(5386,1955)(5405,1975)(5426,1994)
	(5447,2014)(5469,2033)(5491,2053)
	(5514,2073)(5538,2094)(5563,2114)
	(5588,2135)(5613,2155)(5639,2176)
	(5666,2196)(5692,2216)(5719,2236)
	(5746,2256)(5773,2275)(5800,2294)
	(5826,2313)(5853,2331)(5879,2349)
	(5905,2366)(5930,2383)(5955,2399)
	(5980,2414)(6004,2429)(6028,2444)
	(6052,2458)(6075,2472)(6100,2487)
	(6125,2501)(6150,2515)(6176,2529)
	(6201,2543)(6227,2557)(6254,2571)
	(6280,2585)(6307,2599)(6334,2613)
	(6362,2626)(6389,2640)(6416,2653)
	(6444,2666)(6471,2679)(6498,2692)
	(6524,2704)(6550,2716)(6576,2728)
	(6601,2739)(6625,2750)(6649,2760)
	(6672,2770)(6695,2780)(6717,2790)
	(6738,2799)(6759,2808)(6780,2817)
	(6804,2827)(6828,2837)(6852,2848)
	(6876,2858)(6901,2868)(6926,2879)
	(6953,2891)(6982,2903)(7011,2915)
	(7043,2928)(7076,2942)(7109,2956)
	(7144,2971)(7179,2985)(7213,3000)
	(7245,3013)(7273,3025)(7298,3035)
	(7318,3044)(7332,3050)(7342,3054)
	(7348,3056)(7350,3057)
\put(420,87){\makebox(0,0)[lb]{\smash{{\SetFigFont{9}{10.8}{\familydefault}{\mddefault}{\updefault}$0$}}}}
\put(1635,87){\makebox(0,0)[lb]{\smash{{\SetFigFont{9}{10.8}{\familydefault}{\mddefault}{\updefault}$t\longrightarrow$}}}}
\put(15,2922){\makebox(0,0)[lb]{\smash{{\SetFigFont{9}{10.8}{\familydefault}{\mddefault}{\updefault}$10$}}}}
\put(60,357){\makebox(0,0)[lb]{\smash{{\SetFigFont{9}{10.8}{\familydefault}{\mddefault}{\updefault}$0$}}}}
\put(2895,87){\makebox(0,0)[lb]{\smash{{\SetFigFont{9}{10.8}{\familydefault}{\mddefault}{\updefault}$1$}}}}
\put(1590,2292){\makebox(0,0)[lb]{\smash{{\SetFigFont{9}{10.8}{\familydefault}{\mddefault}{\updefault}$x_1(t)$}}}}
\put(7215,87){\makebox(0,0)[lb]{\smash{{\SetFigFont{9}{10.8}{\familydefault}{\mddefault}{\updefault}$1$}}}}
\put(4605,87){\makebox(0,0)[lb]{\smash{{\SetFigFont{9}{10.8}{\familydefault}{\mddefault}{\updefault}$0$}}}}
\put(5820,87){\makebox(0,0)[lb]{\smash{{\SetFigFont{9}{10.8}{\familydefault}{\mddefault}{\updefault}$t\longrightarrow$}}}}
\put(4245,357){\makebox(0,0)[lb]{\smash{{\SetFigFont{9}{10.8}{\familydefault}{\mddefault}{\updefault}$0$}}}}
\put(6270,2247){\makebox(0,0)[lb]{\smash{{\SetFigFont{9}{10.8}{\familydefault}{\mddefault}{\updefault}$x_2(t)$}}}}
\put(4200,2922){\makebox(0,0)[lb]{\smash{{\SetFigFont{9}{10.8}{\familydefault}{\mddefault}{\updefault}$1$}}}}
\end{picture}
}}
 \caption{Original trace $x(t) = \tuple{x_1(t), x_2(t)}$.}
\label{figure:ExampleRetimingOrig}
\end{figure}
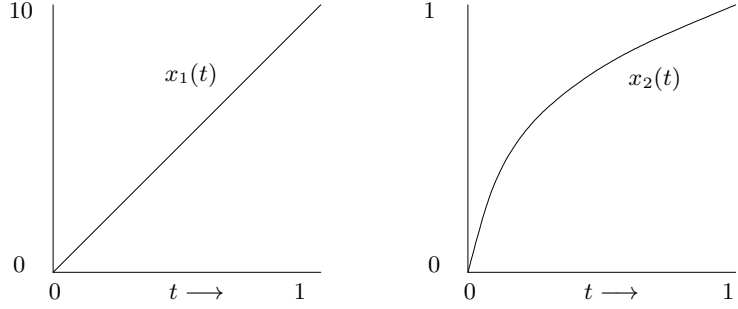

\begin{figure}[t]
\strut\centerline{\setlength{\unitlength}{0.00052493in}
\begingroup\makeatletter\ifx\SetFigFont\undefined%
\gdef\SetFigFont#1#2#3#4#5{%
  \reset@font\fontsize{#1}{#2pt}%
  \fontfamily{#3}\fontseries{#4}\fontshape{#5}%
  \selectfont}%
\fi\endgroup%
{\renewcommand{\dashlinestretch}{30}
\begin{picture}(7407,3108)(0,-10)
\path(465,3057)(465,357)(3165,357)
\path(1815,1032)(3165,3057)
\path(4650,357)(6000,1707)
\path(4650,3057)(4650,357)(7395,357)
\path(465,357)(468,358)(473,360)
	(484,364)(499,370)(520,378)
	(547,388)(577,400)(611,413)
	(646,427)(682,441)(719,455)
	(754,469)(788,482)(820,494)
	(850,506)(878,517)(905,528)
	(930,537)(954,547)(977,556)
	(999,565)(1021,573)(1043,582)
	(1066,591)(1089,601)(1113,610)
	(1137,620)(1162,630)(1186,641)
	(1211,651)(1236,662)(1262,673)
	(1287,684)(1312,695)(1336,705)
	(1360,716)(1383,727)(1406,738)
	(1427,748)(1448,759)(1468,769)
	(1486,779)(1504,788)(1521,798)
	(1538,807)(1557,818)(1575,830)
	(1593,842)(1611,855)(1629,868)
	(1647,883)(1667,899)(1687,916)
	(1708,934)(1729,953)(1750,971)
	(1769,989)(1785,1004)(1798,1016)
	(1807,1025)(1813,1030)(1815,1032)
\path(6000,1707)(6001,1710)(6003,1717)
	(6006,1729)(6010,1746)(6017,1769)
	(6024,1796)(6033,1827)(6042,1858)
	(6052,1890)(6062,1921)(6072,1951)
	(6083,1979)(6093,2006)(6104,2032)
	(6115,2056)(6126,2080)(6138,2103)
	(6151,2126)(6165,2149)(6176,2168)
	(6189,2187)(6201,2206)(6215,2226)
	(6229,2246)(6245,2267)(6261,2288)
	(6278,2310)(6296,2333)(6314,2356)
	(6334,2379)(6354,2402)(6375,2425)
	(6396,2448)(6417,2471)(6439,2494)
	(6461,2516)(6483,2538)(6505,2559)
	(6528,2579)(6550,2599)(6572,2619)
	(6594,2637)(6616,2655)(6638,2673)
	(6660,2689)(6682,2706)(6705,2722)
	(6728,2738)(6753,2754)(6778,2769)
	(6804,2785)(6833,2802)(6863,2818)
	(6894,2836)(6928,2853)(6964,2872)
	(7002,2891)(7041,2911)(7082,2930)
	(7123,2950)(7163,2969)(7201,2988)
	(7237,3004)(7268,3019)(7295,3032)
	(7316,3041)(7332,3049)(7342,3053)
	(7347,3056)(7350,3057)
\put(2895,87){\makebox(0,0)[lb]{\smash{{\SetFigFont{9}{10.8}{\familydefault}{\mddefault}{\updefault}$1$}}}}
\put(7215,87){\makebox(0,0)[lb]{\smash{{\SetFigFont{9}{10.8}{\familydefault}{\mddefault}{\updefault}$1$}}}}
\put(4605,87){\makebox(0,0)[lb]{\smash{{\SetFigFont{9}{10.8}{\familydefault}{\mddefault}{\updefault}$0$}}}}
\put(5820,87){\makebox(0,0)[lb]{\smash{{\SetFigFont{9}{10.8}{\familydefault}{\mddefault}{\updefault}$t\longrightarrow$}}}}
\put(4245,357){\makebox(0,0)[lb]{\smash{{\SetFigFont{9}{10.8}{\familydefault}{\mddefault}{\updefault}$0$}}}}
\put(4200,2922){\makebox(0,0)[lb]{\smash{{\SetFigFont{9}{10.8}{\familydefault}{\mddefault}{\updefault}$1$}}}}
\put(420,87){\makebox(0,0)[lb]{\smash{{\SetFigFont{9}{10.8}{\familydefault}{\mddefault}{\updefault}$0$}}}}
\put(1635,87){\makebox(0,0)[lb]{\smash{{\SetFigFont{9}{10.8}{\familydefault}{\mddefault}{\updefault}$t\longrightarrow$}}}}
\put(15,2922){\makebox(0,0)[lb]{\smash{{\SetFigFont{9}{10.8}{\familydefault}{\mddefault}{\updefault}$10$}}}}
\put(60,357){\makebox(0,0)[lb]{\smash{{\SetFigFont{9}{10.8}{\familydefault}{\mddefault}{\updefault}$0$}}}}
\put(1590,2292){\makebox(0,0)[lb]{\smash{{\SetFigFont{9}{10.8}{\familydefault}{\mddefault}{\updefault}$x^{\retime}_1(t)$}}}}
\put(6585,2292){\makebox(0,0)[lb]{\smash{{\SetFigFont{9}{10.8}{\familydefault}{\mddefault}{\updefault}$x^{\retime}_2(t)$}}}}
\end{picture}
}}
 \caption{Retimed trace $x\left(\retime(t)\right) = \tuple{x^{\retime}_1(t), x^{\retime}_2(t)}$.}
\label{figure:ExampleRetimingRetimed}
\end{figure}

\smallskip\noindent\textbf{Polygonal and Discrete-Time Traces.}
A polygonal trace $x: [T_i, T_e] \mapsto \myO$ where $\myO$ is a vector space 
with the scalar field  $\reals$  is a trace such that
there exists a finite sequence $T_i= t_0 < t_1 < \dots < t_m = T_e$ such that 
the trace segment between $t_k$ and $t_{k+1}$ is affine for all $0\leq k < m$,
\emph{i.e.} for $t_k \leq t \leq t_{k+1}$ we have 
$x(t) = x(t_k) + \frac{t- t_k}{t_{k+1}-t_k}\vdot \left(x( t_{k+1}) - x(t_k)\right)$.
Polygonal traces are obtained when discrete-time traces are completed by
linear interpolation.
A \emph{discrete-time}  trace $\dt{x_d}$  is a finite sequence 
$\tuple{x_0,t_0}, \, \tuple{x_1,t_1}, \, \dots  \tuple{x_m,t_m}$ such that
for all $0\leq k \leq m$, we have $x_k\in \myO$;  and that the $t_k$ sequence is a strictly
increasing sequence of time points with values in $\reals_+$.
If $\myO$ is a vector space with the scalar field  $\reals$ ,
the \emph{linear interpolation trace} $\LI(\dt{x_d}) ()$ of  the discrete time trace
$\dt{x_d}$  is the continuous time polygonal trace $x()$ over $[t_0, t_n]$ defined as
$x(t) =  x_k   +\lambda\vdot (x_{k+1} - x_k)$ for $t_k\leq t \leq t_{k+1}$ where
$t=  t_k   +\lambda\vdot (t_{k+1} - t_k)$ and $0\leq \lambda \leq 1$.

\smallskip\noindent\textbf{Continuous time 
Skorokhod Distance given Discrete-Time  Traces.}
Given two discrete-time traces $\dt{x_d}, \dt{y_d}$ the (continuous time)
Linear Interpolation
Skorokhod distance between them is defined to be the Skorokhod distance between 
the polygonal traces obtained by linear interpolation.
 \[
\skoro_{\LI}(\dt{x_d}, \dt{y_d}) = \skoro\left(\LI(\dt{x_d}) , \, \LI(\dt{y_d}) \right)\]
That is, we assume that given the discrete-time samples of the traces, the actual 
continuous time traces are polygonal traces where the affine segments can be obtained
by linear interpolation of the sampled values; the Skorokhod distance  is computed 
for these polygonal traces.

\begin{remark}
Let $y() =  \LI(\dt{y_d})$ be a linear interpolation trace.
We remark that after retiming, the retimed trace $y \circ \retime$ \emph{need not} 
be piecewise linear.
As an example, let $y: [0,1] \rightarrow [0,100]$ be a linear trace 
defined by $y(t)= 100\vdot t$.
Let the retiming  $\retime: [0,1] \rightarrow  [0,1]$ be defined by $\retime(t) = t^2$.
Then $y\circ \retime$ is the trace $z: [0,1] \rightarrow [0,100]$ where
$z(t) = 100 \vdot t^2$.
Figure~\ref{figure:retiming} illustrates four valid 
retimed traces from the  original trace $y()$ in the center.
\begin{figure}[h]
\strut\centerline{\setlength{\unitlength}{0.00043745in}
\begingroup\makeatletter\ifx\SetFigFont\undefined%
\gdef\SetFigFont#1#2#3#4#5{%
  \reset@font\fontsize{#1}{#2pt}%
  \fontfamily{#3}\fontseries{#4}\fontshape{#5}%
  \selectfont}%
\fi\endgroup%
{\renewcommand{\dashlinestretch}{30}
\begin{picture}(5247,5259)(0,-10)
\path(540.000,5007.000)(510.000,5127.000)(480.000,5007.000)
\path(510,5127)(510,492)(5235,492)
\path(5115.000,462.000)(5235.000,492.000)(5115.000,522.000)
\path(510,492)(1455,672)(2130,1032)
	(2940,1167)(3345,1662)(4020,1887)
	(4245,2292)(4470,2382)(5010,4992)
\thicklines
\path(510,492)(5010,4992)
\thinlines
\path(555,537)(1095,2337)(3795,3102)
	(4110,4767)(4965,4992)
\path(510,537)(510,540)(511,545)
	(513,556)(516,572)(519,593)
	(524,620)(529,651)(535,686)
	(541,722)(548,760)(554,797)
	(561,834)(568,869)(574,903)
	(580,935)(586,966)(592,995)
	(598,1022)(604,1049)(611,1075)
	(617,1101)(623,1126)(630,1152)
	(636,1176)(643,1200)(650,1225)
	(658,1250)(666,1276)(674,1302)
	(683,1329)(692,1356)(702,1383)
	(713,1411)(723,1439)(735,1467)
	(746,1494)(759,1521)(771,1548)
	(784,1574)(797,1599)(810,1624)
	(823,1647)(837,1669)(851,1691)
	(865,1711)(879,1730)(893,1748)
	(908,1766)(923,1782)(939,1799)
	(956,1815)(974,1830)(993,1844)
	(1012,1858)(1033,1872)(1054,1885)
	(1076,1897)(1098,1909)(1122,1921)
	(1146,1932)(1170,1942)(1194,1952)
	(1219,1961)(1244,1970)(1268,1979)
	(1293,1987)(1317,1995)(1341,2002)
	(1364,2010)(1387,2017)(1410,2023)
	(1433,2030)(1455,2037)(1477,2044)
	(1500,2051)(1523,2059)(1546,2067)
	(1569,2075)(1593,2084)(1617,2093)
	(1641,2103)(1665,2114)(1690,2125)
	(1714,2137)(1738,2149)(1762,2162)
	(1785,2176)(1808,2190)(1829,2204)
	(1850,2219)(1871,2235)(1890,2250)
	(1908,2267)(1925,2283)(1942,2301)
	(1958,2318)(1973,2337)(1987,2356)
	(2001,2377)(2014,2398)(2027,2420)
	(2040,2444)(2052,2468)(2064,2494)
	(2076,2520)(2088,2547)(2099,2575)
	(2110,2603)(2121,2632)(2131,2660)
	(2141,2689)(2150,2717)(2159,2745)
	(2168,2772)(2176,2798)(2184,2823)
	(2192,2848)(2199,2872)(2206,2894)
	(2213,2916)(2220,2937)(2228,2961)
	(2237,2985)(2245,3008)(2254,3030)
	(2264,3052)(2274,3074)(2284,3095)
	(2295,3116)(2306,3136)(2318,3155)
	(2331,3173)(2344,3191)(2357,3207)
	(2370,3223)(2384,3237)(2399,3251)
	(2413,3264)(2428,3275)(2444,3286)
	(2460,3297)(2477,3307)(2495,3317)
	(2514,3327)(2534,3336)(2555,3345)
	(2576,3355)(2599,3364)(2623,3372)
	(2647,3381)(2671,3389)(2695,3397)
	(2720,3405)(2743,3412)(2767,3419)
	(2789,3426)(2811,3432)(2831,3438)
	(2851,3444)(2870,3449)(2888,3454)
	(2909,3461)(2930,3467)(2950,3474)
	(2970,3481)(2989,3488)(3008,3496)
	(3026,3504)(3044,3512)(3061,3521)
	(3078,3531)(3093,3541)(3107,3551)
	(3121,3562)(3134,3573)(3146,3585)
	(3158,3597)(3168,3609)(3178,3622)
	(3188,3635)(3199,3650)(3209,3666)
	(3219,3682)(3230,3700)(3240,3718)
	(3250,3737)(3259,3756)(3268,3776)
	(3277,3795)(3285,3814)(3292,3833)
	(3299,3852)(3305,3870)(3310,3887)
	(3315,3904)(3320,3924)(3324,3943)
	(3328,3963)(3332,3983)(3336,4003)
	(3340,4024)(3343,4045)(3347,4065)
	(3351,4086)(3355,4105)(3360,4124)
	(3365,4142)(3370,4159)(3376,4175)
	(3383,4190)(3390,4204)(3398,4218)
	(3408,4232)(3418,4246)(3430,4259)
	(3443,4272)(3458,4285)(3474,4298)
	(3491,4310)(3510,4321)(3529,4332)
	(3548,4341)(3569,4350)(3589,4358)
	(3610,4365)(3631,4372)(3653,4377)
	(3670,4381)(3689,4385)(3708,4388)
	(3728,4391)(3749,4394)(3770,4397)
	(3792,4400)(3814,4403)(3837,4405)
	(3859,4408)(3882,4411)(3904,4414)
	(3925,4417)(3946,4420)(3966,4423)
	(3985,4427)(4003,4431)(4019,4435)
	(4035,4440)(4050,4444)(4067,4451)
	(4084,4459)(4100,4468)(4115,4478)
	(4129,4488)(4143,4500)(4157,4512)
	(4169,4526)(4181,4540)(4192,4554)
	(4202,4568)(4212,4583)(4221,4597)
	(4229,4611)(4237,4625)(4245,4639)
	(4253,4654)(4261,4668)(4269,4683)
	(4279,4699)(4289,4715)(4300,4731)
	(4312,4748)(4325,4764)(4339,4780)
	(4355,4795)(4371,4810)(4389,4824)
	(4407,4838)(4427,4850)(4448,4861)
	(4470,4872)(4486,4879)(4503,4885)
	(4521,4891)(4541,4898)(4562,4904)
	(4586,4910)(4611,4916)(4640,4922)
	(4670,4928)(4703,4934)(4739,4941)
	(4776,4948)(4815,4954)(4855,4961)
	(4894,4967)(4930,4973)(4964,4978)
	(4993,4983)(5017,4986)(5034,4989)
	(5046,4991)(5052,4992)(5055,4992)
\path(510,492)(510,494)(511,499)
	(513,508)(515,521)(519,540)
	(524,565)(530,595)(537,631)
	(546,671)(557,715)(568,761)
	(581,809)(595,859)(611,908)
	(627,956)(644,1004)(662,1049)
	(681,1092)(700,1133)(721,1172)
	(743,1208)(766,1242)(791,1274)
	(816,1304)(844,1331)(873,1357)
	(904,1382)(937,1405)(972,1427)
	(1010,1447)(1050,1467)(1081,1481)
	(1113,1494)(1146,1508)(1180,1521)
	(1216,1533)(1253,1546)(1292,1559)
	(1332,1571)(1373,1584)(1415,1597)
	(1459,1609)(1503,1622)(1549,1635)
	(1596,1649)(1644,1662)(1692,1676)
	(1742,1690)(1791,1704)(1842,1719)
	(1893,1734)(1944,1749)(1995,1765)
	(2046,1781)(2097,1798)(2148,1815)
	(2199,1832)(2248,1850)(2298,1869)
	(2346,1888)(2394,1907)(2441,1927)
	(2487,1948)(2531,1969)(2575,1991)
	(2617,2013)(2658,2036)(2698,2060)
	(2737,2084)(2774,2109)(2810,2135)
	(2844,2162)(2877,2189)(2909,2218)
	(2940,2247)(2969,2278)(2998,2309)
	(3025,2342)(3051,2376)(3076,2412)
	(3101,2448)(3125,2486)(3147,2525)
	(3170,2566)(3191,2607)(3212,2650)
	(3232,2693)(3252,2738)(3271,2784)
	(3290,2831)(3308,2878)(3325,2927)
	(3342,2976)(3359,3025)(3375,3075)
	(3391,3125)(3407,3176)(3422,3226)
	(3437,3277)(3452,3327)(3466,3377)
	(3480,3427)(3494,3476)(3508,3524)
	(3522,3572)(3536,3619)(3549,3665)
	(3563,3710)(3577,3755)(3591,3798)
	(3605,3840)(3619,3881)(3633,3921)
	(3648,3960)(3663,3997)(3678,4034)
	(3694,4069)(3710,4104)(3728,4137)
	(3748,4174)(3770,4210)(3792,4245)
	(3816,4279)(3842,4311)(3869,4343)
	(3898,4373)(3928,4403)(3961,4432)
	(3995,4461)(4032,4490)(4071,4518)
	(4113,4546)(4157,4574)(4203,4602)
	(4252,4631)(4303,4659)(4356,4686)
	(4411,4714)(4467,4742)(4525,4769)
	(4582,4795)(4639,4821)(4695,4846)
	(4749,4869)(4800,4890)(4848,4910)
	(4892,4928)(4930,4944)(4963,4957)
	(4991,4968)(5013,4976)(5030,4983)
	(5042,4987)(5049,4990)(5053,4991)(5055,4992)
\put(2310,87){\makebox(0,0)[lb]{\smash{{\SetFigFont{9}{10.8}{\familydefault}{\mddefault}{\updefault}$t$}}}}
\put(105,1122){\makebox(0,0)[lb]{\smash{{\SetFigFont{9}{10.8}{\familydefault}{\mddefault}{\updefault}$y$}}}}
\put(195,447){\makebox(0,0)[lb]{\smash{{\SetFigFont{9}{10.8}{\familydefault}{\mddefault}{\updefault}$0$}}}}
\put(4965,132){\makebox(0,0)[lb]{\smash{{\SetFigFont{9}{10.8}{\familydefault}{\mddefault}{\updefault}$1$}}}}
\put(15,5037){\makebox(0,0)[lb]{\smash{{\SetFigFont{9}{10.8}{\familydefault}{\mddefault}{\updefault}$100$}}}}
\put(375,132){\makebox(0,0)[lb]{\smash{{\SetFigFont{9}{10.8}{\familydefault}{\mddefault}{\updefault}$0$}}}}
\end{picture}
}}
 \caption{Retimed traces from $y(t) = 100\vdot t$.}
\label{figure:retiming}
\end{figure}
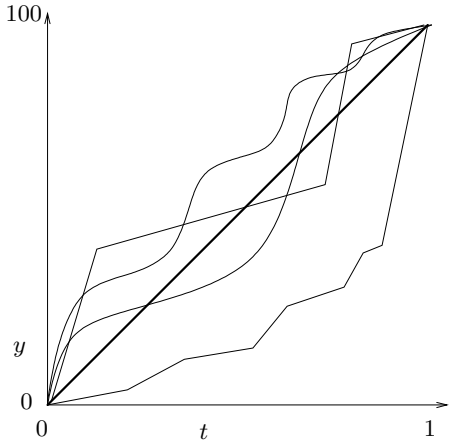
\end{remark}

\section{From Skorokhod to \frechet Distances}
\label{section:FrechetSetting}
We apply the work on \frechet distances~\cite{AltG95} towards computing
Skorokhod distances between polygonal traces.
The work of~\cite{AltG95} deals with polygonal  \emph{curves}. 

\begin{definition}[Curves]
A \emph{curve} is a continuous mapping $f: [a,b] \rightarrow \myV$ where $a,b\in \reals_+$ with
$a<b$, and
${\myV}$ is a vector space with the scalar field $\reals$ .
A \emph{polygonal curve} is a curve $P: [0, m] \rightarrow {\myV}$ with  $m\in \nat$ , such that for
$0\leq i < m$ the segment of $P$ over $[i, i+1] $ is obtained by linear interpolation, i.e.,
$P(i+\lambda) =  (1-\lambda)\vdot P(i)   +\lambda\vdot P( i+1)$ for $0\leq \lambda \leq 1$.

\end{definition}

\noindent Given two points $x.x'$, we denote the affine segment between the two points as
$\myline(x,x')$.

\begin{definition}[\frechet distance]
\label{def:frechet}
Let $f:[a_f, b_f] \rightarrow {\myV}$ and $g:[a_g, b_g] \rightarrow {\myV}$ be curves. 
The \frechet distance between the two curves is defined to be
\[
\fre(f,g) = \inf_{\substack{
\alpha_f: [0,1] \rightarrow [a_f, b_f] \\
\alpha_g: [0,1] \rightarrow [a_g, b_g] }}\quad
\max_{0\leq \theta \leq 1} \norm {f\left(\alpha_f(\theta)\right) - g\left(\alpha_g(\theta)\right)}
\]
 where $\lVert \, \rVert$ is the norm  over ${\myV}$, and where $\alpha_f, \alpha_g$ range over  continuous and strictly increasing bijective functions 
onto $ [a_f, b_f] $ and  $[a_g, b_g] $ respectively.
\end{definition}

Intuitively, the reparameterizations $\alpha_f, \alpha_g$ control the ``speed'' of traversal along the
two curves by two entities. 
The positions of the two entities in the two curves at ``time'' $\theta$  is given by
$\alpha_f(\theta)$ and $\alpha_g(\theta)$ respectively; with the value of the curves
at those positions being
$ f\left(\alpha_f(\theta)\right) $ and  $g\left(\alpha_g(\theta)\right) $ respectively.
The two entities always have a strictly greater than $0$ speed.

Observe that unlike the Skorokhod distance, the \frechet distance does not impose
a penalty for the amount of position mismatch, in particular, almost the whole curve $g$ can be matched to a tiny portion of $f$, with no penalty.
We show that, still,  the Skorokhod distance between continuous traces
can be obtained using the \frechet  distance between two derived curves.
First we obtain a curve given a continuous time trace.

\begin{definition}[Curve of a Trace]
\label{def:CurveFromTrace}
Let $\myO$ be a vector space with the scalar field $\reals$.
Let $x: [T_i^x, T_e^x] \rightarrow \myO$ be 
a continuous trace.
The \emph{curve} of the trace $x()$ is the curve 
$\curve_x: [T_i^x, T_e^x] \mapsto \myO\times [T_i^x, T_e^x]$ 
defined by the following parameterization.
For $\rho \in [T_i^x, T_e^x] $ we define
(a)~$\curve_x^{\myO}(\rho) = x(\rho)$; and
(b)~$\curve_x^\mytime(\rho) = \rho$.
The curve $\curve_x$ is given by $\tuple{\curve_x^{\myO}(\rho), \curve_x^{\mytime}(\rho)} $ 
over $[T_i^x, T_e^x]$ with values in  $\myO\times \reals$.
We defined the norm of the space $\myO\times \reals$ as 
$\lVert \tuple{o,t} \rVert = \max(\lVert o\rVert, |t|)$.
We refer to this norm as the $L_{\myO}^{\skoro}$ norm, where
$L_{\myO}$ is the norm of  $\myO$,  and to the space 
$\myO\times \reals$  as the $\myO^{\skoro}$ space.
\end{definition}

We note that $\tuple{\curve_x^{\myO}(\rho), \curve_x^{\mytime}(\rho)} $ is a continuous mapping from
$[T_i^x, T_e^x]$ to  $\myO\times [T_i^x, T_e^x]$ since $x()$ is a continuous trace.
The next lemma shows that  $L_{\myO}^{\skoro}$ is a norm.
Thus, the previous definition is sound.

\begin{lemma}
Let $\myO$ be a vector space with the scalar field $\reals$.
Consider the space  $\myO\times \reals$. 
The function
$L_{\myO}^{\skoro} : \myO\times \reals \mapsto \reals_+$ defined by
 $L_{\myO}^{\skoro}( \tuple{p,t} )= \max(\norm{p}_{\myO}, |t|)$ is a norm on $\myO\times \reals$.
\end{lemma}
\begin{proof}
It can be easily seen that 
$L_{\myO}^{\skoro}( \tuple{a\vdot p, a\vdot t}) =
 \abs{a} \vdot L_{\myO}^{\skoro}( \tuple{\vdot p,t} $ for $a\in \reals$, and
$L_{\myO}^{\skoro}( \tuple{p,t} ) = 0$ iff $p=\vec{0}$ and $t=0$, where $\vec{0}$ is the zero vector 
of ${\myO}$.
We now check for the triangle inequality.
\begin{align*}
L_{\myO}^{\skoro}\left( \tuple{p,t_p} \right) +   L_{\myO}^{\skoro}\left( \tuple{q,t_q}\right) 
& = \max\left(\norm{p}_{\myO}, \abs{t_p}\right) +  \max\left(\norm{q}_{\myO}, \abs{t_q}\right) \\
&\geq
\begin{array}{l} 
 \norm{p}_{\myO} + \norm{q}_{\myO}; \text{ and}\\
 \abs{t_p} + \abs{t_q} 
\end{array}
\text{ both.}\\
\text{Thus, }
L_{\myO}^{\skoro}\left( \tuple{p,t_p} \right) +   L_{\myO}^{\skoro}\left( \tuple{q,t_q}\right)
& \geq \max\left( 
\norm{p}_{\myO} + \norm{q}_{\myO}\, ,\, 
 \abs{t_p} +  \abs{t_q}
\right)  \\
& \geq\max\left(
 \norm{p+q}_{\myO},  \abs{t_p + t_q}
\right)\\
&= L_{\myO}^{\skoro}\left( \tuple{p+q,t_p+t_q} \right)\\
& = L_{\myO}^{\skoro}\left( \tuple{p,t_p} +  \tuple{q,t_q} \right)
\end{align*}
Thus the triangle inequality holds and this concludes the proof.
\end{proof}

The next result shows that the
Skorokhod distance betwee two continuous traces is equal to the
\frechet distance between two related curves in a corresponding normed space.
\begin{proposition}[From Skorokhod to \frechet]
\label{proposition:SkoroToFrechet}
Let ${\myO}$ be a vector space with the scalar field $\reals$.
Let $x: [T_i^x, T_e^x] \rightarrow {\myO}$ and $y: [T_i^y, T_e^y] \rightarrow {\myO}$ be
two  continuous traces.
We have
\[
\skoro(x,y) =\fre(\curve_x, \curve_y)
 \]
where the curves $\curve_x$ and $ \curve_y$ are the curves in the space
${\myO}^{\skoro}$ corresponding to
the traces $x$ and $y$  as defined in Definition~\ref{def:CurveFromTrace}.
\end{proposition}
\begin{proof}
We prove $\skoro(x,y)  \leq \fre(\curve_x, \curve_y)$, and that
$\fre(\curve_x, \curve_y) \leq \skoro(x,y)$.
Let $f=\curve_x$ and $g=C_y$.
\begin{enumerate}
\item $\skoro(x,y)  \leq \fre(\curve_x, \curve_y)$.\
Let $\alpha_f, \alpha_g$ be the reparametrizations as in Definition~\ref{def:frechet}.
Consider a retiming $ \retime : [T_i^x, T_e^x] \rightarrow [T_i^y, T_e^y]$ defined as
$\retime(t) = \alpha_g\left(\alpha^{-1}_f\left(t\right)\right)$.
It can be checked that the function as defined is a valid retiming.
Now,
\begin{align*}
 \max(\lVert \retime-\iden \rVert_{\sup} \, ,\,  \lVert x\,-\, y\circ \retime\rVert_{\sup})  & =\  
\max(\lVert\retime-\iden\rVert_{\sup} \,  , \, \max_{ T_i^x \leq t \leq T_e^x} 
\lVert x(t) - y\left(\alpha_g\left(\alpha^{-1}_f\left(t\right)\right)\right) \rVert )\\
& =\ 
\max\left(
\lVert\retime-\iden\rVert_{\sup} \,  ,  \, 
\max_{ 0 \leq \theta \leq 1} 
\left\lVert 
x\left(\alpha_f(\theta)\right) - y\left(\alpha_g\left(\theta\right)\right) 
\right\rVert
\right)\\
& \text{by setting }t=\alpha_f(\theta)  \text{in the last equation.}\\
& = 
\max\left(
\max_{T_i^x \leq t \leq T_e^x}|\alpha_g\left(\alpha^{-1}_f\left(t\right)\right)-t| \,  ,  \, 
\max_{ 0 \leq \theta \leq 1} 
\left\lVert 
x\left(\alpha_f(\theta)\right) - y\left(\alpha_g\left(\theta\right)\right) 
\right\rVert
\right)\\
& = 
\max\left(
\max_{ 0 \leq \theta \leq 1} |\alpha_g(\theta) - \alpha_f(\theta)|\, , \, 
\max_{ 0 \leq \theta \leq 1} 
\left\lVert 
x\left(\alpha_f(\theta)\right) - y\left(\alpha_g\left(\theta\right)\right) 
\right\rVert
\right)\\
& =
\max_{ 0 \leq \theta \leq 1}
\max\left(|\alpha_g(\theta) - \alpha_f(\theta)|\, , \, 
\left\lVert 
x\left(\alpha_f(\theta)\right) - y\left(\alpha_g\left(\theta\right)\right) 
\right\rVert
\right)\\
& \text{by interchanging the order of taking maximums}\\
&\text{ in the previous equation.}\\
& =
\max_{ 0 \leq \theta \leq 1} \lVert \tuple{
x\left(\alpha_f(\theta)\right) - y\left(\alpha_g\left(\theta\right)\right) \, ,\, 
\alpha_f(\theta) - \alpha_g(\theta)
}
\rVert
\\
&=
\max_{0\leq \theta \leq 1} \left\lVert f\left(\alpha_f(\theta)\right) - g\left(\alpha_g(\theta)\right) \right\rVert
\end{align*}
Thus, for every valid choice of $\alpha_f, \alpha_g$, there is a valid retiming
function $\retime$ such that 
$ \max(\lVert \retime-\iden \rVert_{\sup} \, ,\,  \lVert x\,-\, y\circ \retime\rVert_{\sup}) 
\ = \ \max_{0\leq \theta \leq 1} \left\lVert f\left(\alpha_f(\theta)\right) - g\left(\alpha_g(\theta)\right) \right\rVert$.
Hence $\skoro(x,y)  \leq \fre(\curve_x, \curve_y)$.

\item $\skoro(x,y)  \geq \fre(\curve_x, \curve_y)$.\\
Let $ \retime : [T_i^x, T_e^x] \rightarrow [T_i^y, T_e^y]$  be a retiming.
Let $\alpha_f: [0,1] \rightarrow  [T_i^x, T_e^x] $ be defined as
$\alpha_f(\theta) = (1-\theta)\vdot T_i^x + \theta\vdot  T_e^x$.
Let $\alpha_g: [0,1] \rightarrow  [T_i^y, T_e^y] $ be defined as
$\alpha_g(\theta) = 
\retime\big(\left(1-\theta\right)\vdot T_i^x + \theta\vdot  T_e^x\big)$.
Thus, $\alpha_g(\theta) =  \retime(\alpha_f(\theta))$.
Observe that $\alpha_f, \alpha_g$ satisfy the conditions of  Definition~\ref{def:frechet}.
We have,
\begin{align*}
\max_{0\leq \theta \leq 1} \left\lVert f\left(\alpha_f(\theta)\right) - g\left(\alpha_g(\theta)\right) \right\rVert & =
\max_{ 0 \leq \theta \leq 1}
\max\big(|\alpha_f(\theta) - \alpha_g(\theta)|\, , \, 
\left\lVert 
x\left(\alpha_f(\theta)\right) - y\left(\alpha_g\left(\theta\right)\right) 
\right\rVert
\big)\\
& =
\max\left(
\max_{ 0 \leq \theta \leq 1} |\alpha_f(\theta) - \alpha_g(\theta)|\, , \, 
\max_{ 0 \leq \theta \leq 1} 
\left\lVert 
x\left(\alpha_f(\theta)\right) - y\left(\alpha_g\left(\theta\right)\right) 
\right\rVert
\right)\\
& \text{by  interchanging the order of taking maximums}\\
&\text{ in the previous equation.}\\
& =
\max\left(
\max_{T_i^x \leq t \leq T_e^x} |t- r(t) | \, , \, 
\max_{T_i^x \leq t \leq T_e^x}\left\lVert x(t) - y\left(r\left(t\right)\right) \right\rVert
\right)\\
& \text{by setting }t=\alpha_f(\theta)  \text{in the last equation.}\\
& = \max(\lVert \iden - \retime\rVert_{\sup} \, ,\,  \lVert x\,-\, y\circ \retime\rVert_{\sup})
\end{align*}
Thus, for every valid retiming function $\retime$, there 
is a valid choice of $\alpha_f, \alpha_g$ such that
$ \max(\lVert \retime-\iden \rVert_{\sup} \, ,\,  \lVert x\,-\, y\circ \retime\rVert_{\sup}) 
\ = \ \max_{0\leq \theta \leq 1} \left\lVert f\left(\alpha_f(\theta)\right) - g\left(\alpha_g(\theta)\right) \right\rVert$.
Hence $\fre(\curve_x, \curve_y) \leq \skoro(x,y)$.
\end{enumerate}
\end{proof}

\section{Computation of the \frechet Distance}
\label{section:FrechetMain}

In this section we  explore in detail   the algorithm sketch of~\cite{AltG95} 
for computing \frechet distances in $\reals^2$ for the $L_2$ norm 
(see also~\cite{Knauer:PhD} and~\cite{Wenk:PhD}).
We derive, and prove where necessary,  missing results, and generalizations of the
results
which are needed in our algorithm for computing the
Skorokhod distances in higher dimensions  for the $L_1$, $L_2$, and $L_{\infty}$ metrics.
We first solve for the \emph{decision problem} in $\reals^n$ 
for the different norms $L_1, L_2$ and $L_{\infty}$, and also for
$L_1^{\skoro}, L_2^{\skoro}$ and $L_{\infty}^{\skoro}$ 
(the three new norms are required for solving the Skorokhod distance problem in the 
three standard norms).
We then solve for the \emph{value computation} problem by 
obtaining  a geometric characterization for quantities which are necessary for
computing the value of the \frechet distance.
Using this geometric characterization, the desired quantities required in the
\frechet distance  are computed by solving
geometric problems for the various norms in Section~\ref{section:Geometry}.

\subsection{The Free Space}
This subsection present the notion of \emph{Free Space} from~\cite{AltG95}
and its characterization for $L_1, L_2$, $L_{\infty}$
and also for  $L_1^{\skoro}, L_2^{\skoro}$ and $L_{\infty}^{\skoro}$  norms in the $\reals^n$  space.
The Free Space concept allows us to reduce reasoning about the \frechet 
distance in any dimension to
reasoning about paths in $\reals^2$.

Given a polygonal curve $f: [0, n] \rightarrow {\myV}$, we let
$f_{[i]}$ denote the curve segment between $[i, i+1]$.
Thus, $f_{[i]}: [i, i+1] \rightarrow  {\myV}$ and
$f(\rho) =  f_{[i]}(\rho)$ for $i\leq \rho\leq i+1$.

\begin{definition}[Free Space~\cite{AltG95}]
Given polygonal 
curves $f:[0, m_f] \rightarrow \myV$ and $g:[0, m_g] \rightarrow {\myV}$, and a
real number $\delta \geq 0$, \emph{$\delta$-Free Space} of $f,g$ is the set
\[
\free_{\delta}(f,g) = \set{(\rho_f, \rho_g) \in [0, m_f]\times [0,m_g]\,  \mid
\text{we have } \left\lVert f(\rho_f) - g(\rho_g)\right\rVert_{\myV} \leq \delta}
\]
\end{definition}
The tuples $(\rho_f, \rho_g) $ belonging to $\free_{\delta}(f,g) $
denote the positions in the two curves such that
the difference in the values of the two curves is less than $\delta$.
A pictorial representation of the free space is referred to as the 
\emph{free space diagram}.
The space $[0, m_f]\times [0,m_g]$ can be viewed as consisting of 
$m_f\vdot m_g$  \emph{cells}, with  cell $i,j$ being $[i,i+1]\times [j,j+1]$
for $0\leq i<m_f$, and $0\leq j<m_g$,
Observe that $\free_{\delta}(f,g) $ intersected with cell $i,j$ is just the free space
corresponding to the curve segments  $f_{[i]}, g_{[j]}$;
\emph{i.e} the intersection of the cell $i,j$ with $\free_{\delta}(f,g) $ is equal to
 $\free_{\delta}(f_{[i]}, g_{[j]})$.

\begin{proposition}[\cite{AltG95}]
\label{proposition:FreeSpaceCurve}
Given two polygonal curves $f,g$, we have $\fre(f,g) \leq \delta$
if there is a curve $\alpha: [0,1] \rightarrow  [0, m_f]\times [0,m_g]$
 in $\free_{\delta}(f,g) $ from $(0,0)$ to $(m_f, m_g)$ which is
strictly increasing in both coordinates
\footnote{The corresponding result (stated without proof)  in~\cite{AltG95} 
also includes the  ``only if'' direction.
However,~\cite{AltG95} only requires
\emph{non-decreasing} reparametrizations, as opposed to our formulation
which requires strictly increasing reparametrizations necessary for utilizing
\frechet distances in computing Skorokhod distances.
A careful analysis shows that the ``only if'' direction of the propositon also holds in case of
non-decreasing  reparametrizations, but not for strictly
increasing reparametrizations.
We address this issue in detail later in the paper.}.
\end{proposition}
The curve $\alpha$ can be thought of as the parameterized curve $(\alpha_f, \alpha_f)$, 
with
$\alpha_f:  [0,1] \rightarrow  [0, m_f]$, $\alpha_g:  [0,1] \rightarrow  [0, m_g] $.
These functions $\alpha_f, \alpha_g$ can be viewed as  the reparametrization functions
in Definition~\ref{def:frechet}.
An example of the free space for two polygonal curves is given in
Figure~\ref{figure:FreeSpace}.
The unshaded portion is the free space.
The figure also includes a continuous curve which is 
strictly increasing in both coordinates, from $(0,0)$ to $(m_f, m_g)$.
We now analyze the properties of $\free_{\delta}(f_{[i]}, g_{[j]})$.

\begin{figure}[h]
\strut\centerline{\setlength{\unitlength}{0.00034996in}
\begingroup\makeatletter\ifx\SetFigFont\undefined%
\gdef\SetFigFont#1#2#3#4#5{%
  \reset@font\fontsize{#1}{#2pt}%
  \fontfamily{#3}\fontseries{#4}\fontshape{#5}%
  \selectfont}%
\fi\endgroup%
{\renewcommand{\dashlinestretch}{30}
\begin{picture}(8148,6012)(0,-10)
\path(915,4164)(8115,4164)
\path(2715,5964)(2715,564)
\path(4515,5964)(4515,564)
\path(6315,5964)(6315,564)
\texture{55888888 88555555 5522a222 a2555555 55888888 88555555 552a2a2a 2a555555 
	55888888 88555555 55a222a2 22555555 55888888 88555555 552a2a2a 2a555555 
	55888888 88555555 5522a222 a2555555 55888888 88555555 552a2a2a 2a555555 
	55888888 88555555 55a222a2 22555555 55888888 88555555 552a2a2a 2a555555 }
\shade\path(915,5964)(8115,5964)(8115,564)
	(915,564)(915,5964)
\path(915,5964)(8115,5964)(8115,564)
	(915,564)(915,5964)
\path(915,2364)(8115,2364)
\whiten\path(915,1464)(1365,1914)(2265,2364)
	(2715,1914)(2715,1239)(2265,789)
	(1590,564)(915,564)(915,1464)
\path(915,1464)(1365,1914)(2265,2364)
	(2715,1914)(2715,1239)(2265,789)
	(1590,564)(915,564)(915,1464)
\whiten\path(2715,1914)(3840,2364)(4290,2364)
	(4515,1914)(4515,1014)(3390,564)
	(2715,1239)(2715,1914)
\path(2715,1914)(3840,2364)(4290,2364)
	(4515,1914)(4515,1014)(3390,564)
	(2715,1239)(2715,1914)
\whiten\path(4515,1914)(4965,2139)(6090,1464)
	(5415,789)(4515,1014)(4515,1914)
\path(4515,1914)(4965,2139)(6090,1464)
	(5415,789)(4515,1014)(4515,1914)
\whiten\path(4515,3714)(5415,4164)(6315,4164)
	(6315,3039)(5640,2589)(4515,2814)(4515,3714)
\path(4515,3714)(5415,4164)(6315,4164)
	(6315,3039)(5640,2589)(4515,2814)(4515,3714)
\whiten\path(6315,3039)(7215,2589)(8115,3264)
	(7440,4164)(6315,4164)(6315,3039)
\path(6315,3039)(7215,2589)(8115,3264)
	(7440,4164)(6315,4164)(6315,3039)
\whiten\path(2715,3714)(2715,3264)(3840,2364)
	(4290,2364)(4515,2814)(4515,3714)
	(3840,4164)(3165,4164)(2715,3714)
	(2715,3264)(2715,3714)
\path(2715,3714)(2715,3264)(3840,2364)
	(4290,2364)(4515,2814)(4515,3714)
	(3840,4164)(3165,4164)(2715,3714)
	(2715,3264)(2715,3714)
\whiten\path(6315,4839)(7215,5964)(8115,5964)
	(8115,5064)(7440,4164)(6315,4164)(6315,4839)
\path(6315,4839)(7215,5964)(8115,5964)
	(8115,5064)(7440,4164)(6315,4164)(6315,4839)
\whiten\path(5415,4164)(4515,4839)(4515,5964)
	(5865,5964)(6315,4839)(6315,4164)(5415,4164)
\path(5415,4164)(4515,4839)(4515,5964)
	(5865,5964)(6315,4839)(6315,4164)(5415,4164)
\whiten\path(3165,4164)(3615,5964)(4515,5964)
	(4515,4839)(3840,4164)(3165,4164)
\path(3165,4164)(3615,5964)(4515,5964)
	(4515,4839)(3840,4164)(3165,4164)
\thicklines
\path(915,564)(2715,1464)(3840,1689)
	(4065,2364)(4515,3039)(5640,4164)
	(6315,4389)(7440,4614)(8115,5964)
\thinlines
\path(2715,564)(2715,5964)
\path(6315,5964)(6315,564)
\path(915,4164)(8115,4164)
\path(4515,564)(4515,5964)
\put(1140,114){\makebox(0,0)[lb]{\smash{{\SetFigFont{9}{10.8}{\familydefault}{\mddefault}{\updefault}$\rho_f$}}}}
\put(15,789){\makebox(0,0)[lb]{\smash{{\SetFigFont{9}{10.8}{\familydefault}{\mddefault}{\updefault}$\rho_g$}}}}
\end{picture}
}}
 \caption{The Free Space $\free_{\delta}(f,g)$.}
\label{figure:FreeSpace}
\end{figure}
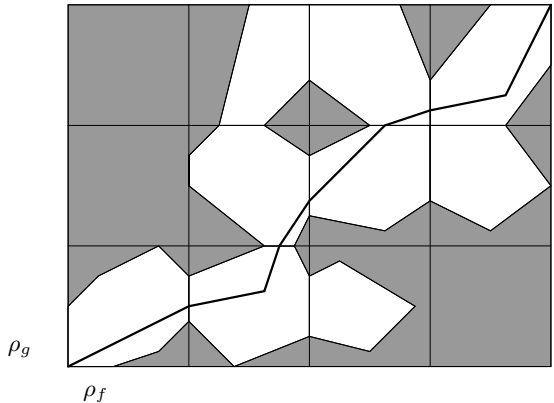

\begin{definition}
For $(\rho_f, \rho_g) \in [ i, i+1] \times  [j, j+1]$, we denote
$\rho_f^* = \rho_f-i$, and $\rho_g^* = \rho_g-j$, \emph{i.e.},
we move the origin to $(i,j)$.\qed
\end{definition}

\begin{proposition}[Free Space for Affine Line Segments in $\reals^1$]
\label{prop:FreeSpaceLines}
Let $\mytime_x: [i, i+1] \rightarrow  \reals$ and 
$\mytime_y: [j, j+1] \rightarrow  \reals$ be affine line segments. 
The set $ \free_{\delta}(\mytime_x, \mytime_y)$ for the absolute value norm  is
a portion of
$[i,i+1]\times [j,j+1] $ which lies in the intersection of two half-planes, the boundaries of
which are parallel.
\begin{gather}
\mytime_x(i) -  \mytime_y(j) + \left(\mytime_x(i+1) -  \mytime_x(i) \right) \vdot \rho_x^* +
 \left(\mytime_y(j) -  \mytime_y(j+1) \right)  \vdot \rho_y^*\leq \delta\\
\mytime_x(i) -  \mytime_y(j) + \left(\mytime_x(i+1) -  \mytime_x(i) \right) \vdot \rho_x ^*+
 \left(\mytime_y(j) -  \mytime_y(j+1) \right)    \vdot \rho_y^* \geq -\delta
\end{gather}
\end{proposition}
\begin{proof}
We have 
$\mytime_x(i+\rho_x) = (1-\rho_x)\vdot \mytime_x(i) + \rho_x\vdot \mytime_x(i+1)$
for $0\leq \rho_x\leq 1$;
and similarly
$\mytime_y(j+\rho_y) = (1-\rho_y)\vdot \mytime_y(j) + \rho_y\vdot \mytime_y(j+1)$
for $0\leq \rho_y\leq 1$.
Thus, 
\[ 
\mytime_x(i+\rho_x) - \mytime_y(j+\rho_y) =
 \mytime_x(i) -  \mytime_y(j) + \rho_x\vdot\left(\mytime_x(i+1) -  \mytime_x(i) \right) -
 \rho_y\vdot\left(\mytime_y(j+1) -  \mytime_y(j) \right)
\] 
We have $|\mytime_x(i+\rho_x) - \mytime_y(j+\rho_y) |\leq \delta$
iff $\mytime_x(i+\rho_x) - \mytime_y(j+\rho_y) \leq \delta$ and
$\mytime_y(j+\rho_y) -\mytime_x(i+\rho_x) \leq \delta$.
Thus, the points $(i+\rho_x, j+\rho_y)$ which are in the intersection of the 
two half-planes of the proposition constitute the set $ \free_{\delta}(\mytime_x, \mytime_y)$.
We conclude the proof by noting that the two lines denoting the half-plane
boundaries are parallel.
\end{proof}

\begin{proposition}[Convexity of Free Space]
\label{proposition:FreeSpaceConvex}
Let $f_{[i]}:[i, i+1] \rightarrow \reals^n$ and $g_{[j]}:[j, j+1] \rightarrow \reals^n$ be 
straight line segments of two polygonal curves
Given any $\delta \geq 0$, the set  $\free_{\delta}(f_{[i]}, g_{[j]})$ is convex for any norm.
\end{proposition}
\begin{proof}
Let $0\leq \lambda\leq 1$, and let $(\rho_f^0, \rho_g^0)\in [i, i+1]\times[j,j+1]$ and
 $(\rho_f^1, \rho_g^1)  \in [i, i+1]\times[j,j+1] $ be in $\free_{\delta}(f, g)$.
Consider $(1-\lambda)\vdot (\rho_f^0, \rho_g^0) + \lambda\vdot (\rho_f^1, \rho_g^1)$.\\
We have $\lVert
f_{[i]}\left( (1-\lambda)\vdot\rho_f^0  +  \lambda\vdot \rho_f^1 \right)\, - \, 
g_{[j]}\left( (1-\lambda)\vdot\rho_g^0  +  \lambda\vdot \rho_g^1 \right)
\rVert = $
\begin{align*}
&  \left\lVert\, 
(1-\lambda)\vdot f_{[i]}( \rho_f^0)  +  \lambda\vdot f_{[i]}(\rho_f^1) \, - \, 
\Big((1-\lambda)\vdot g_{[j]}( \rho_g^0)  +  \lambda\vdot g_{[j]}(\rho_g^1)\Big)  \, 
\right\rVert\\
& \text{ Since } f_{[i]} \text{ and } g_{[j]} \text{ are  line segments containing }\\
& \qquad f_{[i]}( \rho_f^0) \text{ and }  f_{[i]}( \rho_f^1); \text{ and }   
 g_{[j]}( \rho_g^0) \text{ and }  g_{[j]}( \rho_g^1) \text{ respectively.}\\
& =  \left\lVert\, 
(1-\lambda)\vdot \Big(f_{[i]}( \rho_f^0)   - g_{[j]}( \rho_g^0)  \Big) +
\lambda\vdot \Big(f_{[i]}( \rho_f^1)   - g_{[j]}( \rho_g^1)  \Big) 
\right\rVert\\
& \leq (1-\lambda)\vdot  \left\lVert\, f_{[i]}( \rho_f^0)   - g_{[j]}( \rho_g^0) \right\rVert
\, +\, 
\lambda\vdot  \left\lVert\, f_{[i]}( \rho_f^1)   - g_{[j]}( \rho_g^1) \right\rVert\\
& \leq  (1-\lambda)\vdot \delta + \lambda\vdot\delta
\end{align*}
Thus, $(1-\lambda)\vdot (\rho_f^0, \rho_g^0) + \lambda\vdot (\rho_f^1, \rho_g^1)$
belongs to  $\free_{\delta}(f_{[i]}, g_{[j]})$.
\end{proof}

\subsection{Characterizing the Free Space for Different Norms}
We give characterizations of 
$\free_{\delta}(f_{[i]}, g_{[j]}$ for affine line segments
$f_{[i]}, g_{[j]}$ for the three
standard norms $L_1, L_2$, and$L_{\infty}$, as well as for
the three other norms  $L_1^{\skoro}, L_{\infty}^{\skoro}$ and $L_2^{\skoro}$.

Recall that
the $L_1$ norm defines $\lVert (d_1, \dots, d_n) \rVert_{L_1}$ to be
$\sum_{k=1}^n |d_k|$;
the $L_2$ norm defines $\lVert (d_1, \dots, d_n) \rVert_{L_2}$ to be
$\sqrt{\sum_{k=1}^n d_k^2}$; and
the $L_{\infty}$ norm defines $\lVert (d_1, \dots, d_n) \rVert_{L_{\infty}}$ to be
$\max_k |d_k|$.
The  $L_1^{\skoro}, L_{\infty}^{\skoro}$ and $L_2^{\skoro}$ norms are obtained from
these three standard norms as defined in Definition~\ref{def:CurveFromTrace},
thus,  $\lVert\tuple{o,t} \rVert_{\chi^{\skoro}} = \max(\lVert o\rVert_{\chi}, |t|)$
for $\chi\in \set{L_1, L_2, L_{\infty}}$.
If $f$ is a curve $f:  [a_f, b_f] \rightarrow \reals^n$, we denote by $f_k$ 
the $k$-th dimension component
of $f$ for $0\leq k\leq n$.
Note that  $f_k$ is also a curve $f_k: [a_f, b_f] \rightarrow \reals$
In the following we use the following two affine identities for the affine
portions of polygonal curves.
For $(\rho_f, \rho_g) \in [ i, i+1] \times  [j, j+1]$, 
and  $(\rho_f^*, \rho_g^*) = (\rho_f-i, \rho_g-j) $
we have:
\begin{equation}
\label{equation:Affine}
\begin{array}{l}
f(\rho_f^*) = f(i) +\left (f(i+1)-f(i) \right) \vdot \rho_f^*\\ 
g(\rho_g^*) = g(j) +\left (g(j+1)-g(j) \right) \vdot \rho_g^*
\end{array}
\end{equation}

\begin{proposition}[Free Space for $L_{\infty}$]
\label{prop:LinftyGeneral}
Let $f:[a_f, b_f] \rightarrow \reals^n$ and $g:[a_g, b_g] \rightarrow \reals^n$ be curves.
If the space  $\reals^n$ has the $L_{\infty}$ norm, then
given a $\delta \geq 0$ we have 
$\free_{\delta}(f, g) = \cap_{k=1}^m \free_{\delta}(f_k, g_k) $, where $f_k, g_k$ are the 
$k$-th dimension component curves of $f,g$.
\end{proposition}
\begin{proof}
We have
$(\rho_{f}, \rho_{g})  \in \free_{\delta}(f, g)$ iff
$\max_k |f_k(\rho_f)- g_k(\rho_g)| \leq \delta$.
This holds iff $|f_k(\rho_f)- g_k(\rho_g)| \leq \delta$ for all $0\leq k \leq n$.
The result follows.
\end{proof}

\begin{proposition}[Polygonal Free Space for $L_{\infty}$]
Let $f:[0, m_f] \rightarrow \reals^n$ and $g:[0, m_g] \rightarrow \reals^n$ be polygonal
curves.
If the space  $\reals^n$ has the $L_{\infty}$ norm, then
given a $\delta \geq 0$ we have 
$\free_{\delta}(f, g) \cap \left( [i, i+1]  \times [j, j+1]\right)$ for $i<m_f$ and
$j<m_g$ to be the intersection of the regions between  $n$ pairs of parallel lines.
Each such region is the intersection of the following two half planes for $1\leq k \leq m$,
and $(\rho_f^*, \rho_g^*) \in [ 0, 1]^2$:
\begin{gather}
f_k(i) -  g_k(j) + \left(f_k(i+1) -  f_k(i) \right) \vdot \rho_x^* +
 \left(g_k(j) -  g_k(j+1) \right)  \vdot \rho_y^*\leq \delta\\
f_k(i) -  g_k(j) + \left(f_k(i+1) -  f_k(i) \right) \vdot \rho_x^*  +
 \left(g_k(j) -  g_k(j+1) \right)    \vdot \rho_y^* \geq -\delta
\end{gather}
\end{proposition}
\begin{proof}
The result follows using similar ideas as in the proofs of  Proposition~\ref{prop:LinftyGeneral} and
Proposition~\ref{prop:FreeSpaceLines}.
\end{proof}

\begin{proposition}[Polygonal Free Space for $L_{1}$ norm]
\label{proposition:FreeSpaceLOne}
Let $f:[0, m_f] \rightarrow \reals^n$ and $g:[0, m_g] \rightarrow \reals^n$ be polygonal
curves.
If the space  $\reals^n$ has the $L_{\infty}$ norm, then
given a $\delta \geq 0$ we have 
$\free_{\delta}(f, g)  \cap \left( [i, i+1]  \times [j, j+1]\right)$ for $i<m_f$ and
$j<m_g$ to be the  part of a polytope,
formed by the intersection of the following $2^ n$ half-planes $H_{(q_1,\dots, q_n)}$,
which lies inside $[i, i+1]  \times [j, j+1]$:
\[
\sum_{k=1}^n (-1)^{q_k} \vdot \Big(
 f_k(i) -  g_k(j) + \left(f_k(i+1) -  f_k(i) \right) \vdot \rho_f^*\ +\ 
 \left(g_k(j) -  g_k(j+1) \right)\vdot \rho_g^*
\Big)\ \leq \ \delta
\]

for $(q_1,\dots, q_n) \in \set{1,2}^n$; or equivalently the region of points 
$(\rho_f, \rho_g) \in [i, i+1]  \times [j, j+1]$ 
 satisfying the
following inequality:
\[
\sum_{k=1}^n \left\lvert \Big(
 f_k(i) -  g_k(j) + \left(f_k(i+1) -  f_k(i) \right) \vdot \rho_f^*\ +\ 
 \left(g_k(j) -  g_k(j+1) \right)\vdot \rho_g^*
\right\rvert \ \leq \delta
\]
where  $(\rho_f^*, \rho_g^*) = (\rho_f-i, \rho_g-j) $.
\end{proposition}
\begin{proof}
We have 
\[
(\rho_{f}, \rho_{g})  \in \free_{\delta}(f, g) \text{ iff }
\sum_{k=1}^n |f_k(\rho_f)-g_k(\rho_g)| \ \leq \ \delta.\]
Let $f_k(\rho_f)-g_k(\rho_g) = d_k$.
Let $q^*(d_k) = 1$ if $d_k \geq 0$ and $-1$ otherwise.
Then we have 
\begin{equation}
\label{proof:LoneA}
\sum_{k=1}^n |d_k| \leq \delta \text{ iff }
\sum_{k=1}^n q^*(d_k)\vdot d_k \leq \delta. 
\end{equation}
Now observe the following. 
For any $(q_1,\dots, q_n) \in \set{1,2}^n$, we have
\begin{equation}
\label{proof:LoneB}
\sum_{k=1}^n (-1)^{q_k} \vdot  d_k\  \leq\  \sum_{k=1}^n q^*(d_k)\vdot d_k.
\end{equation}
This is because for any $q_k\in\set{1,2}$, we have 
$(-1)^{q_k} \vdot  d_k \leq  q^*(d_k)\vdot d_k $.
Moreover, we have that there exists a $(q_1,\dots, q_n) \in \set{1,2}^n$ such that
an exact inequality holds in
Equation~\ref{proof:LoneB}.
Using these facts, we have that Equation~\ref{proof:LoneA} holds iff
 the following $2^ n$ equations hold, one for each choice of 
$(q_1,\dots, q_n) \in \set{1,2}^n$:
\[
\sum_{k=1}^n (-1)^{q_k} \vdot (f_k(\rho_f)-g_k(\rho_g)) \ \leq \ \delta
\]
The result follows using Equations~\ref{equation:Affine}
\end{proof}

\begin{proposition}[Polygonal Free Space for $L_{2}$]
\label{proposition:Ltwo}
Let $f:[0, m_f] \rightarrow \reals^n$ and $g:[0, m_g] \rightarrow \reals^n$ be polygonal
curves.
If the space  $\reals^n$ has the $L_{2}$ norm, then
given a $\delta \geq 0$ we have the boundary of the the free space
$\free_{\delta}(f, g) \cap \left( [i, i+1]  \times [j, j+1]\right)$ for $i<m_f$ and
$j<m_g$ to be the intersection of an ellipse, or a parabola, with $ [i, i+1]  \times [j, j+1]$.
\end{proposition}
\begin{proof}
We have $(\rho_{f}, \rho_{g})  \in \free_{\delta}(f, g) $ iff
\[
\sum_{k=1}^m \left( f(\rho_f) - g(\rho_g) \right)^2 \leq \delta^2
\]
Using Equations~\ref{equation:Affine}, the above is equivalent to
\[
\sum_{k=1}^n \Big(
 f_k(i) -  g_k(j) + \left(f_k(i+1) -  f_k(i) \right) \vdot \rho_f^*\ +\ 
 \left(g_k(j) -  g_k(j+1) \right)\vdot \rho_g^*
\Big)^2
\leq \delta^2
\]

Let $d_k^0 =  f_k(i) -  g_k(j)$; $d_k^f= f_k(i+1) -  f_k(i)$;
and $d_k^g=g_k(j) -  g_k(j+1)$.
The previous equation can then be written as
\[
\sum_{k=1}^n \Big(
d_k^0 + d_k^f\vdot  \rho_f^* + d_k^g\vdot  \rho_g^*
\Big)^2
\leq \delta^2
\]
Expanding the above, we get
\[
\sum_{k=1}^n \Big(
(d_k^0)^2 + (d_k^f)^2\vdot  (\rho_f^*)^2 + (d_k^g)^2\vdot  (\rho_g^*)^2 +
2\vdot\left(
d_k^0 \vdot d_k^f\vdot  \rho_f^*\, + \, 
d_k^f\vdot   d_k^g\vdot  \rho_f^*\vdot \rho_g^* \, +\, 
d_k^0 \vdot d_k^g\vdot  \rho_g^*
\right) 
\Big)
\ \leq \delta^2
\]
Expanding the summation and rearranging, we get
\begin{multline}
\label{Equation:Conic}
\left(\sum_{k=1}^n(d_k^f)^2\right)\vdot  (\rho_f^*)^2 \, + 
\left(\sum_{k=1}^n(d_k^g)^2\right) \vdot   (\rho_g^*)^2\,  +
\left(\sum_{k=1}^n 2\vdot d_k^f\vdot   d_k^g\right) \vdot  \rho_f^*\vdot \rho_g^*\,  +
\left(\sum_{k=1}^n 2\vdot d_k^0 \vdot d_k^f\right)\vdot  \rho_f^* \, +\\
\hspace{9cm}
\left(\sum_{k=1}^n 2\vdot d_k^0 \vdot d_k^g\right)\vdot  \rho_g^* \, +
 \sum_{k=1}^n (d_k^0)^2 -  \delta^2 
\quad \leq\quad 0
\end{multline}

Recall that the type of the conic $Ax^2 + B xy + C^2y + Fx + Gy +H =0$ depends
on the sign of $B^2-4AC$ (see \emph{e.g.} \cite{BrannanBook}).
This difference in our case is
\[
\left(\sum_{k=1}^n 2\vdot d_k^f\vdot   d_k^g\right)^2 \, -\, 
4\vdot \left(\sum_{k=1}^n(d_k^f)^2\right) \vdot \left(\sum_{k=1}^n(d_k^g)^2\right)
\]
To simplify the notation, we let $d_k^f = x_k$, and
$d_k^g = y_k$. 
The above term is then, dropping the factor of $4$,
\[
\left(\sum_{k=1}^n x_k\vdot y_k \right)^2 \, -\, 
 \left(\sum_{k=1}^n x_k^2\right)\vdot  \left(\sum_{k=1}^n y_k^2\right)
\]
This term is $\leq 0$ by the Cauchy-Schwartz inequality.
Thus, the boundary of the free space inside a cell is either an ellipse, or a parabola
(it is a parabola iff there is  constant $P$ such that for all $k$  we have
 $d_k^g = P\vdot d_k^f $).
The equation of the boundary inside the cell is Equation~\ref{Equation:Conic}.
\end{proof}

\begin{proposition}[Free Space for $L_1^{\skoro}, L_2^{\skoro}, L_{\infty}^{\skoro}$ norms]
\label{prop:SkoroNormsGeneral}
Let $f:[a_f, b_f] \rightarrow \reals^{n}$ and  $f_{n+1}:[a_f, b_f] \rightarrow \reals$, and
$g:[a_g, b_g] \rightarrow \reals^{n}$ and $g:[a_g, b_g] \rightarrow \reals$ 
 be  curves.
Note that $\tuple{f,f_{n+1}}$ is a curve from $[a_f, b_f] $ to $\reals^{n}\times\reals$.
If the space  $\reals^n\times \reals$ has the norm  ${\chi}^{\skoro}$ for 
${\chi}\in \set{L_1, L_2, L_{\infty}}$, then
given a $\delta \geq 0$ we have 
\[
\free_{\delta}^{{\chi}^{\skoro}}\left(\tuple{f,f_{n+1}}\, , \,\tuple{g,g_{n+1}}\right)
\ = \ \free_{\delta}^{\chi}(f, g) \, \cap\, \free^{L_1}_{\delta}(f_{n+1}, g_{n+1})
\]
where
$\free_{\delta}^{\chi}()$ denote the free space for norm ${\chi}$.
\end{proposition}
\begin{proof}
We have
$(\rho_{f}, \rho_{g})  \in \free_{\delta}(f, g)$ iff
$\max\left(\norm{f(\rho_f)- g(\rho_g)}_{\chi},  \abs{f_{n+1}(\rho_f)- g_{n+1}(\rho_g)} \right)
 \leq \delta$.
This holds iff both
$\norm{f(\rho_f)- g(\rho_g)}_{\chi}\leq \delta$ and 
$ \abs{f_{n+1}(\rho_f)- g_{n+1}(\rho_g)} \leq\delta$.
The result follows.
\end{proof}

\subsection{Computing Free Space Cell Boundaries}

The algorithm for computing the \frechet distance requires computation of
the free space at the cell boundaries, \emph{i.e}
$\big(\set{i}\times[j, j+1]\big) \cap \free_{\delta}(f,g)$ for
$0\leq i \leq m_f$ and $0\leq j \leq m_g-1$;
and
$\big([i, i+1] \times \set{j}\big) \cap \free_{\delta}(f,g)$ for
$0\leq i \leq m_f-1$ and $0\leq j \leq m_g$.
Using  Proposition~\ref{proposition:FreeSpaceConvex}, we get that thus
the free space at the cell boundaries is also convex, and hence just
a line.
Hence, it suffices to just compute maximum and minimum coordinate values for
the free space at the cell boundaries.
We do this for the different norms.

\smallskip\noindent\textbf{Computing Free Space Cell Boundaries
for $L_1$.}
Recall from Proposition~\ref{proposition:FreeSpaceLOne} that the
free space for for cell $i,j$ for $0\leq i\leq m_f-1$ and $0\leq j\leq m_g-1$
is given by the inequaltity 
\[
\sum_{k=1}^n \left\lvert \Big(
 f_k(i) -  g_k(j) + \left(f_k(i+1) -  f_k(i) \right) \vdot \rho_f^*\ +\ 
 \left(g_k(j) -  g_k(j+1) \right)\vdot \rho_g^*
\right\rvert \ \leq \delta
\]
where  $(\rho_f^*, \rho_g^*) = (\rho_f-i, \rho_g-j) $, with $i\leq \rho_f \leq i+1$ and
$j\leq \rho_g \leq j+1$.
Letting $r_k^0 =  f_k(i) -  g_k(j) $, and
$r_k^f = f_k(i+1) -  f_k(i) $, and
$r_k^g = g_k(j) -  g_k(j+1) $, and
$x= \rho_f^*$, and 
$y= \rho_g^*$, we obtain the inequality
\begin{equation}
\label{equation:FreeSpaceBoundaryCellLOne}
\sum_{k=1}^n \left\lvert 
r_k^0 + r_k^f \vdot x + r_k^g \vdot y
\right\rvert \ \leq \delta
\end{equation}
with $0\leq x\leq 1$, and $0\leq y\leq 1$.

We show how to obtain the free space boundary at the cell boundary
$0\leq x\leq 1, y=1$. 
The other cases are similar.
We need to compute the maximum, and the minimum values of $0\leq x\leq 1$
such that Equation~\ref{equation:FreeSpaceBoundaryCellLOne} holds with
$y=1$.
Substituting $y =1$, and letting $r_k^1 = r_k^0- r_k^g$, we thus, wish to solve
the following two optimization problems.
\begin{alignat*}{2}
  \text{minimize  }\  & x \\
  \text{subject to }\ & \sum_{k=1}^n\abs{r_k^1 +  r_k^f \vdot x } \leq \delta\\
  &0\leq x\leq 1
\end{alignat*}
and 
\begin{alignat*}{2}
  \text{maximize  }\  & x \\
  \text{subject to }\ & \sum_{k=1}^n\abs{r_k^1 +  r_k^f \vdot x } \leq \delta\\
  &0\leq x\leq 1
\end{alignat*}
where $r_k^1,  r_k^f,$ and $\delta$ are given constants.
We show how to solve the maximization problem (the minimization problem is similar).
Assume none of $ r_k^f$ is zero (if some $ r_k^f$ is zero, remove 
$\abs{r_k^1 +  r_k^f \vdot x }$ from the sum, and change $\delta$ to 
$\delta - \abs{r_k^1}$).
We have:
\begin{equation}
\label{equartion:AbsForm}
\abs{r_k^1 +  r_k^f \vdot x } = 
\begin{cases}
r_k^1 +  r_k^f \vdot x & \text{if } x \geq -\frac{r_k^1}{ r_k^f} \text{ and }  r_k^f > 0\\
-\left(r_k^1 +  r_k^f \vdot x\right)
 & \text{if } x \leq   -\frac{r_k^1}{ r_k^f} \text{ and }  r_k^f > 0\\
r_k^1 +  r_k^f \vdot x & \text{if } x \leq -\frac{r_k^1}{ r_k^f} \text{ and }  r_k^f < 0\\
-\left(r_k^1 +  r_k^f \vdot x\right)  & \text{if } x \geq  -\frac{r_k^1}{ r_k^f} \text{ and }  r_k^f < 0
\end{cases}
\end{equation}
We compute the $m$ values 
$-\frac{r_k^1}{ r_k^f}$
for $1\leq k\leq n$, and sort these values in increasing order into an array
$X^{\dagger}[1..n]$.
We remove all values that are $<0$ or $>1$, add the values $0$ and $1$,
and remove all duplicates to get the sorted array $X[1..n']$ with no duplicates,
and in increasing order,
with $X[1] = 0$, and $X[n'] = 1$.
Observe that for $X[k] \leq x\leq X[k+1]$, each of
$\abs{r_k^1 +  r_k^f \vdot x } $ is either $r_k^1 +  r_k^f \vdot x$, or
$-\left(r_k^1 +  r_k^f \vdot x\right)$, \emph{i.e.}, the value of
$r_k^1 +  r_k^f \vdot x$ is either $\geq 0$ throughout, or $\leq 0$ throughout
the interval.
Using this fact, we determine the form of the function 
$\sum_{k=1}^n\abs{r_k^1 +  r_k^f \vdot x }$ over the interval
 $X[j] \leq x\leq X[j+1]$ as follows.
We have:
\begin{equation}
\label{equation:AbsSumForm}
\sum_{k=1}^n\abs{r_k^1 +  r_k^f \vdot x } \ = \ 
\sum_{k=1}^n (-1)^{\mu_j(k)} \cdot \left( r_k^1 +  r_k^f \vdot x\right)
\text{ over the interval  }X[j] \leq x\leq X[j+1] \text{ where, }
\end{equation}

\begin{equation}
\mu_j(k) =
\begin{cases}
 0 & \text{ if } X[j] \geq  -\frac{r_k^1}{ r_k^f} \text{ and }  r_k^f > 0\\
1 & \text{ if } X[j+1]  \leq   -\frac{r_k^1}{ r_k^f} \text{ and }  r_k^f > 0\\
0 &  \text{ if } X[j+1]  \leq   -\frac{r_k^1}{ r_k^f} \text{ and }  r_k^f < 0\\
 1 & \text{ if } X[j] \geq  -\frac{r_k^1}{ r_k^f} \text{ and }  r_k^f < 0
\end{cases}
\end{equation}
We prove that the function $\mu()$ above is well defined.
\begin{compactenum}
\item First, we show that either
$X[j] \geq  -\frac{r_k^1}{ r_k^f}$, or
$X[j+1]  \leq   -\frac{r_k^1}{ r_k^f}$ hold for all $1\leq k, j \leq n'$.
This is because if this condition does not hold, then both 
$X[j] <-\frac{r_k^1}{ r_k^f}$ and 
$X[j+1]  >  -\frac{r_k^1}{ r_k^f}$ hold, \emph{i.e.},
$X[j] <-\frac{r_k^1}{ r_k^f}   < X[j+1] $, which is impossible since
$X$ is a sorted array which contains $-\frac{r_k^1}{ r_k^f} $.
\item 
We show both $X[j] \geq  -\frac{r_k^1}{ r_k^f}$ and
$X[j+1]  \leq   -\frac{r_k^1}{ r_k^f}$ cannot hold.
If both these hold, then 
$X[j] \geq  -\frac{r_k^1}{ r_k^f} \geq X[j+1] $; which is impossible
since $X$ is an sorted array in increasing order with no duplicates.
\end{compactenum}
The fact that 
$\abs{r_k^1 +  r_k^f \vdot x }$ over the interval
$X[j] \leq x\leq X[j+1]$ equals 
$(-1)^{\mu_j(k)} \cdot \left( r_k^1 +  r_k^f \vdot x\right)$
then follows from Equation~\ref{equartion:AbsForm}.
Thus, Equation~\ref{equation:AbsSumForm} is correct.

Once the function $\mu_j()$ has been computed, the
function $\sum_{k=1}^n\abs{r_k^1 +  r_k^f \vdot x }$
can be written as 
$ \left(\sum_{k=1}^n  (-1)^{\mu_j(k)} \cdot r_k^1\right) + 
\left(\sum_{k=1}^n  (-1)^{\mu_j(k)} \cdot r_k^1\right)\cdot x $
over the interval $X[j] \leq x\leq X[j+1]$.
This function is either monotonically non-decreasing, or
non-increasing.
The maximum value of $\sum_{k=1}^n\abs{r_k^1 +  r_k^f \vdot x }$
over the interval $X[j] \leq x\leq X[j+1]$  is thus,
\[
\begin{cases}
\sum_{k=1}^n (-1)^{\mu_j(k)} \cdot \left( r_k^1 +  r_k^f \vdot X[j+1]\right)
& \text{ if } \left(\sum_{k=1}^n  (-1)^{\mu_j(k)} \cdot r_k^f\right) \geq 0\\
\sum_{k=1}^n (-1)^{\mu_j(k)} \cdot \left( r_k^1 +  r_k^f \vdot X[j]\right)
& \text{ if } \left(\sum_{k=1}^n  (-1)^{\mu_j(k)} \cdot r_k^f\right) < 0
\end{cases}
\]
The maximum value of $\sum_{k=1}^n\abs{r_k^1 +  r_k^f \vdot x }$
over the interval $0\leq x\leq 1$
can then be obtained as the maximum
 of the maximum values of $\sum_{k=1}^n\abs{r_k^1 +  r_k^f \vdot x }$
over the $n'$ intervals  $X[j] \leq x\leq X[j+1]$ for $1\leq j <n'$.

Computing and sorting the $-r_k^1/ r_k^f$ values takes $O\left(n\vdot\log(n)\right)$
time.
Computing the function $\mu_j()$ takes time $O(n)$ for each $j$,
thus computing all the functions  $\mu_j()$ takes time $O(n^2)$.
This also means that computing the maximums for all
of the $n'$ intervals takes $O(n^2)$ time.
The rest of the steps take time $O(n)$.
Putting everything together, we have the following proposition.

\begin{proposition}
Given  polygonal curves  $f$ and $g$ with values in $\reals^n$,
the intersection of the free space $\free_{\delta}(f,g)$ with the 
boundary of the cell $i,j$ can be computed in time
$O(n^2)$ for the $L_1$ norm.
\qed
\end{proposition}

\smallskip\noindent\textbf{Computing Free Space Cell Boundaries
for $L_2$.}
We show how to compute the free space boundary at the cell boundary
$0\leq \rho_f^* \leq 1; \rho_g^*= 1$, where $(\rho_f^*, \rho_g^*) = (\rho_f-i, \rho_g-j) $.
The computation for the other boundaries is similar.
Substituting $\rho_g^*= 1$ in Equation~\ref{Equation:Conic}, denoting
$ \rho_f^* $ as $x$, and simplifying, we get:
\[
\left(\sum_{k=1}^n(d_k^f)^2\right)\vdot x^2 + 
\left(\left(\sum_{k=1}^n 2\vdot d_k^f\vdot   d_k^g\right) + 
\left(\sum_{k=1}^n 2\vdot d_k^0 \vdot d_k^f\right)\right)
 \vdot x +
\sum_{k=1}^n\Big((d_k^g)^2 + 2\vdot d_k^0 \vdot d_k^g + (d_k^0)^2\Big)
\, - \delta^2 \ \leq 0.
\]
This is of the form
\[
A\vdot x^2 + B\vdot x + C \leq 0\]
We wish to find the maximum and minimum values of $x\in [0,1]$ which satisfy the above equation.
The equation $A\vdot x^2 + B\vdot x + C= 0$ has two roots given by
$\frac{-b \pm \sqrt{B^2- 4A\vdot C}}{2\vdot A}$.
The following three cases arise.
\begin{compactitem}
\item $B^2- 4A\vdot C < 0$, thus, there are no real roots to $A\vdot x^2 + B\vdot x + C= 0$.
This means that for all $x$, the value of 
$A\vdot x^2 + B\vdot x + C$ is either always greater than $0$, or always less than $0$.
If $C > 0$, then $A\vdot x^2 + B\vdot x + C $ is always greater than $0$ (using $x=0$),
thus, the free space boundary is empty.
If  $C < 0$,  then $A\vdot x^2 + B\vdot x + C $ is always less than $0$, thus,
the free space boundary is the entire segment $0\leq x\leq 1$ at $y=1$.
Note that we cannot have $C=0$ as $x=0$ is not a  solution to
$A\vdot x^2 + B\vdot x + C= 0$
\item $B^2- 4A\vdot C = 0$; thus there is exactly one real root $x^{\dagger}$ to
$A\vdot x^2 + B\vdot x + C= 0$.
This means that for all $x\neq x^{\dagger}$, the value of 
$A\vdot x^2 + B\vdot x + C$ is either always greater than  $0$,
or  always less than  $0$.
\begin{compactitem}
\item 
If  $C \leq  0$,  then $A\vdot x^2 + B\vdot x + C $ is always less than $0$ for
$x\neq x^{\dagger}$, and at  $x= x^{\dagger}$ the value is $0$.
Hence, the free space boundary is the entire segment $0\leq x\leq 1$ at $y=1$.
\item 
If  If $C > 0$, then $A\vdot x^2 + B\vdot x + C $ is always greater than $0$
for $x\neq x^{\dagger}$.
Thus, the only $x$ such that $A\vdot x^2 + B\vdot x + C \leq 0$ is $x= x^{\dagger}$.
Hence the free space boundary is the singeton point $x^{\dagger}$ if
$x^{\dagger} \in [0,1]$; otherwise it is the emptyset.
\end{compactitem}
\item $B^2- 4A\vdot C > 0$;  which means that  there are two real roots to
$A\vdot x^2 + B\vdot x + C= 0$.
The following two situations can arise
\begin{figure}[h]
\strut\centerline{\setlength{\unitlength}{0.00043745in}
\begingroup\makeatletter\ifx\SetFigFont\undefined%
\gdef\SetFigFont#1#2#3#4#5{%
  \reset@font\fontsize{#1}{#2pt}%
  \fontfamily{#3}\fontseries{#4}\fontshape{#5}%
  \selectfont}%
\fi\endgroup%
{\renewcommand{\dashlinestretch}{30}
\begin{picture}(7666,1434)(0,-10)
\path(12,687)(3342,687)
\path(4242,687)(7572,687)
\path(192,1362)(192,12)
\path(4692,1407)(4692,57)
\path(57,192)(59,193)(62,197)
	(69,203)(80,212)(95,225)
	(114,242)(138,262)(165,286)
	(197,314)(231,344)(268,375)
	(307,408)(346,442)(385,475)
	(424,509)(463,541)(500,572)
	(536,601)(570,629)(602,656)
	(633,681)(663,704)(691,726)
	(718,746)(743,765)(768,783)
	(792,800)(815,816)(838,831)
	(860,846)(882,859)(909,876)
	(936,891)(963,906)(990,921)
	(1017,935)(1045,948)(1073,961)
	(1101,974)(1129,985)(1158,996)
	(1187,1007)(1216,1017)(1244,1026)
	(1273,1035)(1301,1043)(1329,1050)
	(1357,1056)(1383,1062)(1410,1067)
	(1436,1071)(1461,1075)(1485,1078)
	(1510,1080)(1533,1082)(1556,1083)
	(1580,1084)(1604,1085)(1629,1085)
	(1655,1085)(1680,1084)(1707,1082)
	(1733,1080)(1760,1077)(1788,1074)
	(1816,1070)(1844,1065)(1872,1059)
	(1901,1053)(1929,1046)(1958,1038)
	(1986,1030)(2013,1022)(2041,1012)
	(2067,1003)(2093,992)(2119,982)
	(2144,971)(2169,959)(2193,947)
	(2217,934)(2239,922)(2262,910)
	(2284,896)(2307,882)(2330,867)
	(2354,852)(2378,836)(2402,819)
	(2426,802)(2451,784)(2476,765)
	(2500,747)(2525,728)(2549,709)
	(2573,690)(2596,671)(2619,652)
	(2642,633)(2663,615)(2684,597)
	(2704,580)(2724,563)(2742,546)
	(2760,530)(2778,515)(2795,499)
	(2814,482)(2833,465)(2852,447)
	(2870,430)(2889,412)(2908,394)
	(2928,375)(2949,356)(2970,335)
	(2993,313)(3016,291)(3039,268)
	(3062,245)(3085,223)(3105,203)
	(3123,186)(3138,171)(3149,160)
	(3156,153)(3160,149)(3162,147)
\path(4549,1084)(4551,1083)(4554,1079)
	(4561,1073)(4572,1064)(4587,1051)
	(4606,1034)(4630,1014)(4657,990)
	(4689,962)(4723,932)(4760,901)
	(4799,868)(4838,834)(4877,801)
	(4916,767)(4955,735)(4992,704)
	(5028,675)(5062,647)(5094,620)
	(5125,595)(5155,572)(5183,550)
	(5210,530)(5235,511)(5260,493)
	(5284,476)(5307,460)(5330,445)
	(5352,430)(5374,416)(5401,400)
	(5428,385)(5455,370)(5482,355)
	(5509,341)(5537,328)(5565,315)
	(5593,302)(5621,291)(5650,280)
	(5679,269)(5708,259)(5736,250)
	(5765,241)(5793,233)(5821,226)
	(5849,220)(5875,214)(5902,209)
	(5928,205)(5953,201)(5977,198)
	(6002,196)(6025,194)(6048,193)
	(6072,191)(6096,191)(6121,191)
	(6147,191)(6172,192)(6199,194)
	(6225,196)(6252,199)(6280,202)
	(6308,206)(6336,211)(6364,217)
	(6393,223)(6421,230)(6450,238)
	(6478,246)(6505,254)(6533,264)
	(6559,273)(6585,284)(6611,294)
	(6636,305)(6661,317)(6685,329)
	(6709,341)(6731,354)(6754,366)
	(6776,380)(6799,394)(6822,409)
	(6846,424)(6870,440)(6894,457)
	(6918,474)(6943,492)(6968,511)
	(6992,529)(7017,548)(7041,567)
	(7065,586)(7088,605)(7111,624)
	(7134,643)(7155,661)(7176,679)
	(7196,696)(7216,713)(7234,730)
	(7252,746)(7270,761)(7287,776)
	(7306,794)(7325,811)(7344,829)
	(7362,846)(7381,864)(7400,882)
	(7420,901)(7441,920)(7462,941)
	(7485,963)(7508,985)(7531,1008)
	(7554,1031)(7577,1053)(7597,1073)
	(7615,1090)(7630,1105)(7641,1116)
	(7648,1123)(7652,1127)(7654,1129)
\put(5682,417){\makebox(0,0)[lb]{\smash{{\SetFigFont{8}{9.6}{\familydefault}{\mddefault}{\updefault}$x\longrightarrow$}}}}
\put(732,417){\makebox(0,0)[lb]{\smash{{\SetFigFont{8}{9.6}{\familydefault}{\mddefault}{\updefault}$x\longrightarrow$}}}}
\end{picture}
}}
 \caption{Graphs of $A\vdot x^2 + B\vdot x + C $.}
\label{figure:quadratic}
\end{figure}
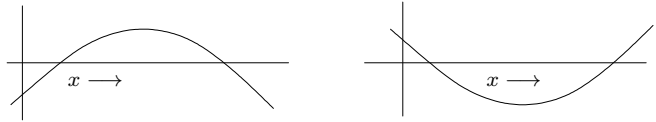
Let the two roots of $A\vdot x^2 + B\vdot x + C \leq 0$  be
$x^{\dagger}_1$ and $x^{\dagger}_2$ with $x^{\dagger}_1 <x^{\dagger}_2$.
The derivative at  $A\vdot x^2 + B\vdot x + C $ is $2\vdot A \vdot x + B$,
thus we have the first situation if $2\vdot A \vdot x^{\dagger}_1 + B> 0$, and
the second situation otherwise.
\begin{compactitem}
\item 
If  $2\vdot A \vdot x^{\dagger}_1 + B> 0$, then
 we have the first situation and hence  the free space boundary is
$[0,  x^{\dagger}_1] \cup [x^{\dagger}_2, 1]$.
\item If  $2\vdot A \vdot x^{\dagger}_1 + B< 0$, then
 we have the second situation and hence the free space boundary is
$[0,1] \cap [x^{\dagger}_1, x^{\dagger}_2]$.
\end{compactitem}

\end{compactitem}
We note that the preceeding computation steps take $O(n)$ time in total.

\medskip\noindent\textbf{Computing Free Space Cell Boundaries
for $L_{\infty}$.}
We show how to compute the free space boundary at the cell boundary
$0\leq \rho_f^* \leq 1; \rho_g^*= 1$, where $(\rho_f^*, \rho_g^*) = (\rho_f-i, \rho_g-j) $.
The computation for the other boundaries is similar.
We use Proposition~\ref{prop:LinftyGeneral}.
We  have that the free space
boundary at $0\leq \rho_f^* \leq 1; \rho_g^*= 1$
for the $k$-th curve components $f_k$ and $g_k$
is the set of $(\rho_f^*, \rho_g^*)$ with $\rho_g^*= 1$ such that
$|f_k(\rho_f^*)- g_k(\rho_g^*)| \leq \delta$.
Using Equation~\ref{equation:Affine}, and substituting $\rho_g^*= 1$,
we get that
\[
\abs{f_k(i) -g_k(j+1) + \left (f(i+1)-f(i) \right) \vdot \rho_f^*} \leq \delta\]
If $f(i+1) = f(i)$, then the above equation gives all $ \rho_f^*$ as being valid
if $ \abs{f_k(i) -g_k(j+1) } \leq \delta$
(thus the free space boundary if the entire segment $0\leq \rho_f^* \leq 1; \rho_g^*= 1$);
 and no  $ \rho_f^*$ as being valid otherwise (thus the free space boundary is the 
emptyset).\\
If $f(i+1) \neq f(i)$, then the equation for the $k$-th components states:
\[
\rho_f^* \in 
\begin{cases}
[\frac{-\left(f_k(i) -g_k(j+1) \right) - \delta}{f(i+1)-f(i)}, 
\frac{-\left(f_k(i) -g_k(j+1) \right)  + \delta}{f(i+1)-f(i)}] & \text{ if } f(i+1)-f(i) > 0\\
[\frac{-\left(f_k(i) -g_k(j+1) \right)  + \delta}{f(i+1)-f(i)}, 
\frac{-\left(f_k(i) -g_k(j+1) \right) - \delta}{f(i+1)-f(i)}] & \text{ if } f(i+1)-f(i) <0

\end{cases}
\]
Thus, in all cases, the  $k$-th component equations give us
the valid $\rho_f^* $  as being the set $ [x_k, x_k']$ with 
$x_k, x_k'$ determined as above.
Using Proposition~\ref{prop:LinftyGeneral}, we get that 
the free space
boundary at $0\leq \rho_f^* \leq 1; \rho_g^*= 1$
is the emptyset if
$[x_k , x_k']\cap [0,1] = \emptyset$ for some $k$,
otherwise, it is 
$[x_f, x_f'] = [0,1]\cap \cap_{k=1}^n[x_k, x_k']$.
This intersection can be determined as follows: 
$x_f = \max\left(0, \max\set{x_k}\right)$, and
$x_f' = \min\left(1, \min\set{x'_k}\right)$.
We note that the preceeding computation steps take $O(n)$ time in total.


\medskip\noindent\textbf{Computing Free Space Cell Boundaries
for $L_1^{\skoro}, L_2^{\skoro}, L_{\infty}^{\skoro}$.}
Let $\tuple{f,f_{n+1}}$ be a curve from $[a_f, b_f] $ to $\reals^{n}\times\reals$,
and similarly for $\tuple{g,g_{n+1}}$ as in Proposition~\ref{prop:SkoroNormsGeneral}.
To compute the free space boundaries of these two curves for
$L_1^{\skoro}, L_2^{\skoro}$ or $ L_{\infty}^{\skoro}$ norms,
we  first compute the free space boundary for the norm
$L_1, L_2$ or $L_{\infty}$ for the curves $f,g$.
Then we compute the free space corresponding to the last component,
$f_{n+1}, g_{n+1}$ -- this computation is the same as the computation of the
free space for individual coordinate components for $L_{\infty}$.
Then we intersect the two boundaries.
The time taken is $O(n^2)$ for $L_1^{\skoro}$, and
$O(n)$ for $L_2^{\skoro}$ and $ L_{\infty}^{\skoro}$.

\subsection{Algorithm for the \frechet-Distance Decision Problem
Given a Fixed $\delta$}

In this section we solve for the \frechet distance decision problem between
two polygonal curves for
a given fixed $\delta$.
The decision problem is solved with a dynamic programming algorithm on 
the free space diagram.
As noted before, our formulation of the \frechet distance requires the reparametrizations
to be strictly increasing, as opposed to the forumlation of~\cite{AltG95} which
only requires non-decreasing reparametrizations.
This introduces some complications which we address in our solution.

Consider cell $i,j$ in the free space diagram of two polygonal curves $f,g$.
The cell together with the free space inside it is depicted in
Figure~\ref{figure:Cell}
\footnote{We use the convention that the first coordinate $i$ of cell $i,j$
increases in the horizontal direction, and that the second coordinate
in the vertical direction, thus, the same as for the cartesian plane.}. 
The non-shaded portion is the free space.
Let $\mye_{i,j}^0, \mye_{i,j}^1,  \mye_{i,j}^2, $ and $ \mye_{i,j}^3$ denote
the bottom, right, top, and left edges of  cell $i,j$ respectively.
Thus, $\mye_{i,j}^0$ is the edge $[i,i+1]\times \set{j}$, and the other edges
are the ones encountered moving anti-clockwise.
Let $\mya_{i,j}^0$ and $\myb_{i,j}^0$ be the starting and ending points of
edge $\mye_{i,j}^0$, and similarly for the other edges (see 
Figure~\ref{figure:CellBoundaries}).
Given a point $s= (p,q) \in \reals^2$, let $\first(s) = p$ denote the first
coordinate, and $\second(s) = q$ denote the second coordinate.
The points $\mya_{i,j}^q, \myb_{i,j}^q$ for $0\leq q\leq 3$ can be obtained for each cell
using the results of the previous section for the $L_1, L_2, L_{\infty}$
and  $L_1^{\skoro}, L_2^{\skoro}, L_{\infty}^{\skoro}$ norms.
Note that $(\mya_{i+1,j}^3, \myb_{i+1,j}^3) = (\mya_{i,j}^1, \myb_{i,j}^1)$;
and $(\mya_{i,j+1}^0, \myb_{i,j+1}^0) = (\mya_{i,j+1}^2, \myb_{i,j+1}^2)$

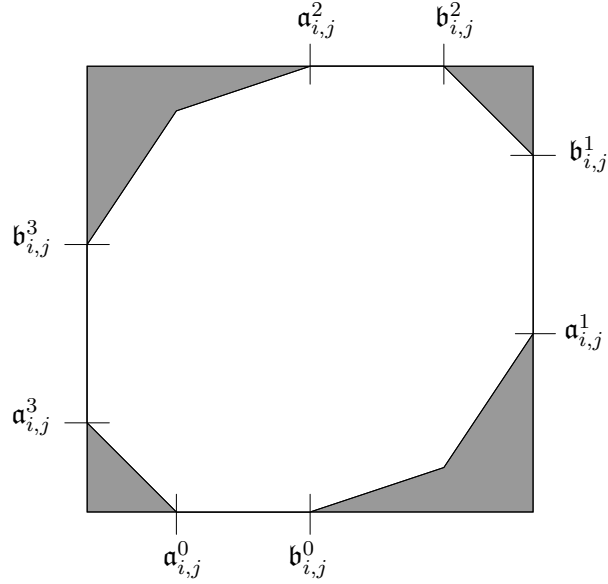
\begin{figure}[h]
\label{figure:CellBoundaries}
\strut\centerline{\setlength{\unitlength}{0.00052493in}
\begingroup\makeatletter\ifx\SetFigFont\undefined%
\gdef\SetFigFont#1#2#3#4#5{%
  \reset@font\fontsize{#1}{#2pt}%
  \fontfamily{#3}\fontseries{#4}\fontshape{#5}%
  \selectfont}%
\fi\endgroup%
{\renewcommand{\dashlinestretch}{30}
\begin{picture}(5655,5787)(0,-10)
\texture{55888888 88555555 5522a222 a2555555 55888888 88555555 552a2a2a 2a555555 
	55888888 88555555 55a222a2 22555555 55888888 88555555 552a2a2a 2a555555 
	55888888 88555555 5522a222 a2555555 55888888 88555555 552a2a2a 2a555555 
	55888888 88555555 55a222a2 22555555 55888888 88555555 552a2a2a 2a555555 }
\shade\path(780,5136)(5280,5136)(5280,636)
	(780,636)(780,5136)
\path(780,5136)(5280,5136)(5280,636)
	(780,636)(780,5136)
\whiten\path(1680,636)(780,1536)(780,3336)
	(1680,4686)(3030,5136)(3480,5136)
	(3930,5136)(4380,5136)(5280,4236)
	(5280,2436)(4380,1086)(3030,636)
	(1680,636)(1680,636)
\path(1680,636)(780,1536)(780,3336)
	(1680,4686)(3030,5136)(3480,5136)
	(3930,5136)(4380,5136)(5280,4236)
	(5280,2436)(4380,1086)(3030,636)
	(1680,636)(1680,636)
\path(1680,411)(1680,816)
\path(3030,411)(3030,816)
\path(5505,2436)(5100,2436)
\path(5505,4236)(5055,4236)
\path(555,3336)(1005,3336)
\path(555,1536)(1005,1536)
\path(3030,5361)(3030,4956)
\path(4380,5361)(4380,4956)
\put(1545,96){\makebox(0,0)[lb]{\smash{{\SetFigFont{11}{13.2}{\familydefault}{\mddefault}{\updefault}$\mya_{i,j}^0$}}}}
\put(2805,96){\makebox(0,0)[lb]{\smash{{\SetFigFont{11}{13.2}{\familydefault}{\mddefault}{\updefault}$\myb_{i,j}^0$}}}}
\put(5595,2391){\makebox(0,0)[lb]{\smash{{\SetFigFont{11}{13.2}{\familydefault}{\mddefault}{\updefault}$\mya_{i,j}^1$}}}}
\put(2895,5541){\makebox(0,0)[lb]{\smash{{\SetFigFont{11}{13.2}{\familydefault}{\mddefault}{\updefault}$\mya_{i,j}^2$}}}}
\put(15,1536){\makebox(0,0)[lb]{\smash{{\SetFigFont{11}{13.2}{\familydefault}{\mddefault}{\updefault}$\mya_{i,j}^3$}}}}
\put(5640,4191){\makebox(0,0)[lb]{\smash{{\SetFigFont{11}{13.2}{\familydefault}{\mddefault}{\updefault}$\myb_{i,j}^1$}}}}
\put(4290,5541){\makebox(0,0)[lb]{\smash{{\SetFigFont{11}{13.2}{\familydefault}{\mddefault}{\updefault}$\myb_{i,j}^2$}}}}
\put(15,3336){\makebox(0,0)[lb]{\smash{{\SetFigFont{11}{13.2}{\familydefault}{\mddefault}{\updefault}$\myb_{i,j}^3$}}}}
\end{picture}
}}
 \caption{Cell $i,j$ in $\free_{\delta}(f,g)$.}
\label{figure:Cell}
\end{figure}

In order to utilize Proposition~\ref{proposition:FreeSpaceCurve} in order to
check for the existence of reparametrizations which demonstrate that 
 $\fre(f,g) \leq \delta$, we present a dynamic programming
based method
(based on the sketch in~\cite{AltG95}) to determine 
the set of points in $\free_{\delta}(f,g) $ that are reachable from $(0,0)$  by a
monotone curve which is strictly increasing in both coordinates.
The method iteratively computes the monotone curve reachable 
set at the
$4$ cell boundaries.
At each cell boundary, this set consists of a closed interval of points,
 together with possibly  one corner point.
We define the parameters   $\creach{\mya}_{i,j}^q, \creach{\myb}_{i,j}^q$ 
for $0\leq q\leq 3$, and the set $\cpoint_{i,j}$ for cell $i,j$. 
The set  $\cpoint_{i,j}$ is either empty, or contains the point $(i,j)$.
%
The ranges of the  parameters  $\creach{\mya}_{i,j}^q$ and $\creach{\myb}_{i,j}^q$ 
 are:
\begin{compactitem}
\item $\Big(\big([i, i+1] \times\set{\cdot, +}\big) \times \set{j, j+1}\Big) \, 
\cup\,  \set{\bot}$ for $\creach{\mya}_{i,j}^0, \creach{\myb}_{i,j}^0$,
and $\creach{\mya}_{i,j}^2, \creach{\myb}_{i,j}^2$.
A second coordinate value of  $j$ corresponds to  the bottom cell boundary,
and a value of $j+1$ corresponds to the top cell boundary.

\item  $\Big( \set{i, i+1} \times \big([j, j+1] \times\set{\cdot, +}\big)  \Big) \, 
\cup\,  \set{\bot}$  for $\creach{\mya}_{i,j}^1, \creach{\myb}_{i,j}^1$,
and $\creach{\mya}_{i,j}^3, \creach{\myb}_{i,j}^3$.
A first  coordinate value of  $i$ corresponds to  the left cell boundary,
and a value of $i+1$ corresponds to the right cell boundary.

\end{compactitem} 
A value $(\tuple{x,+}, y)$ denotes a point which has the first
coordinate value that is $\epsilon$ greater than $x$, for $\epsilon$ arbitrarily small,
and a second coordinate value that is exactly equal to $y$.
We need this $\epsilon$-formalism as we only interested in 
monotone curves in the free space which are \emph{strictly} increasing 
in both coordinates.
The  line from $\creach{\mya}_{i,j}^q $ to $ \creach{\myb}_{i,j}^q$ is the
subportion of the line
from ${\mya}_{i,j}^q $ to ${\myb}_{i,j}^q$ that is reachable from a monotone curve 
 from $(0,0)$ (the reachable sub-portion may be open).
The special value $\bot$ is used to indicate that
the  monotone curve reachable portion of cell boundary is empty.
The corner point needs to be treated differently for technical reasons.
%
%
%
The dynamic programming algorithm bases on the sketch from~\cite{AltG95}  for
computing the parameters is presented in  Algorithm~\ref{algo:reach}
\footnote{Note that Algorithm~\ref{algo:reach} computes whether there are
reparametrizations which \emph{achieve} the \frechet distance $\delta$.
This is not the same as determining whether the \frechet distance is at most
$\delta$.}.
The requirement of reparametrizations being strictly increasing, instead of
only being non-decreasing as in~\cite{AltG95} introduces some complications,
\emph{e.g},  now the stritcly increasing  monotone curve portion at a cell boundary
need not be a line segment, rather, it is a line segment togther  with possibly
a corner point.
The procedure for computing the $\creach{\mya}_{i,j}^q, \creach{\myb}_{i,j}^q$
values used in the algorithm  is given in the appendix.
\begin{algorithm}
  \SetKwInOut{Input}{Input}
  \Input{Polygonal curves 
    $f:[0, m_f] \rightarrow \reals^n$ and $g:[0, m_g] \rightarrow \reals^n$; and
    $\delta \geq 0$}
  \lForEach{$0\leq i\leq m_f$ and $ 0\leq j\leq m_g$}{
    compute $\mya_{i,j}^3, \myb_{i,j}^3, \mya_{i,j}^0, \myb_{i,j}^0$ \;}
   \lForEach(\\ \tcc*[f]{Vertical reach boundaries of bottom row}){$0\leq i\leq m_f$}{
     compute $\creach{\mya}_{i,0}^1, \creach{\myb}_{i,0}^1, 
     \creach{\mya}_{i,0}^3, \creach{\myb}_{i,0}^3$ and $\cpoint_{i,0}$}
   \lForEach(\\ \tcc*[f]{Horizontal reach boundaries of  first column}){$0\leq j\leq m_g$}{
     compute $\creach{\mya}_{0, j}^2, \creach{\myb}_{0,j}^2, 
     \creach{\mya}_{0, j}^0, \creach{\myb}_{0, j}^0$ and $\cpoint_{0,j}$ }
   \ForEach{$1\leq  i\leq m_f$}{
     \ForEach{$1\leq  j\leq m_g$}{
       compute $\cpoint_{i,j}$ from 
       $\creach{\mya}_{i-1,j-1}^1, \creach{\myb}_{i-1,j-1}^1$ and
        $\creach{\mya}_{i-1,j-1}^2, \creach{\myb}_{i-1,j-1}^2$\;
       compute $\creach{\mya}_{i,j}^1, \creach{\myb}_{i,j}^1$ from
       $\cpoint_{i-1,j-1}$  and $\creach{\mya}_{i-1,j}^1, \creach{\myb}_{i-1,j}^1$ and
       $\creach{\mya}_{i,j-1}^2, \creach{\myb}_{i,j-1}^2$\;
       compute $\creach{\mya}_{i,j}^2, \creach{\myb}_{i,j}^2$ from
       $\cpoint_{i-1,j-1}$  and $\creach{\mya}_{i,j-1}^2, \creach{\myb}_{i,j-1}^2$ and
       $\creach{\mya}_{i-1,j}^3, \creach{\myb}_{i-1,j}^3$\;
     }
     }
     check if $(m_f, m_g) =  \creach{\myb}_{m_f,m_g}^2$\;
\caption{Dynamic programming algorithm for checking reachability of $(m_f, m_g)$
  by a monotone curve\label{algo:reach}}
\end{algorithm}

\begin{lemma}
\label{lemma:FrechetDecisionParam}
Let $f:[0, m_f] \rightarrow \reals^n$ and $g:[0, m_g] \rightarrow \reals^n$ be polygonal
curves.
 There exists a strictly increasing monotone curve from $(0,0)$ to
$(m_f, m_g)$ in the free space diagram of $f,g$
 iff the point  $(m_f, m_g) $ belongs to the reachable
portion of the top cell boundary of cell  $(m_f-1, m_g-1)$, \emph{i.e.}
$ \creach{\myb}_{m_f,m_g}^2 = (m_f, m_g) $.
\qed
\end{lemma}


Note that since  require reparametrizations to be \emph{strictly} increasing, 
it might be the case that 
there is no strictly increasing monotone curve from $(0,0)$ to
$(m_f, m_g)$ in $\free_{\delta}(f,g)$, and that such curves exist in
$\free_{\delta+\epsilon'}(f,g)$ for every $\epsilon' > 0$.
This presents a complication in determining whether $\fre(f,g) \leq \delta$,
as a value of $\delta$ might be the limit obtained by a sequence of reparametrizations.
In the free space formulation, this situation can arise when there is only a horizontal
line (\emph{i.e.} a non-decreasing monotone curve) 
 which can cross a cell boundary for $\delta$; with strictly increasing monotone
curves only being available for $\delta'> \delta$.
See Figure~\ref{figure:NonBijective} which depicts this situation for cell $i,j$ in the
two free space diagrams, with $\delta' > \delta$.

\begin{figure}[h]
\strut\centerline{\setlength{\unitlength}{0.00061242in}
\begingroup\makeatletter\ifx\SetFigFont\undefined%
\gdef\SetFigFont#1#2#3#4#5{%
  \reset@font\fontsize{#1}{#2pt}%
  \fontfamily{#3}\fontseries{#4}\fontshape{#5}%
  \selectfont}%
\fi\endgroup%
{\renewcommand{\dashlinestretch}{30}
\begin{picture}(5424,2364)(0,-10)
\texture{55888888 88555555 5522a222 a2555555 55888888 88555555 552a2a2a 2a555555 
	55888888 88555555 55a222a2 22555555 55888888 88555555 552a2a2a 2a555555 
	55888888 88555555 5522a222 a2555555 55888888 88555555 552a2a2a 2a555555 
	55888888 88555555 55a222a2 22555555 55888888 88555555 552a2a2a 2a555555 }
\shade\path(12,2337)(1812,2337)(1812,537)
	(12,537)(12,2337)
\path(12,2337)(1812,2337)(1812,537)
	(12,537)(12,2337)
\shade\path(3612,2337)(5412,2337)(5412,537)
	(3612,537)(3612,2337)
\path(3612,2337)(5412,2337)(5412,537)
	(3612,537)(3612,2337)
\whiten\path(3612,1527)(3972,1887)(5052,1887)
	(5412,1617)(5412,1302)(4602,942)
	(4197,942)(3612,1347)(3612,1527)(3612,1527)
\path(3612,1527)(3972,1887)(5052,1887)
	(5412,1617)(5412,1302)(4602,942)
	(4197,942)(3612,1347)(3612,1527)(3612,1527)
\whiten\path(12,1437)(462,1887)(1362,1887)
	(1812,1437)(912,987)(687,987)
	(12,1437)(12,1437)
\path(12,1437)(462,1887)(1362,1887)
	(1812,1437)(912,987)(687,987)
	(12,1437)(12,1437)
\put(462,87){\makebox(0,0)[lb]{\smash{{\SetFigFont{11}{13.2}{\familydefault}{\mddefault}{\updefault}$\free_{\delta}(f,g)$}}}}
\put(3927,87){\makebox(0,0)[lb]{\smash{{\SetFigFont{11}{13.2}{\familydefault}{\mddefault}{\updefault}$\free_{\delta'}(f,g)$}}}}
\end{picture}
}}
 \caption{Cell $i,j$ in $\free_{\delta}(f,g)$ and $\free_{\delta}(f,g)$; with $\delta' > \delta$.}
\label{figure:NonBijective}
\end{figure}
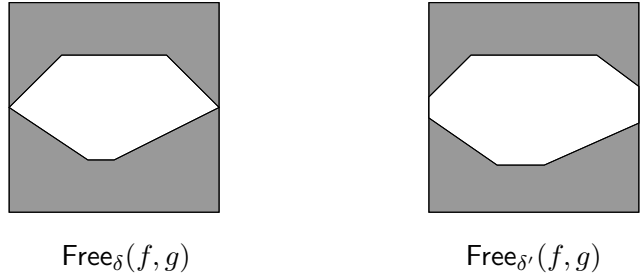

\smallskip\noindent\textbf{The non-bijective \frechet distance.}
Consider a variant of the \frechet distance in which we drop the requirement of the
reparametrizations $\alpha_f, \alpha_g$ to be \emph{strictly} increasing in 
Definition~\ref{def:frechet}; and instead only require them to be continuous and
non-decreasing.
This implies that an entire segment of the curve $f$ can be matched to a single
point of $g$ (and vice versa).
In the free space approach, it means that we can now consider continuous and
non-decreasing curves from  $(0,0)$ to $(m_f, m_g)$.
Let us denote this version of the distance as the non-bijective \frechet distance
\footnote{We use the superscript $^{\nbij}$ for entities relating to
the  non-bijective \frechet distance.}, and the strictly increasing version as
the bijective one.
For this version, we cannot have that there exist
non-decreasing monotone reparametrizations for every $\delta' > \delta^*$, but
not for $\delta^*$.
This is because the free space at the cell boundaries is always a closed interval
for every $\delta$.  
We can show that the endpoints of the intersections of
these intervals for opposite  cell boundaries are continuous functions for
$L_1, L_2, L_{\infty}, L_1^{\skoro}, L_2^{\skoro}, L_{\infty}^{\skoro}$ norms.
Finally, using compactness of $[\delta^*, \delta']$ and continuity, we can show that
the intersection of the opposite cell boundaries will be non-empty for $\delta^*$.
Thus, for the non-bijective variant of the \frechet distance, there exist reparametrizations
which achieve the \frechet distance.
This gives us the following result.

\begin{proposition}[\cite{AltG95}]
Let $f:[0, m_f] \rightarrow \reals^n$ and $g:[0, m_g] \rightarrow \reals^n$ be polygonal
curves.
 There exists a non-decreasing  monotone curve from $(0,0)$ to
$(m_f, m_g)$ in the free space diagram $\free_{\delta}(f,g)$ iff
$\fre^{\nbij}(f,g) \leq \delta$.\qed
\end{proposition}

The algorithm of the bijective \frechet distance decision problem
 can be easily modified for the non-bijective version.
The next lemma shows that the distance under the two semantics remains the same.

\begin{lemma}
\label{lemma:BothFrechetSame}
Let $f:[a_f, b_f] \rightarrow {\myV}$ and $g:[a_g, b_g] \rightarrow {\myV}$ be curves
such that $ a_f \neq b_f$ and $ a_g \neq b_g$.
We have $\fre^{\nbij}(f,g) = \fre(f,g)$.
\end{lemma}
\begin{proof}
Suppose $\fre^{\nbij}(f,g) \leq \delta$. 
We show $\fre(f,g) \leq \delta$ (the other direction is trivial).
Since $\fre^{\nbij}(f,g) \leq \delta$, given any $\epsilon > 0$, there exist
  reparametrizations
$\alpha_f^{\nbij}, \alpha_g^{\nbij}$ as in Definition~\ref{def:frechet} (for
non-bijective \frechet distance) that are continuous and non-decreasing such that
$\max_{0\leq \theta \leq 1} \left\lVert f\left(\alpha^{\nbij}_f(\theta)\right) - 
g\left(\alpha^{\nbij}_g(\theta)\right) \right\rVert \leq \delta + \epsilon$.
Since $f$ and $g$ are continuous, and each is defined over an interval that is
not a singleton point, given any $\epsilon' > 0$, 
 we can obtain parameterizations $\alpha_f, \alpha_g$ from 
$\alpha_f^{\nbij}, \alpha_g^{\nbij}$  which are continuous and \emph{strictly}
increasing such that 
$\max_{0\leq \theta \leq 1} \left\lVert f\left(\alpha_f(\theta)\right) - g\left(\alpha_g(\theta)\right) \right\rVert \leq \delta + \epsilon + \epsilon'$.
The result follows choosing $\epsilon,  \epsilon' \longrightarrow 0$.
\end{proof}

The  parameters $\creach{\mya}_{i,j}^q, \creach{\myb}_{i,j}^q$
for the non-bijective \frechet distance variant can be computed inductively as before.
This, together with the results of the present section  gives us the following result.

\begin{proposition}[The \frechet distance decision problem]
\label{proposition:ReachParamDecision}
Let $f:[0, m_f] \rightarrow \reals^n$ and $g:[0, m_g] \rightarrow \reals^n$ be polygonal
curves, and let $ \reals^n$ be equipped with the 
norm $\chi$.
Given $\delta \geq 0$,  there is an algorithm running in
 $O\left(m_f\vdot m_g \vdot H(\chi) \right)$ 
time which decides whether $\fre(f,g) \leq \delta$, where
$ H(\chi)$ is the time required to determine the parameters
$\mya_{i,j}^q, \myb_{i,j}^q$ for $0\leq q\leq 3$
 for a cell $i,j$ (\emph{i.e.} to compute the free space boundaries
for two lines), 
for   the norm  $\chi$.
The algorithm also decides whether there exist reparametrizations
$\alpha_f, \alpha_g$ such that
$\max_{0\leq \theta \leq 1} \left\lVert f\left(\alpha_f(\theta)\right) - g\left(\alpha_g(\theta)\right) \right\rVert \leq \delta$.
\end{proposition}
\begin{proof}
The existence of the algorithm, and its complexity
follows from  
Lemmas~\ref{lemma:FrechetDecisionParam},~\ref{lemma:BothFrechetSame},
and the dynamic programming algorithm for computing the parameters on free space
boundaries of cells.

For the second objective of the algorithm, 
we first determine whether $\fre^{\nbij}(f,g) \leq \delta$, and if so
check whether there exists a strictly increasing monotone curve from
from $(0,0)$ to
$(m_f, m_g)$ in $\free_{\delta}(f,g)$.
If there exists such a curve, then the optimal \frechet distance can be achieved using
strictly increasing reparametrizations $\alpha_f, \alpha_g$
such that 
$\max_{0\leq \theta \leq 1} \left\lVert f\left(\alpha_f(\theta)\right) - g\left(\alpha_g(\theta)\right) \right\rVert \leq \delta$.
If there does not exist such a curve, then the value $\delta$ cannot be achieved
by any strictly increasing reparametrizations $\alpha_f, \alpha_g$.

\end{proof}

\smallskip\noindent\textbf{Limiting the Reparametrizations using Windows.}
Given polygonal curves $f:[0, m_f] \rightarrow \reals^n$ and 
$g:[0, m_g] \rightarrow \reals^n$,
the \frechet reparametrizations allow portions of the line segment
$f_{[j]}$ to be matched to line segments $g_{[k]}$ for any $k\geq 0$.
An additional window 
requirement on the reparametrizations can be imposed 
which requires that the curve segment
$f_{[k]}$ be only be allowed to match the curves segments 
$g_{[\max(0, k-W)]}, g_{[\max(0, k-W+1)]}, \dots  g_{[\min(k+W, m_g)]}$, thus a window of
$W$ around each index, and $\abs{\alpha_f(\nu)-\alpha_g(\nu) }\leq W$ for all 
$\nu\in [0,1]$.
This means that only the free cells $i,j$  such that $\abs{i-j} \leq W $ are relevant.
There are at most $W\cdot \max(m_f, m_g)$ such cells.
The rationale behind the window requirement is that for practical applications, we
are often only interested in reparametrizations for which the maximal curve
parameter deviation is bounded by a constant.
The window requirement 
 can be used the bring down the complexity of the \frechet distance decision
problem to $O\left(W\vdot\max(m_f, m_g )\vdot H(\chi) \right)$,
where  $H(\chi)$ is as in Proposition~\ref{proposition:ReachParamDecision}.

\subsection{Algorithm for Determining the Value of 
the \frechet Distance}




In this section we use the non-bijective \frechet distance formulation.

\smallskip\noindent\textbf{Critical Values of $\delta$.}
The free space $\free^{\nbij}_{\delta}(f,g)$ keeps increasing as we increase
$\delta$.
As the free space gets bigger, new paths open up which make
a non-decreasing curve from $(0,0)$ to
$(m_f, m_g)$ in the free space feasible.
Consider the free space diagram.
We explore when new paths become feasible.
Based on the  sketch in~\cite{AltG95} for $\reals^2$, this happens at the following
values of $\delta$ for $\reals^n$.
\begin{enumerate}
\item Values of $\delta$ for which $ (0,0)$ and $(m_f, m_g)$ ``get'' into
the free space.
These two values of $\delta$  are $\lVert f(0) -g(0)\rVert$ and $\lVert f(m_f) -g(m_g)\rVert$.
\item 
 Values of $\delta$ which enable a monotone non-decreasing curve to enter a cell.
A curve can enter a cell $i,j$ if either $ \myb_{i,j}^3$ or   $ \myb_{i,j}^0$ is not equal
to $\bot$
(the case where a curve enters cell $i,j$  through the corner point $(i,j)$ is covered
by this condition).
Let $\delta$ be the least value which makes 
 $ \myb_{i,j}^3= (i, \Delta_j)\neq \bot$.
This is the least value of  $\delta$ for which we have $\myline( \myb_{i,j}^3, \mya_{i,j}^3)$
is non-empty.
Thus, this thee least value of  $\delta$ for which
the  \emph{point} $f(i)$ is at most $\delta$ away from
the \emph{line segment} $g_{[j]}$.
It follows that this value of $\delta$ is just the distance of the point
$f(i)$ from the line segment   $g_{[j]}$.
Similarly, the least value of $\delta$  which makes 
 $ \myb_{i,j}^0\neq \bot$  is  the distance of the point
$g(j)$ from the line segment   $f_{[i]}$.
Since there are $m_f\vdot m_g$ cells, there are  $2\vdot m_f\vdot m_g$ such critical
$\delta$ values.

\item 
Value of $\delta$ which enables a curve to go from cell $i,j$ to cell $k,j$ (for $ k>i$), 
that is a value
which makes the free space big enough so that at least a  horizontal line can
(possibly) go from  cell $i,j$ to cell $k,j$.
This happens when $\second(\mya_{i,j}^1) \leq  \second(\myb_{k,j}^3)$.
Observe that if  $\second(\mya_{i,j}^1) > \second(\myb_{k,j}^3)$, then even if the rest of the cells are fully free,
there cannot be a curve from cell  $i,j$ to cell $k,j$.
When $\second(\mya_{i,j}^1) $ becomes equal to, or greater than
$\second(\myb_{k,j}^3)$, by increasing $\delta$, it 
enables a curve to go from cell $i,j$ to cell $k,j$.
We can obtain the value of this special $\delta$ as follows.
For a point $s$, let $\ball(s, \delta)$ denote the set of points which are at most
$\delta$ away from $s$, formally
$\ball(s, \delta)= \set{q \mid \norm{q-s}\leq \delta}$.
We are interested in the least $\delta$ such that
there is some point on the line segment $ g_{[j]}$ that is at most $\delta$ away from
the point $f(i)$, and also from the point $f(k)$.
Mathematically, this is equivalen to finding the least $\delta$ such that
$\ball(f(i), \delta) \cap \ball(f(k), \delta)  \cap g_{[j]}$ is non-empty.
We prove in the next lemma that such a least $\delta$  exists.
These $\delta$ values are called \emph{horizontally clamped} $\delta$ values.

See Figure~\ref{figure:HClamp} for a horizontally clamped situation,
where
$\mya_{i,j}^1 = (i+1, \Delta)$, and $\myb_{k,j}^3 = (k, \Delta)$, \emph{i.e.}, both points have the same $\rho_g$ values.
Note that the only monotone non-decreasing  curve which can pass from cell $i,j$ to
cell $k,j$ must have a horizontal 
straight line segment from point $\mya_{i,j}^1$ to point
$\myb_{k,j}^3$.
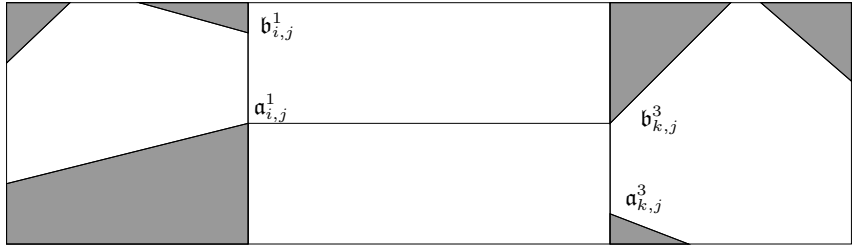
\begin{figure}[h]
\strut\centerline{\setlength{\unitlength}{0.00034996in}
\begingroup\makeatletter\ifx\SetFigFont\undefined%
\gdef\SetFigFont#1#2#3#4#5{%
  \reset@font\fontsize{#1}{#2pt}%
  \fontfamily{#3}\fontseries{#4}\fontshape{#5}%
  \selectfont}%
\fi\endgroup%
{\renewcommand{\dashlinestretch}{30}
\begin{picture}(12624,3639)(0,-10)
\path(3612,3612)(9012,3612)(9012,12)
	(3612,12)(3612,3612)
\path(9012,3612)(12612,3612)(12612,12)
	(9012,12)(9012,3612)
\path(12,3612)(3612,3612)(3612,12)
	(12,12)(12,3612)
\texture{55888888 88555555 5522a222 a2555555 55888888 88555555 552a2a2a 2a555555 
	55888888 88555555 55a222a2 22555555 55888888 88555555 552a2a2a 2a555555 
	55888888 88555555 5522a222 a2555555 55888888 88555555 552a2a2a 2a555555 
	55888888 88555555 55a222a2 22555555 55888888 88555555 552a2a2a 2a555555 }
\shade\path(3612,1812)(3612,12)(12,12)
	(12,912)(3612,1812)(3612,1812)
\path(3612,1812)(3612,12)(12,12)
	(12,912)(3612,1812)(3612,1812)
\shade\path(12,2712)(957,3612)(12,3612)
	(12,2712)(12,2712)
\path(12,2712)(957,3612)(12,3612)
	(12,2712)(12,2712)
\shade\path(3612,3162)(1992,3612)(3612,3612)
	(3612,3162)(3612,3162)
\path(3612,3162)(1992,3612)(3612,3612)
	(3612,3162)(3612,3162)
\shade\path(9012,1812)(9012,3612)(10812,3612)
	(9012,1812)(9012,1812)
\path(9012,1812)(9012,3612)(10812,3612)
	(9012,1812)(9012,1812)
\shade\path(11262,3612)(12612,2442)(12612,3612)
	(11262,3612)(11262,3612)
\path(11262,3612)(12612,2442)(12612,3612)
	(11262,3612)(11262,3612)
\shade\path(9012,462)(9012,12)(10182,12)
	(9012,462)(9012,462)
\path(9012,462)(9012,12)(10182,12)
	(9012,462)(9012,462)
\path(3612,1812)(9012,1812)
\put(9237,597){\makebox(0,0)[lb]{\smash{{\SetFigFont{9}{10.8}{\familydefault}{\mddefault}{\updefault}$\mya_{k,j}^3$}}}}
\put(9462,1767){\makebox(0,0)[lb]{\smash{{\SetFigFont{9}{10.8}{\familydefault}{\mddefault}{\updefault}$\myb_{k,j}^3$}}}}
\put(3792,3162){\makebox(0,0)[lb]{\smash{{\SetFigFont{9}{10.8}{\familydefault}{\mddefault}{\updefault}${\myb}_{i,j}^1$}}}}
\put(3702,1947){\makebox(0,0)[lb]{\smash{{\SetFigFont{9}{10.8}{\familydefault}{\mddefault}{\updefault}${\mya}_{i,j}^1$}}}}
\end{picture}
}}
 \caption{Horizontally clamped cells $(i,j)$ and  $(k,j)$ in $\free_{\delta}(f,g)$.}
\label{figure:HClamp}
\end{figure}

For each $0\leq i\leq m_f-1$ and  $0\leq j\leq m_g-1$,
 there are $m_f-1 -i$ such critical  $\delta$  values for cell $i,j$.
Thus, for each $j$, there are $\sum_{i=0}^{m_f-1} (m_f-1 -i)\  =\ (m_f-1)\vdot (m_f-2)/2 $
such critical $\delta$  values.
Hence overall there are $(m_g-1)\vdot (m_f-1)\vdot (m_f-2)/2 $ such values.

A similar analysis applies when we consider vertical lines, and in this case there are
$(m_f-1)\vdot (m_g-1)\vdot (m_g-2)/2 $ such critical \emph{vertically clamped} 
$\delta$ values.


\end{enumerate}

\begin{lemma}
Let $s_1, s_2, l_1,l_2$ be four points in the space $\reals^n$ with the norm
$L_{1}, L_2, L_{\infty},  L_{1}^{\skoro}, 
L_{2}^{\skoro}$, or $L_{\infty}^{\skoro}$.  
There exists $\delta^*$ such that
\[\delta^* = 
\min_{\delta \geq 0}\set{\delta \mid 
 \ball(s_1,\delta) \cap \ball(s_2, \delta) \cap \myline(l_1,l_2) \neq\emptyset}.\]
\end{lemma}
\begin{proof}
Consider any norm under consideration.
The equation of the line is $l_1 + \lambda\vdot (l_2-l_1)$ for $0\leq \lambda \leq 1$.
Define the function $h(\lambda) =
 \max \left(
\norm{l_1 + \lambda\vdot (l_2-l_1) -s_1}\, ,\,  \norm{l_1 + \lambda\vdot (l_2-l_1) -s_1}
\right) $.
In words, $\lambda$ gives a point on the line, and $h(\lambda)$ is the maximum
of the distances to points $s_1$ and $s_2$.
Observe that
\[
\inf_{\delta \geq 0}\set{\delta \mid 
 \ball(s_1,\delta) \cap \ball(s_2, \delta) \cap \myline(l_1,l_2) \neq\emptyset}\ =\ 
\inf\set{\delta\in  h([0,1])}.
\]
It can be shown that $h$ is continuous (and bounded) on $[0,1]$.
Since $[0,1]$ is compact, and $h$ is continuous, we have that $h([0,1])$ is compact, and
thus closed (and bounded). 
Thus, $h([0,1])$ contains the infimum $\inf h([0,1])$.
Thus, there is a point on the line $\myline(l_1,l_2)$
which is at most $\delta^*$ away from  $s_1$, and from
$s_2$.
This proves the statement of the lemma.
\end{proof}

To compute the least possible value of $\delta$ which makes
a non-decreasing curve from $(0,0)$ to
$(m_f, m_g)$ in the free space $\free_{\delta}(f,g)$ 
feasible, we find the least critical value amongst the
 $O(m_g\vdot m_f^2 + m_f\vdot m_g^2)$ values of $\delta$ for which there is such a
curve.
To do this, we sort the $O(m_g\vdot m_f^2 + m_f\vdot m_g^2) $ values and perform a binary
search using the decision procedure from the previous section.

\begin{theorem}[Computing the \frechet distance]
\label{theorem:FrechetAlgo}
Let $f:[0, m_f] \rightarrow \reals^n$ and $g:[0, m_g] \rightarrow \reals^n$ be polygonal
curves, and let $ \reals^n$ be equipped with the 
norm $\chi$.
 There is an algorithm running in time
 \[
O\Big(
\left(m_g\vdot m_f^2 + m_f\vdot m_g^2\right)\vdot
\big( P(\chi) + \log(m_g\vdot m_f) \big)\ 
+\ 
m_f\vdot m_g \vdot H(\chi) 
 \Big)
\]
which computes the  value $\fre(f,g)$
where
(i)~$ H(\chi)$ is the time required to determine the parameters
${\mya}_{i,j}^q, {\myb}_{i,j}^q$ for $0\leq q\leq 3$
for a cell $i,j$ (\emph{i.e.} to compute the free space boundaries
for two lines),
and
(ii)~$P(\chi)$ is the time required to compute a critical value of $\delta$ as outlined
previously
 for   the desired norm  $\chi$.\qed
\end{theorem}

\smallskip\noindent\textbf{The \frechet Distance with Windows.}
As for the decision problem, we can apply a window of $W$ in the
computation of the \frechet distance.
This allows us to restrict our attention to $W\cdot \max(m_f, m_g)$  cells.
For each of these cells, there are at most $W$ horizontal and $W$ vertical 
clamping $\delta$.
Thus, there are at most   $W^2\cdot \max(m_f, m_g)$  candidate values for
the \frechet distance.
We run a binary search on these, thus, we only have to run the
\frechet distance decision algorithm (with windows) at most
$O\left(\log\left(W\vdot\max(m_f, m_g)\right)\right)$ times.
Hence the complexity of the entire algorithm with windows is
$O\Big(
W^2\cdot M \cdot \big(P(\chi) + \log\left(W\vdot M \right)\big)+
W\cdot M\cdot \log\left(W\vdot M \right)\cdot H(\chi) 
\Big)
$, where $M= \max(m_f, m_g )$; and
$H(\chi),$ and $P(\chi)$ are as in Theorem~\ref{theorem:FrechetAlgo}.
If $W$ can be taken to be a constant, the complexity is
$O\Big(  M \cdot P(\chi)  + M\cdot \log(M)\cdot H(\chi) \Big)$.

\section{Computing the Geometric Primitives}
\label{section:Geometry}

This section is concerned with solving for the following two geometric primitives
for the $6$ norms $ L_1, L_2, L_{\infty}, L_1^{\skoro}, L_2^{\skoro}, L_{\infty}^{\skoro}$
in the space $\reals^n$.
As shown in the previous section, these primitves can be used to compute a set of
``critical'' value of $\delta$  which contains the \frechet distance value for polygonal
curves.
The last $3$ norms are required by the \frechet distance based algorithm of the previous
section for computing the Skorokhod distance between two linear interpolation
traces.
\begin{enumerate}
\item 
The distance of a point $s$ to a line $\myline(z,z')$.
\item 
Given four points $s_1, s_2, z,z'$, the least $\delta\geq 0$ such that 
$\ball(s_1,\delta)\cap \ball(s_2,\delta)\cap \myline(z,z')$ is non-empty.
\end{enumerate}

We present the solutions for the various norms.
The formal definition of the distance from a point to a set is given below.
\begin{definition}[Distance]
The distance of a point  $x$ to a set $ S$ in a metric space ${\myV}$  is defined to be 
$\inf_{y\in S} \dist_{\myV}(y,x)$, where $ \dist_{\myV}$ is the metric associated with the metric space
${\myV}$.
\end{definition}

\subsection{$L_1$-norm}

In this section, all norms are $L_1$-norms.

\begin{proposition}[Distance of point to line: $L_1$]
\label{proposition:DistancePointLineLwo}
The distance $\dist_{L_1}(s, \myline(z,z'))$
 from a point  $s\in \reals^n$ to the affine line segment between two
distinct  points $z$ and $z'$ in  $\reals^n$ is the solution of the following
linear program with $2\vdot n +1$ variables
$U_i^+ ,  U_i^-$ for $1\leq i\leq n$, and $\lambda$.
\begin{alignat*}{2}
     \textsf{minimize  }\  &  \sum_{i=1}^n (U_i^+  + U_i^-) \\
     \textsf{subject to }\  & U_i^+ - U_i^- 
     = z_i-s_i + \lambda\vdot (z'_i - z_i)  & \qquad &  \text{ for } 1\leq i \leq n\\
     & 0\leq \lambda \leq 1\\
     & U_i^+ \geq 0  & \qquad &  \text{ for } 1\leq i \leq n\\
     & U_i^- \geq 0  & \qquad &  \text{ for } 1\leq i \leq n
   \end{alignat*}
 \end{proposition}

\begin{proof}
The equation of the  affine line segment between two points $z$ and $z'$
is $z + \lambda\vdot v$ with  $0\leq \lambda  \leq 1$, where
$v = z'-z$.
By definition, we have 
\[
\dist_{L_1}(s, \myline(z,z')) =
 \inf_{0\leq \lambda  \leq 1} \sum_{i=1}^n \abs{z_i-s_i + \lambda\vdot v_i}\]
We transform 
the above into a standard linear program in two steps.
Let $U_i = z_i-s_i + \lambda\vdot v_i$ be $n$ new variables.
%
Then, we have $\dist_{L_1}(s, \myline(z,z')) $ to be a solution of
the following constraint problem:
\begin{alignat*}{2}
     \text{minimize  }\  &  \sum_{i=1}^n \abs{U_i} \\
     \text{subject to }\  &  U_i = z_i-s_i + \lambda\vdot v_i  & \qquad &  \text{ for } 1\leq i \leq n\\
     & 0\leq \lambda \leq 1
   \end{alignat*}
   
The absolute values in the objective function can be removed as follows 
(based on the sketch in~\cite{AIMMS}).
Let $U_i = U_i^+ - U_i^-$ such that  $U_i^+ \geq 0$ and  $U_i^- \geq 0$; and
 $\abs{U_i } =  U_i^+  + U_i^-$.
The previous constraint problem has the same solution as the linear program in
the statement of the lemma.
The proof of correctness of the  last transformation can be found
 in~\cite{AIMMS}.
\end{proof}

\smallskip\noindent\textbf{Computation of least $\delta$ such that 
$\ball(q, \delta) \cap \ball(q', \delta)  \cap \myline(z,z')$ is non-empty.}\\
The least $\delta$ can be computed as follows.
If $q=q'$, then the least  $\delta$ is simply $\dist_{L_1}(q, \myline(z,z'))$
by definition.
If $z=z'$, then the least  $\delta$ is
$\min(\norm{q-z},\, \norm{q'-z})$.

Consider the remaining case where $q\neq q'$ and  $z\neq z'$.
The equation of the  affine line segment between two points $z$ and $z'$
is $z + \lambda\vdot v$ with  $0\leq \lambda \leq 1$, where
$v = z'-z$.
We have the following optimization problem involving two variables
 $\delta, t$:
 \begin{alignat*}{2}
\text{minimize  }\  &   \delta \\
 \text{subject to } \ &  
 \sum_{i=1}^n \abs{q_i-(z_i + \lambda\vdot v_i)} \ \leq\ \delta\\
      & \sum_{i=1}^n \abs{q_i'-(z_i + \lambda\vdot v'_i)} \ \leq\ \delta\\
  &  0 \leq \lambda \leq 1\\
  & 0 \leq \delta
\end{alignat*}
We add $2\vdot n$ new variables
$U_i, U_i'$ for $1\leq i \leq n$ such that
$U_i = q_i-(z_i + \lambda\vdot v_i)$, and $U_i = q_i'-(z_i + \lambda\vdot v_i)$.
The above optimization problem can then be written as:
\begin{equation}
\label{LP:OrigBall}
\begin{alignedat}{2}
\text{minimize  }\  &   \delta \\
 \text{subject to } \ &  
  \sum_{i=1}^n \abs{U_i}  - \delta \ \leq\  0\\
  & \sum_{i=1}^n \abs{U'_i} -\delta  \ \leq\  0\\
  & U_i = q_i-(z_i + \lambda\vdot v_i)  & \qquad &  \text{ for } 1\leq i \leq n\\
  & U_i' = q'_i-(z_i + \lambda\vdot v_i)  & \qquad &  \text{ for } 1\leq i \leq n\\
  &  0 \leq \lambda \leq 1\\
  & 0 \leq \delta
\end{alignedat}
\end{equation}
We remove the absolute values in the constraints as follows.
We introduce another $2\vdot n$ new variables
$\absU_i, \absU_i'$ for $1\leq i \leq n$ such that
$\absU_i \geq U_i$ and $\absU_i \geq -U_i$ and similarly for
 $\absU_i'$.
Consider the  linear program:
\begin{equation}
\label{LP:Ball}
\begin{alignedat}{2}
\text{minimize  }\  &   \delta \\
 \text{subject to } \ &  
  \left(\sum_{i=1}^n \absU_i\right)  - \delta \ \leq\  0\\
  & \left(\sum_{i=1}^n \absU'_i \right)-\delta  \ \leq\  0\\
  & U_i = q_i-(z_i + \lambda\vdot v_i)  & \qquad &  \text{ for } 1\leq i \leq n\\
  & U_i' = q'_i-(z_i + \lambda\vdot v_i)  & \qquad &  \text{ for } 1\leq i \leq n\\
  & U_i - \absU_i \leq 0 & \qquad &  \text{ for } 1\leq i \leq n\\
   & -U_i - \absU_i \leq 0 & \qquad &  \text{ for } 1\leq i \leq n\\
  & U'_i - \absU'_i \leq 0 & \qquad &  \text{ for } 1\leq i \leq n\\
   & -U'_i - \absU'_i \leq 0 & \qquad &  \text{ for } 1\leq i \leq n\\
  &  \absU_i \geq 0 & \qquad &  \text{ for } 1\leq i \leq n\\
  &  \absU'_i \geq 0 & \qquad &  \text{ for } 1\leq i \leq n\\
   &  0 \leq \lambda \leq 1\\
  & 0 \leq \delta
\end{alignedat}
\end{equation}

Note that $\absU_i \geq U_i$ and $\absU_i \geq -U_i$ only ensure
$\absU_i \geq \abs{U_i}$, equality is \emph{not} guaranteed. 
However, the new constraint
$\sum_{i=1}^n \absU_i  - \delta \ \leq\  0$ 
ensures that $\absU_i$ can be taken to be 
$  |U_i|$ in the computation of the optimum in the linear
program.
This can be seen as follows. 
Let $\delta^{\dagger}$ be the value of the original constraint 
problem~\ref{LP:OrigBall}.
Suppose $\delta= \delta^*$ is the optimum value of the  linear program~\ref{LP:Ball}.
It can be seen that $\delta^* \leq \delta^{\dagger}$ as the feasible
region is bigger in the  linear program~\ref{LP:Ball}.
For this optimum $\delta^*$, suppose we have $\absU_k= \alpha$ such that 
$\alpha > \abs{U_k}$ for some $k$.
Since $ \absU_i \geq 0 $ for all $i$, we must have
$\delta^* > 0$ by the inequality 
$\left(\sum_{i=1}^n \absU_i  \right) - \delta^* \ \leq\  0$.
We have the following two cases.
\begin{compactenum}
 \item  Suppose $\left(\sum_{i=1}^n \absU'_i \right) - \delta^* \ <  0$.
Consider a new value of  $\absU_k$ which is equal to  $\abs{U_k}$.
Observe that all the constraints are still satisfied with this new
value of  $\absU_k$.
Also observe that we can decrease the value of $\delta$, from
$ \delta^*$ by a non-zero amount 
$\min\left( \left( \delta^* - (\sum_{i=1}^n \absU'_i \right),  \alpha- \abs{U_k}\right)$,
and still satisfy the constraints of the  linear program~\ref{LP:Ball}.
This is a contradiction as $ \delta^*$ was assumed to the optimal value of the program.

\item  Suppose $\left(\sum_{i=1}^n \absU'_i \right) - \delta^* \ =  0$
in this instantiation of the variables.
Then, if we decrease the value of $\absU_k$ from $ \alpha$ to
$ \abs{U_k}$, all the constraints are still satisfied,
in particular $\left(\sum_{i=1}^n \absU_i \right) - \delta^* \ \leq\  0$ still holds.
Thus,  we can set $\absU_k = \abs{U_k}$ without changing the
optimal value of the objective function.
Iterating over $i$, we see that we can set  $\absU_i = \abs{U_i}$ 
for all $i$ without changing the
optimal value of the objective function.

\end{compactenum}
Repeating the argument for the  $\absU'_i $ variables, we get that $\absU_i$ 
and $\absU'_i$
can be taken to be 
$  |U_i|$ and $  |U'_i|$ respectively in the computation of the optimum in the linear
program~\ref{LP:Ball}.

The results of the proceeding discussion are summarized in the following
proposition.
\begin{proposition}[Computation of least $\delta$ such that 
$\ball(q, \delta) \cap \ball(q', \delta)  \cap \myline(z,z')$ is non-empty: $L_1$]
\label{lemma:LOneBall}
Let $q, q', z, z'$ be points in $\reals^n$.
The value of
$\inf_{\ball(q, \delta) \cap \ball(q', \delta)  \cap \myline(z,z') \neq \emptyset} \delta $
for the $L_1$ norm 
is:
\[
\begin{cases}
\dist_{L_1}(q, \myline(z,z')) & \text{ if } q=q'\\
\min(\norm{q-z},\, \norm{q'-z}) & \text{ if } q\neq q' \text{ and } z=z'\\
\text{The value of the linear program~\ref{LP:Ball}} & \text{ otherwise.} 
\end{cases}
\]
\qed
\end{proposition}

\subsection{$L_2$-norm}

In this section, all norms are $L_2$-norms.
Let $\reals^n$ be the $n$-dimensional vector space over $\reals$ with the
$L_2$-norm.
We identify an $n$-tuple $x=(x_1,\dots, x_n)$ with $x_i\in \reals$ for $1\leq i\leq n$
with the corresponding vector $\vec{x} = (x_1,\dots, x_n)$ in the vector space 
$\reals^n$.
The vector equation of the  affine line segment between two points $z$ and $z'$
is $\vec{z} + \lambda\vdot (\vec{z'}-\vec{z})$ with $0\leq \lambda \leq 1$.
Given two vectors $\vec{x }, \vec{y}$, we denote their dot product by $\vec{x } \odot \vec{y }$.
Formally, for  $\vec{x} = (x_1,\dots, x_n)$ and  $\vec{y} = (y_1,\dots, y_)$ we have
the \emph{dot product} $\vec{x } \odot \vec{y }$ to be the scalar
$\displaystyle \sum_{i=1}^n x_i\vdot y_i$.
Given two vectors  $\vec{x}, \vec{y}$,  the angle $\theta$ between
them is  \emph{defined}  by the relation 
$\cos(\theta) \triangleq 
\frac{\vec{x } \odot \vec{y }}{\norm{\vec{x }}\vdot \norm{\vec{y }}}$.
Two vectors are said to be \emph{perpendicular} if the angle between them
is $90$ degrees, \emph{i.e.} if $\vec{x } \odot \vec{y } = 0$.
For the  basics of $L_2$ vector geometry, we refer the reader 
to~\cite{BloomBook,SloughterBook}.
When we want to emphasize vector operations such as the dot product, we 
use the vector notation, \emph{e.g.}, $\vec{x}$.

\begin{proposition}[Distance of point to line: $L_2$]
\label{proposition:DistancePointLineTwo}
The distance from a point  $s\in \reals^n$ to the affine line segment between two
distinct  points $z$ and $z'$ in  $\reals^n$
is 
\[\dist_{L_2}\left(s, \myline(z,z')\right) =
\begin{cases}
\norm{ \vec{z} - \vec{s} + 
\left(\frac{(\vec{z'}- \vec{z})  \odot(\vec{s} - \vec{z})}{\norm{ \vec{z'}- \vec{z} }^2}
\right)\vdot (\vec{z'}-\vec{z}) }& 
\text{ if } 0\leq
 \frac{(\vec{z'}- \vec{z})  \odot(\vec{s} - \vec{z})}{\norm{ \vec{z'}- \vec{z} }^2}
 \leq 1\\
\min\left( \norm{ \vec{z} - \vec{s}},\ \norm{ \vec{z'} - \vec{s}} \right) &
\text{ otherwise}
\end{cases}
\]
Moreover, the only point on the line which is $\dist_{L_2}\left(s, \myline(z,z')\right)$
away from $s$ is 
$z + \lambda_p\vdot (z'-z)$ with
\[
\lambda_p =  
\frac{(\vec{z'}- \vec{z})  \odot(\vec{s} - \vec{z})}{\norm{ \vec{z'}- \vec{z} }^2}
.\]

\end{proposition}
\begin{proof}
The vector equation of the  affine line segment between two points $z$ and $z'$
is $\vec{z} + \lambda\vdot (\vec{z'}-\vec{z})$ with $0\leq \lambda \leq 1$.
By letting $\lambda$ range over all reals, we get the equation of the infinite line 
passing through $z$ and $z'$.
Denoting $\vec{z'}-\vec{z} = \vec{v}$, the equation of the line
is $\vec{z} + \lambda\vdot \vec{v}$.

Suppose $s$ does not lie on the line
$\myline(z,z')$.
Let $\lambda_p$ be such that the vector $\vec{l_p}- \vec{s}= 
\vec{z} -\vec{s} + \lambda_p\vdot \vec{v}$
is perpendicular to the vector $\vec{v}$, \emph{i.e.}
$\vec{v}\odot (\vec{l_p} -  \vec{s}) = 0$.
Intuitively, $\vec{v}$ gives the direction of the line, and the vector from
the point $ s$ to  the point $l_p$ is perpendicular to the direction of the line.
The distance of the point $s$ from the infinite line is
$\norm{\vec{l_p} - \vec{s}}$ (see \emph{e.g.},~\cite{SloughterBook}).

If $0\leq \lambda_p\leq 1$ then the distance of $s$ to  the infinite line segment corresponds to the distance
of $s$ to the affine line segment between  $z$ and $z'$.
Otherwise,  if  $\lambda_p< 0$ or $\lambda_p > 1$, we claim 
 the distance is either $\norm{ \vec{z} - \vec{s}}$ or
 $\norm{ \vec{z'} - \vec{s}}$.
This can be seen as follows.
It can be easily shown that for any vector $\vec{l}$ on the line that
$\norm{\vec{l} - \vec{s}}^2 = \norm{\vec{l_p} - \vec{s}}^2  +
\norm{\vec{l} - \vec{l_p}}^2$.
Thus,
\begin{equation}
\label{equation:LTwoPointLine}
\norm{\vec{l} - \vec{s}}^2  = \norm{\vec{l_p} - \vec{s}}^2  +
(\lambda-\lambda_p)^2\vdot \norm{\vec{z'}-\vec{z}}^2.
\end{equation}
The minimum is achieved when $|\lambda-\lambda_p|$ 
is minimum,  thus, if $\lambda_p< 0$ or $\lambda_p > 1$,
the minimum on the line segment $\myline{z, z'}$ (\emph{i.e.}, 
constraining $\lambda$ to be such that $0\leq \lambda \leq 1$
is achieved at one of the affine line segment boundaries;
if $\lambda_p< 0$ then the point $z$ is closese to $s$, and if
if $\lambda_p> 1$ then the point $z'$ is closese to $s$.

The quantity $ \lambda_p$ can be computed as follows.
By definition we have
$\vec{v}\odot (
\vec{z} -\vec{s} + \lambda_p\vdot \vec{v}) = 0$.
Expanding the dot product,
$\vec{v}\odot (
\vec{z} -\vec{s})  + \lambda_p\vdot \norm{\vec{v}}^2 = 0$.
Thus, 
$ \lambda_p = \frac{\vec{v}\odot(\vec{s} - \vec{z})}{\norm{\vec{v}}^2}
=  \frac{(\vec{z'}- \vec{z})  \odot(\vec{s} - \vec{z})}{\norm{ \vec{z'}- \vec{z} }^2}$.

If $s$ lies on the line $\myline(z',z)$, or its infintite extension, then
 there is no point $l_p$ on the line such that 
 the vector $\vec{l_p}- \vec{s}= 
\vec{z} -\vec{s} + \lambda_p\vdot (\vec{z'}-\vec{z}) $
is perpendicular to the vector $\vec{z'}-\vec{z} $.
However, we show the $ \lambda_p$ value computed previously gives the right result.
Suppose $s = z + \lambda_s \vdot(z'-z)$.
The value of  $ \lambda_p$ will be correct if $l_p = s$, \emph{i.e.} if
 $ \lambda_p= \lambda_s$. 
We have 
$\lambda_p = 
\frac{(\vec{z'}- \vec{z})  \odot(\vec{s} - \vec{z})}{\norm{ \vec{z'}- \vec{z} }^2}
=
\frac{(\vec{z'}- \vec{z})  \odot\left( \vec{z} + \lambda_s \cdot(\vec{z'}-\vec{z}) - \vec{z}\right)}{\norm{ \vec{z'}- \vec{z} }^2}$.
Simplifying, we get $\lambda_s$.

Putting everything together, we have the desired result.
\end{proof}

\smallskip\noindent\textbf{Computation of the least $\delta$ such that 
$\ball(s_1, \delta) \cap \ball(s_2, \delta)  \cap \myline(z,z')$ is non-empty.}
We first show that the computation of the clamping values can be simplified for the $L_2$
norm.
Let $\sphere(s,\delta)$ denotes the set of points that are exactly $\delta$ distance
away from $s$.
We show that under certain cases,
\[
\min_{\delta \geq 0}\set{\delta \mid 
 \ball(s_1,\delta) \cap \ball(s_2, \delta) \cap \myline(z,z') \neq\emptyset}\ =\ 
\inf_{\delta \geq 0}\set{\delta \mid 
 \sphere(s_1,\delta) \cap \sphere(s_2, \delta) \cap \myline(z,z') \neq\emptyset}.
\]
The optimization over the sphere-intersections is an easier problem to
solve than the ball-intersections one.
The proof strategy for this change of optimization constraints is as follows.
Consider the set of points
$ \ball(s_1,\delta) \cap \myline(z,z') $; and $ \ball(s_2,\delta) \cap \myline(z,z') $.
Each set is a closed set of points on the line $ \myline(z,z')$.
In terms of the line parameter $\lambda$ (\emph{i.e}, where
a point on the line is $z+\lambda\vdot(z'-z)$), both sets can be represented as
closed subintervals of $[0,1]$ denoted $[\lambda_1, \lambda_1']$ and
$[\lambda_2, \lambda_2']$.
The least $\delta$ in the original optimization problem is the least $\delta$ such that
these two $\lambda$ sets intersect.
We show that under certain cases, the intersection point corresponds to the
boundary of the two balls, and thus, the value is the same as for the
 sphere-intersection optimization problem (intuitively, the intervals  $[\lambda_i, \lambda_i']$ expand in a strictly monotonic fashion).
Moreover, for the cases where this change cannot be made,
 the ball-intersection optimization problem has a charaterization which makes it
amenable to solve.
We present the formal proof next.

First, we need a technical lemma which characterizes
the set $ \ball(s,\delta) \cap \myline(z,z') $ for the norm $L_2$.
Given two points $z, z'$, let $\myline_{\infty}(z. z')$ denote the infinite straight
line passing throught the two points $z +\lambda \vdot (z'-z)$ for
$\lambda \in \reals$.

\begin{lemma}
\label{lemma:BallLineIntersect}
Let $s$ and $ z\neq z'$ be three  points in the space $\reals^n$ with the norm
$ z'$.
For 
$ \delta \geq \dist(s, \myline(z, z'))$ let
$h(\delta) = \set{\lambda \mid 0\leq \lambda \leq 1\text{ and }\dist(s, z +\lambda \vdot (z'-z)) \leq \delta}$.
Let $h_{\max}(\delta) = \sup h(\delta)$, and
$h_{\min}(\delta) = \inf h(\delta)$,
We have the following facts.
\begin{compactenum}
\item 
$h(\delta)$ is closed.
\item Over $[\dist\left(s, \myline(z, z')\right)\, , \, \dist(s, z')] $,
we have that (a)~$h_{\max}() $ is  continuous and strictly increasing, and
(b)~if $h_{\max}(\delta) = \lambda $, then 
$ \dist(s, z +\lambda \vdot (z'-z)) = \delta$.
\item Over $[\dist\left(s, \myline(z, z')\right)\, , \, \dist(s, z)] $,
we have that
(a)~$h_{\min}(\delta) $ is continuous and strictly decreasing, and
(b)~if $h_{\min}(\delta) = \lambda $, then 
$ \dist(s, z +\lambda \vdot (z'-z)) = \delta$.
\end{compactenum}
\end{lemma}
\begin{proof}
The fact that $h(\delta)$ is closed follows from the fact that
$h(\delta)$ is equal to the intersection of
$ \ball(s,\delta) $ and $\myline(z, z')$.
It can be easily seen that if $h(\delta) \neq \emptyset$, it is of the 
form $ [h_{\min}(\delta), h_{\max}(\delta')]$.

We now prove the second part.
Assume $ \dist(s, z') > \dist\left(s, \myline(z, z')\right)$ 
(otherwise the claim is easy to prove).
If $s$ does not lie on the line $\myline(z, z')$, 
we have for any $\lambda \in \reals$,
\begin{align*}
\dist(s, z +\lambda \vdot (z'-z))^2  & =
\norm{\vec{z}-\vec{s} + \lambda\vdot( \vec{z'}-\vec{z})}^2\\
& =  \norm{\vec{z}-\vec{s} + + \lambda_p\vdot( \vec{z'}-\vec{z})  \ +
 (\lambda-\lambda_p)\vdot( \vec{z'}-\vec{z})}^2
\text{ where  } \lambda_p \in \reals\text{ is such that }\\
& \hspace{3cm} \vec{z}-\vec{s} + + \lambda_p\vdot( \vec{z'}-\vec{z})  \text{ is perpendicular to }
 \vec{z'}-\vec{z}.\\
& =  \norm{\vec{z}-\vec{s} + \lambda_p\vdot( \vec{z'}-\vec{z}) }^2
\ + \  \abs{\lambda-\lambda_p}\vdot\norm{ \vec{z'}-\vec{z}}^2
\text{ by Equation~\ref{equation:LTwoPointLine}}\\
& = A +  \abs{\lambda-\lambda_p}^2\vdot B \text{ where } A, B 
\text{ are constants.}
\end{align*}
We comment that $\lambda_p$ need not necessarily be in $[0,1]$.
If $s = z+\lambda_s(z'-z)$ is on  the line $\myline(z, z')$,
we have $\dist(s, z +\lambda \vdot (z'-z)) = 
\abs{\lambda_s-\lambda}\vdot \norm{ z'- z}$, which is
equal to $\abs{\lambda_s-\lambda}\vdot\sqrt{B}$.
Letting $\lambda_p = \lambda_s$ in this case when $s = z+\lambda_s(z'-z)$,
we have that  $\dist(s, z +\lambda \vdot (z'-z))^2 =
 A + \abs{\lambda-\lambda_p}^2\vdot B $ ($A$ is $0$ in this case).

Thus, in all cases, for $\delta \geq \dist\left(s, \myline(z, z')\right)$
(which means $\delta \geq \sqrt{A}$ since $\dist\left(s, \myline(z, z')\right)\geq \sqrt{A}$),
 we have, we have
\[
h(\delta) = \set{\lambda \mid  
0\leq \lambda\leq 1 \text{ such that } 
\abs{\lambda-\lambda_p} \leq \sqrt{\frac{\delta^2-A}{B}}}.\]
Since  $ \dist(s, z') >  \dist\left(s, \myline(z, z')\right)$,
we have that $\lambda_p < 1$.
This is because we must have  by the proof of
Lemma~\ref{proposition:DistancePointLineTwo}
and Equation~\ref{equation:LTwoPointLine}, that either
$0\leq \lambda_p < 1$, or
$\dist\left(s, \myline(z, z')\right) =\dist(s, z') $ (recall that we
we assumed $\dist\left(s, \myline(z, z')\right) <\dist(s, z') $).
Hence, for $\delta \geq \dist\left(s, \myline(z, z')\right)$ we have,
\[
h_{\max}(\delta) = \min\left(1, \lambda_p + \sqrt{\frac{\delta^2-A}{B}}\right)\]
Suppose $ \dist(s, z') \geq   \delta  \geq \dist(s, \myline(z, z'))$.
We show $\lambda_p + \sqrt{\frac{\delta^2-A}{B}} \leq 1$ by contradiction.
If $\lambda_p + \sqrt{\frac{\delta-A}{B}} > 1$,
then  $\delta^2 > A+(1-\lambda_p)^2\vdot B = \dist(s, z') ^2$, \emph{i.e.}
$ \dist(s, z') < \delta$ which is a contradiction.
Thus, for $ \dist(s, z') \geq   \delta  \geq \dist(s, \myline(z, z'))$, we have 
\begin{equation}
\label{equation:HMax}
h_{\max}(\delta)  = \lambda_p + \sqrt{\frac{\delta^2-A}{B}}.
\end{equation}
Hence $h_{\max}(\delta)  $ is continuous and stictly increasing
over $[\dist\left(s, \myline(z, z')\right)\, , \, \dist(s, z')] $.

If $\delta = \dist\left(s, \myline(z, z')\right)$, then
$h(\delta)$ is just a single point (this can be inferred from the proof of
Proposition~\ref{proposition:DistancePointLineTwo}.
Thus, for this $\delta$, we have $ \dist(s, z +h_{\max}(\delta)\vdot (z'-z)) = \delta$.
If $ \dist(s, z') \geq   \delta  > \dist(s, \myline(z, z'))$,
then since $h_{\max}()  $ is  stictly increasing, for all 
$ \dist(s, z') \geq   \delta  > \delta' \geq \dist(s, \myline(z, z'))$, we have
 $h_{\max}(\delta') <   h_{\max}(\delta)$, which means that
we must have  $ \dist(s, z +h_{\max}(\delta)\vdot (z'-z)) = \delta$.
This concludes the proof for  $h_{\max}$.
The proof for $h_{\min}$ is similar.
\end{proof}

Using the previous lemma, we now show that the original  optimization
problem can be simplified using spheres instead of balls in contraiints.

\begin{lemma}
\label{lemma:BallSphere}
Let $s_1, s_2, z,z'$ be four points in the space $\reals^n$ with the norm
$ z'$.
Suppose $\min_{\delta \geq 0}\set{\delta \mid 
 \ball(s_1,\delta) \cap \ball(s_2, \delta) \cap \myline(z,z') \neq\emptyset}$
is not equal to either
$\dist(s_1,  \myline(z,z') )$ or $\dist(s_2,  \myline(z,z') )$.
Then, $\min_{\delta \geq 0}\set{\delta \mid 
 \sphere(s_1,\delta) \cap \sphere(s_2, \delta) \cap \myline(z,z') \neq\emptyset}$
exists, and
\[
\min_{\delta \geq 0}\set{\delta \mid 
 \ball(s_1,\delta) \cap \ball(s_2, \delta) \cap \myline(z,z') \neq\emptyset}\ =\ 
\min_{\delta \geq 0}\set{\delta \mid 
 \sphere(s_1,\delta) \cap \sphere(s_2, \delta) \cap \myline(z,z') \neq\emptyset}.
\]
Moreover, for the minimum $\delta$, there is only one point in the
intersection $\sphere(s_1,\delta) \cap \sphere(s_2, \delta) \cap \myline(z,z') $.
\end{lemma}
\begin{proof}
If $s_1 = s_2$, we have
$\min_{\delta \geq 0}\set{\delta \mid 
 \ball(s_1,\delta) \cap \ball(s_2, \delta) \cap \myline(z,z') \neq\emptyset}\ =
\dist\left(s_1,  \myline(z,z') \right)$, and this is equal to the right hand side
of the equality in the lemma.
If $z = z'$, then $\min_{\delta \geq 0}\set{\delta \mid 
 \ball(s_1,\delta) \cap \ball(s_2, \delta) \cap \myline(z,z') \neq\emptyset}$
is equal to
$\min_{\delta \geq 0}\set{\delta \mid z \in  \ball(s_1,\delta) \text{ and }
z \in  \ball(s_2,\delta) }$ which is equal to
$\max\left(\dist(s_1, z), \dist(s_2, z)\right)$, and this violates the assumptions 
of the lemma.
Thus, assume  $s_1\neq s_2$ and  $z \neq z'$.

Let $\delta^* = \min_{\delta \geq 0}\set{\delta \mid 
 \ball(s_1,\delta) \cap \ball(s_2, \delta) \cap \myline(z,z') \neq\emptyset}$.
We have $\delta^* > 0$ (a value of $0$ can arise only if $s_1 = s_2$) 
Let 
\[D_{\min} = \max\left(\dist(s_1,  \myline(z,z') ),\,  \dist(s_2,  \myline(z,z') )\right).
\]
Observe  that $\delta^*  \geq  D_{\min}$ by definition, and since $\delta^*$ is
not  equal to either
$\dist(s_1,  \myline(z,z') )$ or $\dist(s_2,  \myline(z,z') )$,  we must have
$\delta^*  > D_{\min}$.

Consider a  $\delta$ value $\delta \geq  D_{\min}$.
Consider the sets  $\ball(s_1, \delta) \cap \myline(z,z')$ 
and $\ball(s_2, \delta) \cap \myline(z,z')$.
These sets are the $h(\delta)$ sets of Lemma~\ref{lemma:BallLineIntersect}.
We denote them as $h^{s_1}(\delta)$ and $h^{s_2}(\delta)$ respectively.
Using the results from Lemma~\ref{lemma:BallLineIntersect},
we have that $h^{s_1}(\delta)$ is the closed
interval $ [h^{s_1}_{\min}(\delta), h^{s_1}_{\max}(\delta) ]$, and similarly for
$h^{s_2}(\delta)$.
Observe that 
\begin{equation}
\label{equation:BallSphere}
\ball(s_1,\delta) \cap \ball(s_2, \delta) \cap \myline(z,z') \neq\emptyset \quad
\text{ iff }\quad
[h^{s_1}_{\min}(\delta), h^{s_1}_{\max}(\delta) ]
\, \cap\, 
[h^{s_2}_{\min}(\delta), h^{s_2}_{\max}(\delta) ] \ \neq \emptyset
\end{equation}
Also note that by defintion, for all $\delta >D_{\min}$, we have
$[h^{s_1}_{\min}(\delta), h^{s_1}_{\max}(\delta) ] \neq \emptyset$; and also
$[h^{s_2}_{\min}(\delta), h^{s_2}_{\max}(\delta) ] \ \neq \emptyset$.

Observe that 
(a)~$[h^{s_1}_{\min}(D_{\min}), h^{s_1}_{\max}(D_{\min}) ] \ \neq \emptyset$; and
(b)~$[h^{s_2}_{\min}(D_{\min}), h^{s_2}_{\max}(D_{\min}) ] \ \neq \emptyset$; and
(c)~$[h^{s_1}_{\min}(D_{\min}), h^{s_1}_{\max}(D_{\min}) ]
\, \cap\, 
[h^{s_2}_{\min}(D_{\min}), h^{s_2}_{\max}(D_{\min}) ] \ =\emptyset$,
since $\delta^* >  D_{\min}    $.
Consider the case when $h^{s_1}_{\max}(D_{\min}) < h^{s_2}_{\min}(D_{\min})$ 
(the other case is symmetric),
\emph{i.e.}, the interval
$[h^{s_1}_{\min}(D_{\min}), h^{s_1}_{\max}(D_{\min}) ] $ lies to the left of the interval
$[h^{s_2}_{\min}(D_{\min}), h^{s_2}_{\max}(D_{\min}) ] $.
Let 
\[D_{\max} =  \min\left(\dist(s_1, z'),\,  \dist(s_2, z  ) \right).\]
We claim $\delta^*  \leq  D_{\max}$.
The proof is as follows.
By assumption $h^{s_1}_{\max}(D_{\min}) < h^{s_2}_{\min}(D_{\min})$.
This means that  $h^{s_1}_{\max}(D_{\min}) < 1$, and
$ h^{s_2}_{\min}(D_{\min}) > 0$.
We have the following two cases.
\begin{compactenum}
\item $\dist(s_1, z') \leq  \dist(s_2, z  ) $.
Thus $D_{\max} = \dist(s_1, z') $ and hence
$1\in h^{s_1}( D_{\max})$, and so $h^{s_1}_{\max}( D_{\max})=1$.
Using Lemma~\ref{lemma:BallLineIntersect}, we have that
$D_{\max} > D_{\min}$ since $h^{s_1}_{\max}(D_{\min}) < 1$.
\item $\dist(s_2, z  )     \leq \dist(s_1, z') $.
Thus  $D_{\max} = \dist(s_2, z  )  $
and hence we have $0\in h^{s_2}(  D_{\max})$, and so $0 =  h^{s_2}_{\min}(D_{\max})$.
Using Lemma~\ref{lemma:BallLineIntersect}, we have that
$D_{\max} > D_{\min}$  since $ h^{s_2}_{\min}(D_{\min}) > 0$.
\end{compactenum}
Thus, in both cases, both the intervals
$[h^{s_1}_{\min}(D_{\max}), h^{s_1}_{\max}(D_{\max}) ] $  and 
$[h^{s_2}_{\min}(D_{\max}), h^{s_2}_{\max}(D_{\max}) ] $ are nonempty;
and moreover either  $h^{s_1}_{\max}( D_{\max})=1$ or
$0 =  h^{s_2}_{\min}(D_{\max})$.
This means that the intersection of the intervals 
$[h^{s_1}_{\min}(D_{\max}), h^{s_1}_{\max}(D_{\max}) ] $ and
$[h^{s_2}_{\min}(D_{\max}), h^{s_2}_{\max}(D_{\max}) ] $ is non-empty.
Using Equation~\ref{equation:BallSphere}, we get that $\delta^*  \leq  D_{\max}$.
Thus,
\[ D_{\min} < \delta^* \leq D_{\max}.\]
Consider the functions 
$h^{s_1}_{\max}()$ and $h^{s_2}_{\min}()$ over the interval
$[D_{\min}, D_{\max}]$. 
The functions are defined on the interval by the conditions of 
Lemma~\ref{lemma:BallLineIntersect}.
$h^{s_1}_{\max}()$ is continuous and strictly increasing, and
$h^{s_2}_{\min}()$  is continuous and strictly decreasing.
At $D_{\min}$ we have $h^{s_1}_{\max}(D_{\min} )< h^{s_2}_{\min}(D_{\min})$,
and at $D_{\max}$ we have $h^{s_1}_{\max}(D_{\max} )< h^{s_2}_{\min}(D_{\max})$.
Thus these two functions must have the same value at some point in
$(D_{\min}, D_{\max}]$, and this point is where the curves of the functions intersect.
This intersection point is when $\delta = \delta^*$.
The graphs of $h_1$ and $h_2$ are depicted in Figure~\ref{figure:CurveBall}.
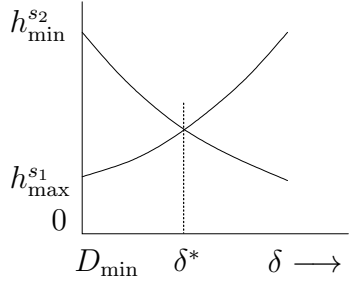
\begin{figure}[h]
\strut\centerline{\setlength{\unitlength}{0.00043745in}
\begingroup\makeatletter\ifx\SetFigFont\undefined%
\gdef\SetFigFont#1#2#3#4#5{%
  \reset@font\fontsize{#1}{#2pt}%
  \fontfamily{#3}\fontseries{#4}\fontshape{#5}%
  \selectfont}%
\fi\endgroup%
{\renewcommand{\dashlinestretch}{30}
\begin{picture}(4032,3349)(0,-10)
\path(870,3322)(870,577)(4020,577)
\dottedline{45}(2070,2122)(2070,585)
\path(870,2962)(872,2960)(875,2956)
	(882,2948)(892,2936)(906,2920)
	(923,2900)(944,2875)(967,2848)
	(993,2818)(1019,2787)(1047,2756)
	(1074,2724)(1101,2693)(1127,2663)
	(1152,2635)(1176,2608)(1199,2582)
	(1222,2557)(1243,2534)(1264,2512)
	(1284,2490)(1304,2469)(1324,2448)
	(1345,2428)(1365,2407)(1384,2388)
	(1404,2368)(1424,2349)(1445,2329)
	(1467,2308)(1489,2287)(1511,2266)
	(1535,2245)(1559,2223)(1583,2201)
	(1608,2179)(1634,2156)(1660,2134)
	(1686,2111)(1713,2089)(1739,2067)
	(1766,2045)(1792,2023)(1818,2002)
	(1844,1982)(1870,1962)(1895,1942)
	(1920,1923)(1945,1905)(1969,1887)
	(1993,1870)(2016,1853)(2040,1837)
	(2064,1821)(2087,1805)(2111,1789)
	(2136,1774)(2160,1758)(2185,1743)
	(2210,1727)(2236,1712)(2262,1696)
	(2289,1681)(2315,1666)(2342,1651)
	(2369,1636)(2396,1621)(2423,1607)
	(2449,1593)(2475,1579)(2501,1566)
	(2526,1553)(2551,1541)(2575,1528)
	(2599,1517)(2622,1506)(2645,1495)
	(2667,1484)(2688,1474)(2709,1464)
	(2730,1454)(2754,1443)(2778,1432)
	(2802,1422)(2826,1411)(2851,1400)
	(2876,1388)(2903,1376)(2932,1364)
	(2961,1351)(2993,1338)(3026,1323)
	(3059,1309)(3094,1294)(3129,1279)
	(3163,1265)(3195,1252)(3223,1239)
	(3248,1229)(3268,1221)(3282,1214)
	(3292,1210)(3298,1208)(3300,1207)
\path(870,1252)(873,1253)(879,1255)
	(890,1258)(907,1263)(930,1270)
	(958,1278)(991,1288)(1027,1299)
	(1065,1311)(1104,1323)(1143,1336)
	(1181,1348)(1218,1360)(1252,1371)
	(1285,1382)(1315,1392)(1344,1403)
	(1371,1413)(1397,1422)(1422,1432)
	(1446,1442)(1469,1452)(1493,1462)
	(1514,1472)(1535,1482)(1557,1492)
	(1579,1503)(1602,1514)(1624,1526)
	(1648,1538)(1671,1551)(1695,1564)
	(1720,1578)(1744,1592)(1769,1606)
	(1795,1621)(1820,1637)(1845,1652)
	(1870,1668)(1895,1684)(1920,1700)
	(1944,1716)(1968,1732)(1992,1748)
	(2015,1764)(2039,1780)(2062,1797)
	(2085,1813)(2108,1829)(2129,1845)
	(2151,1861)(2173,1877)(2195,1894)
	(2218,1912)(2242,1929)(2266,1948)
	(2290,1967)(2315,1986)(2340,2006)
	(2365,2027)(2391,2047)(2417,2068)
	(2443,2090)(2469,2111)(2494,2133)
	(2520,2154)(2545,2176)(2570,2197)
	(2594,2218)(2617,2239)(2640,2259)
	(2663,2279)(2685,2299)(2706,2319)
	(2727,2339)(2747,2358)(2768,2377)
	(2787,2396)(2807,2415)(2826,2435)
	(2846,2455)(2866,2475)(2886,2496)
	(2907,2519)(2929,2542)(2951,2566)
	(2975,2592)(2999,2619)(3024,2647)
	(3051,2677)(3078,2707)(3105,2738)
	(3133,2770)(3160,2801)(3186,2831)
	(3211,2859)(3233,2884)(3252,2906)
	(3268,2925)(3280,2939)(3290,2950)
	(3295,2957)(3299,2960)(3300,2962)
\put(510,622){\makebox(0,0)[lb]{\smash{{\SetFigFont{12}{14.4}{\familydefault}{\mddefault}{\updefault}$0$}}}}
\put(15,2962){\makebox(0,0)[lb]{\smash{{\SetFigFont{12}{14.4}{\familydefault}{\mddefault}{\updefault}$h_{\min}^{s_2}$}}}}
\put(15,1117){\makebox(0,0)[lb]{\smash{{\SetFigFont{12}{14.4}{\familydefault}{\mddefault}{\updefault}$h_{\max}^{s_1}$}}}}
\put(3067,127){\makebox(0,0)[lb]{\smash{{\SetFigFont{12}{14.4}{\familydefault}{\mddefault}{\updefault}$\delta \longrightarrow$}}}}
\put(1951,134){\makebox(0,0)[lb]{\smash{{\SetFigFont{12}{14.4}{\familydefault}{\mddefault}{\updefault}$\delta^*$}}}}
\put(780,135){\makebox(0,0)[lb]{\smash{{\SetFigFont{12}{14.4}{\familydefault}{\mddefault}{\updefault}$D_{\min}$}}}}
\end{picture}
}}
 \caption{Graphs of $h^{s_1}_{\max}$ and   $h^{s_2}_{\min}  $.}
\label{figure:CurveBall}
\end{figure}
Thus, we have $h^{s_1}_{\max}(\delta^*) = h^{s_2}_{\min}(\delta^*) = \lambda^*$ 
(using Lemma~\ref{lemma:BallLineIntersect}).
Which means that for the point
$l^* = z + \lambda^*\vdot(z'-z)$, we have
$\dist(s_1, l^*) = \delta^*$, and $\dist(s_2, l^*) = \delta^*$.
Thus 
\[\inf_{\delta \geq 0}\set{\delta \mid 
 \sphere(s_1,\delta) \cap \sphere(s_2, \delta) \cap \myline(z,z') \neq\emptyset}
\ \geq \delta^*\]
Since we also have
\[
\min_{\delta \geq 0}\set{\delta \mid 
 \ball(s_1,\delta) \cap \ball(s_2, \delta) \cap \myline(z,z') \neq\emptyset}\ \leq \ 
\inf_{\delta \geq 0}\set{\delta \mid 
 \sphere(s_1,\delta) \cap \sphere(s_2, \delta) \cap \myline(z,z') \neq\emptyset}
\]
and
$
\min_{\delta \geq 0}\set{\delta \mid 
 \ball(s_1,\delta) \cap \ball(s_2, \delta) \cap \myline(z,z') \neq\emptyset} = \delta^*$, 
we have that
$\min_{\delta \geq 0}\set{\delta \mid 
 \sphere(s_1,\delta) \cap \sphere(s_2, \delta) \cap \myline(z,z') \neq\emptyset}
$ exists, and is equal to $\delta^*$, which proves the first part of the lemma.
For the second part, we  note that the point corresponding to $\lambda^*$
is the only point on the line $\myline(z, z')$ which  is exactly $\delta^*$ away frp,
$s_1$ (or from $s_2$). 
This is because the distance of a point on the line given by $\lambda$ to a point $s$
(not on the line)
is $ \sqrt{A +  \abs{\lambda-\lambda_p}^2\vdot B}$ where $A,B$ and $\lambda_p$
are constants (see the proof of Lemma~\ref{lemma:BallLineIntersect});
and thus different points on the line have different distances from $s$.
\end{proof}

The following corollary states that in a horizontally clamped situation
for cells $(i,j)$ and $(k,j)$, we have only one unique horizontal montone non-decreasing
curve possible which can go from   cell $(i,j)$ to  $(k,j)$
(see Figure~\ref{figure:HClamp} for a horizontally clamped situation).
A similar result holds for vertically clamped situations.
\begin{corollary}
Given two polygonal curves
$f,g$, let $\delta^1\geq 0$ be such that in the free space   $\free_{\delta^1}(f,g)$, we have
$\second(\mya_{i,j}^1) > \second(\myb_{k,j}^3)$ for $k> i$; and 
$\delta^2\geq 0$ be such that in 
the free space   
$\free_{\delta^2}(f,g)$, we have
$\second(\mya_{i,j}^1) \leq \second(\myb_{k,j}^3)$.
Then there exists $\delta^1<\delta^*\leq \delta^2$, such that
 in the free space   
$\free_{\delta^*}(f,g)$, we have
$\second(\mya_{i,j}^1) =  \second(\myb_{k,j}^3)$.
\end{corollary}
\begin{proof}
For any $\delta \geq 0$, 
we have  $[\second(\mya_{i,j}^1) , \second(\myb_{i,j}^1) ] $ in
free space   
$\free_{\delta}(f,g)$
to be equal to $\ball(f[i], \delta) \cap g_{[j]}$. 
Using the proof of Lemma~\ref{lemma:BallSphere}, we have that there
is a $\delta^*$ for which $\second(\mya_{i,j}^1) =  \second(\myb_{k,j}^3)$.
\end{proof}




Using Lemma~\ref{lemma:BallSphere},
we now compute $\delta^* = \min_{\delta \geq 0}\set{\delta \mid 
 \ball(s_1,\delta) \cap \ball(s_2, \delta) \cap \myline(z,z') \neq\emptyset}$.
We restrict our attention to the cases when
$\delta^*$ is not equal to either
$\dist(s_1,  \myline(z,z') )$ or $\dist(s_2,  \myline(z,z') )$.
It can be checked from the proof of Lemma~\ref{lemma:BallSphere} that
this means that  $s_1, s_2, z,z'$ are all distinct points.
Also, we have  that $\delta^* = \min_{\delta \geq 0}\set{\delta \mid 
 \sphere(s_1,\delta) \cap \sphere(s_2, \delta) \cap \myline(z,z') \neq\emptyset}$.

Given any $\delta > 0$. 
If a point $z + \lambda\vdot(z'-z)$ on the line $\myline(z,z')$ is in
$\sphere(s_1, \delta) \cap \sphere(s_1, \delta)  \cap \myline(z,z')$, then we must have 
that it is exactly $\delta$ distance away from both $s_1$ and $s_2$.
Thus, for some $\delta\in [0,1]$,
\begin{equation} 
\label{equation:SphereCompute}
\norm{\vec{z} - \vec{s_1} + \lambda\vdot(\vec{z'}-\vec{z})}  = 
\norm{\vec{z} - \vec{s_2} + \lambda\vdot(\vec{z'}-\vec{z})}
\end{equation}
\begin{flalign*}
 &\text{Thus,}\   \sum_{k=1}^n\left(z_k-s_{1,k} + \lambda\vdot(z_k-z'_k)\right)^2
= \sum_{k=1}^n\left(z_k-s_{2,k} + \lambda\vdot(z_k-z'_k)\right)^2 &\\
&\text{Expanding, }\ 
 \sum_{k=1}^n\left( (z_k-s_{1,k})^2 + \lambda^2\vdot(z_k-z'_k)^2 +
\lambda\vdot 2\vdot  (z_k-s_{1,k}) \vdot (z_k-z'_k)\right)\  =  &\\
&  \hspace{6cm}\sum_{k=1}^n\left( (z_k-s_{2,k})^2 + \lambda^2\vdot(z_k-z'_k)^2 +
\lambda\vdot 2\vdot  (z_k-s_{2,k}) \vdot (z_k-z'_k)\right)&\\
&
\text{Cancelling out the common term and rearranging, }& \\
&\lambda\vdot 2\vdot \left(\sum_{k=1}^n(z_k-s_{1,k}) \vdot (z_k-z'_k)
- \sum_{k=1}^n (z_k-s_{2,k}) \vdot (z_k-z'_k) \right)\ = \ 
  \sum_{k=1}^n (z_k-s_{2,k})^2 -      \sum_{k=1}^n (z_k-s_{1,k})^2&  \\
&\text{Grouping } k\text{-terms together},  &\\
 &\lambda\vdot 2\vdot \sum_{k=1}^n \Big( (z_k-s_{1,k}) \vdot (z_k-z'_k) -
(z_k-s_{2,k}) \vdot (z_k-z'_k) \Big) = 
 \sum_{k=1}^n\Big(  (z_k-s_{2,k})^2 -  (z_k-s_{1,k})^2 \Big) &\\
& \text{Simplifying, }\ 
\lambda\vdot 2\vdot \sum_{k=1}^n (s_{2,k}- s_{1,k})\vdot (z_k-z'_k) \ = \ 
 \sum_{k=1}^n\Big( s_{2,k}^2 - s_{1,k}^2 + 2\vdot  s_{1,k}\vdot z_k -
 2\vdot s_{2,k}\vdot z_k \Big) &\\
& \text{Simplifying again, }\ 
\lambda\vdot 2\vdot \sum_{k=1}^n (s_{2,k}- s_{1,k})\vdot (z_k-z'_k) \ = \ 
 \sum_{k=1}^n\Big( s_{2,k}^2 - s_{1,k}^2  - 2\vdot z_k\vdot(s_{2,k}- s_{1,k})
 \Big)& \\
& \text{A final simplication gives us: }\ 
\lambda\vdot 2\vdot \sum_{k=1}^n (s_{2,k}- s_{1,k})\vdot (z_k-z'_k) \ = \ 
\sum_{k=1}^n (s_{2,k}- s_{1,k})\vdot  (s_{2,k} + s_{1,k} - 2\vdot  z_k)&
\end{flalign*}

If $\sum_{k=1}^n (s_{2,k}- s_{1,k})\vdot (z_k-z'_k) $ is equal to $0$, then
the value of the LHS in the last equation remains the same for all values of $\lambda$,
Thus, all vaues of $\lambda$ satisfy Equation~\ref{equation:SphereCompute}.
In particular $\min\Big(\dist(s_1,  \myline(z,z') )\, , \, \dist(s_2,  \myline(z,z') )\Big)$.
would satisfy Equation~\ref{equation:SphereCompute}, and thus
$\delta^* $ would be either either
$\dist(s_1,  \myline(z,z') )$ or $\dist(s_2,  \myline(z,z') )$ 
(recall that  $\min\Big(\dist(s_1,  \myline(z,z') )\, , \, \dist(s_2,  \myline(z,z') )\Big)$
is a lower bound on $\delta^*$); violating our intial assumption.
Thus, $\sum_{k=1}^n (s_{2,k}- s_{1,k})\vdot (z_k-z'_k) $ is not equal to $0$ given
the assumptions, and then,
\begin{equation}
\label{equation:OptLambda}
\lambda = \frac{\sum_{k=1}^n (s_{2,k}- s_{1,k})\vdot  (s_{2,k} + s_{1,k} - 2\vdot  z_k)}
{2\vdot \sum_{k=1}^n (s_{2,k}- s_{1,k})\vdot (z_k-z'_k) }
\end{equation}
Since this is the only value of $\lambda$ which satisfies
Equation~\ref{equation:SphereCompute},
we must have
$
 \min_{\delta \geq 0}\set{\delta \mid 
 \sphere(s_1,\delta) \cap \sphere(s_2, \delta) \cap \myline(z,z') \neq\emptyset}$
to be $\norm{\vec{z} - \vec{s_2} + \lambda\vdot(\vec{z'}-\vec{z})} $ for
$\lambda$ given by Equation~\ref{equation:OptLambda}.

The following 
function~\ref{function:BallIntersect} 
combines everything together giving us 
Proposition~\ref{proposition:BallIntersectFinalLTwo}

\begin{function}
  \SetKwInOut{Input}{Input}
  \SetKwInOut{Output}{Output}
  \Input{Points $s_1, s_2. z, z'$ in $\reals^n$}
  \Output{$\delta^*$ where 
    $\delta^* = \min_{\delta \geq 0}\set{\delta \mid 
      \ball(s_1,\delta) \cap \ball(s_2, \delta) \cap \myline(z,z') \neq\emptyset}$}
  \Switch{ $s_1, s_2. z, z'$}{
    \lCase{ $z= z'$ }{
      \Return{$\max\left(\dist(s_1, z)\, ,\, \dist(s_2, z)\right)$\;}
    }
    \lCase{ $s_1= s_2$ and $z\neq z'$}{
      \Return{$\dist\left(s_1, \myline(z, z')\right)$\;}
    }
  \uCase{$s_1\neq s_2$ and $z\neq z'$}{
    $ \DSet:=\emptyset$\;
    \Begin( \tcp*[f]{Check if $ \delta= \dist\left(s_1, \myline(z, z')\right)$
      is such that})
    {\hspace{4cm}\tcp*[f]{$\ball(s_1,\delta) \cap \ball(s_2, \delta) \cap \myline(z,z') \neq\emptyset$}
      $D_1 := \dist\left(s_1, \myline(z, z')\right)$\;
      $\lambda_1 :=  
      \frac{(\vec{z'}- \vec{z})  \odot(\vec{s_1} - \vec{z})}{\norm{ \vec{z}- \vec{z'} }_{L_2}^2}$ 
      \tcp*[r]{This value of  $\lambda_1$ is such that the point}
      \hspace{4cm}\tcp*[f]{
        $z + \lambda_1\vdot(z' - z)$ is $ \dist\left(s_1, \myline_{\infty}(z, z')\right)$}\\
      \hspace{2cm}\tcp*[f]{distance away from  $s_1$}\\
      \lIf{$\lambda_1 < 0$}{$\lambda_1 :=0$
        \tcp*[r]{$z$ is the closest point on  $ \myline(z, z')$ to $s_1$}
      }
      \lElseIf{$\lambda_1 > 1$}{$\lambda_1 :=1$
        \tcp*[r]{$z'$ is the closest point on  $ \myline(z, z')$ to $s_1$}
      }
      \lIf{$\dist\left(s_2, z+ \lambda_1\vdot(z'-z) \right) \leq D_1$}{
        $ \DSet:=\set{D_1}$\;
      }
    }
    \Begin( \tcp*[f]{Check if $ \delta= \dist\left(s_2, \myline(z, z')\right)$
      is such that})
    {\hspace{4cm}\tcp*[f]{$\ball(s_1,\delta) \cap \ball(s_2, \delta) \cap \myline(z,z') \neq\emptyset$}
      $D_2 := \dist\left(s_2, \myline(z, z')\right)$\;
      $\lambda_2 :=  
      \frac{(\vec{z'}- \vec{z})  \odot(\vec{s_2} - \vec{z})}{\norm{ \vec{z}- \vec{z'} }_{L_2}^2}$ 
      \tcp*[r]{This value of  $\lambda_2$ is such that the point}
      \hspace{2cm}\tcp*[f]{
        $z + \lambda_2\vdot(z' - z)$ is $ \dist\left(s_2, \myline_{\infty}(z, z')\right)$}\\
      \hspace{2cm}\tcp*[f]{distance away from  $s_2$}\\
      \lIf{$\lambda_2 < 0$}{$\lambda_2 :=0$
        \tcp*[r]{$z'$ is the closest point on  $ \myline(z, z')$ to $s_2$}
      }
      \lElseIf{$\lambda_2 > 1$}{$\lambda_2 :=1$
        \tcp*[r]{$z'$ is the closest point on  $ \myline(z, z')$ to $s_2$}
      }
      \lIf{$\dist\left(s_2, z+ \lambda_2\vdot(z'-z) \right) \leq D_2$}{
        $ \DSet :=\DSet \cup \set{D_2}$\;
      }
    }
    \lIf{$\DSet\neq\emptyset$}{\Return{$\min\left(\DSet\right)$\;}}
    \lElse{\Return{
        \[\norm{\vec{z}-\vec{s_1} + \left(
            \frac{\sum_{k=1}^n (s_{2,k}- s_{1,k})\cdot  (s_{2,k} + s_{1,k} - 2\vdot  z_k)}
            {2\cdot \sum_{k=1}^n (s_{2,k}- s_{1,k})\cdot (z_k-z'_k) }
          \right)\cdot (\vec{z'}-\vec{z})}_{L_2}\;
        \]}
    }
}
}
\caption{$\Phi_{L_2}$($s_1, s_2, z, z'$)}
  \label{function:BallIntersect}
\end{function}

\begin{proposition}[Computation of least $\delta$ such that 
$\ball(q, \delta) \cap \ball(q', \delta)  \cap \myline(z,z')$ is non-empty: $L_2$]
\label{proposition:BallIntersectFinalLTwo}
Given points $s_1, s_2. z, z'$ in $\reals^n$,
Function~\ref{function:BallIntersect} computes
$\min_{\delta \geq 0}\set{\delta \mid 
      \ball(s_1,\delta) \cap \ball(s_2, \delta) \cap \myline(z,z') \neq\emptyset}$
for the $L_2$ norm.
\qed
\end{proposition}

\subsection{$L_{\infty}$-norm}

We note that the $L_{\infty}^{\skoro}$ norm is equivalent to
the  $L_{\infty}$ norm:  $\norm{\tuple{\tuple{s_1, \dots, s_n},t}}_ {L_{\infty}^{\skoro}}= 
\norm{\tuple{s_1,\dots, s_n, t}}_ {L_{\infty}}$ for $s_i, t\in \reals$.
Thus, the results in this section are also applicable to the  $L_{\infty}^{\skoro}$ norm.

\begin{proposition}[Distance of point to line: $L_{\infty}$]
\label{proposition:DistancePointLineInfty}
The distance from a point  $s\in \reals^n$ to the affine line segment between two
distinct  points $z$ and $z'$ in  $\reals^n$ denoted
$\dist_{L_{\infty}}(s, \myline(z,z')) $ is a solution of the following linear program.
\begin{alignat*}{2} 
     \textsf{minimize  }\ & \delta \\ 
     \textsf{subject to }\  &
     z_i-s_i + \lambda\vdot (z'_i - z_i) \leq \delta & \qquad &  \text{ for } 1\leq i \leq n\\ 
    &  -\left(z_i-s_i + \lambda\vdot (z'_i - z_i) \right) \leq \delta 
                   & \qquad &  \text{ for } 1\leq i \leq n\\
     & \delta \geq 0\\
     & 0 \leq \lambda  \leq 1
   \end{alignat*}
 \end{proposition}
 \begin{proof}
The  equation of the  affine line segment between two points $z$ and $z'$
is $z + \lambda\vdot ( z'-z)$ with  $0\leq t\leq 1$.
By definition, we have 
\[
\dist_{L_{\infty}}(s, \myline(z,z')) =
 \inf_{0\leq t \leq 1} \max\set{\abs{z_i-s_i + \lambda\vdot(z'_i - z_i)} \mid 1\leq i\leq n}\]
This can be written as the system of linear constraints in the statement of the lemma,
using the fact that $\abs{A} \leq B $ iff $A \leq B$ and $-A\leq B$ for real numbers
$A,B$; and that $\max\set{A_1,\dots, A_n}$ is by definition the least $\delta$ such that
$A_i \leq \delta$ for every $i$.
\end{proof}

\begin{proposition}[Computation of least $\delta$ such that 
$\ball(q, \delta) \cap \ball(q', \delta)  \cap \myline(z,z')$ is non-empty: $L_{\infty}$]
\label{proposition:BallsIntersectLInfty}
Let $q, q', z, z'$ be points in $\reals^n$.
The value of
$\inf_{\ball(q, \delta) \cap \ball(q', \delta)  \cap \myline(z,z') \neq \emptyset} \delta $
for the $L_{\infty}$ norm 
is
\[
\begin{cases}
\dist_{L_{\infty}}(q, \myline(z,z')) & \text{ if } q=q'\\
\min(\norm{q-z},\, \norm{q'-z}) & \text{ if } q\neq q' \text{ and } z=z'\\
\text{The value of the following  linear program~\ref{LP:InftyBall}} & \text{ otherwise.}\\
\end{cases}
\]
\begin{equation}
\label{LP:InftyBall}
\begin{alignedat}{2}
  \text{minimize  }\  &   \delta \\
  \text{subject to } \ &
  z_i-q_i + \lambda\vdot (z'_i - z_i) \leq \delta & \qquad &  \text{ for } 1\leq i \leq n\\ 
  &  -\left(z_i-q_i + \lambda\vdot (z'_i - z_i) \right) \leq \delta 
  & \qquad &  \text{ for } 1\leq i \leq n\\
  &z_i-q'_i + \lambda\vdot (z'_i - z_i) \leq \delta & \qquad &  \text{ for } 1\leq i \leq n\\ 
  &  -\left(z_i-q'_i + \lambda\vdot (z'_i - z_i) \right) \leq \delta 
  & \qquad &  \text{ for } 1\leq i \leq n\\
  & \delta \geq 0\\
  & 0 \leq \lambda  \leq 1
\end{alignedat}
\end{equation}
\end{proposition}
\begin{proof}
We consider the case when $q\neq q'$ and $z\neq z'$ (the other cases are
as in Lemma~\ref{lemma:LOneBall}).
We want to find the least $\delta$ such that 
there exists some point on $\myline(z,z')$ such that that point is at most $\delta$ away 
from both $q$ and $q'$.
This can be written down as the constraint system~\ref{LP:InftyBall}, using similar 
reasonsing as in the proof of Lemma~\ref{proposition:DistancePointLineInfty}.
\end{proof}

\subsection{$L_1^{\skoro}$-norm}

In this section, we compute the geometric primitives for
the $L_2^{\skoro}$-norm defined by  $\norm{\tuple{s,t}}_ {L_2^{\skoro}}= 
\max\left(\norm{s}_{L_2}, \abs{t}\right)$ for $\tuple{s,t}\in \reals^n\times \reals$.
We note the following identity that we will use:
\begin{equation}
\label{equation:MaxSkoroLOne}
\left(\norm{\tuple{s,t}}_ {L_1^{\skoro}} \leq \delta\right)
\ \text{ iff } \ 
\left(\norm{s}_{L_1} \leq \delta \text{ and } \abs{t} \leq \delta\right)
\end{equation}

\smallskip\noindent\textbf{Distance from a point to a Line.}
Let us be given the points $\tuple{s,t_s}, \tuple{z,t_z}, \tuple{z',t_{z'}}$
with $s,z,z'\in \reals^n$ and $t\in \reals$.
For the $L_1^{\skoro}$-norm, the distance
$\dist_{L_1^{\skoro}}\left(\tuple{s,t_s}, \myline(\tuple{z,t_z}, \tuple{z,t_{z'}})\right)$
is equal to the following 
\begin{align*} 
\dist_{L_1^{\skoro}}\left(\tuple{s,t_s}, \myline(\tuple{z,t_z}, \tuple{z',t_{z'}})\right)
& = \inf_{0\leq \lambda  \leq 1} \norm{\tuple{z,t_z} - \tuple{s,t_s}
+ \lambda \vdot ( \tuple{z',t_{z'}} - \tuple{z,t_z})}_{L_1^{\skoro}}\\
& =  \inf_{0\leq \lambda  \leq 1} \left(\max\left(
\norm{z-s + \lambda \vdot(z'-z)}_{L_1}\, , \, 
\abs{t_z-t_s + \lambda \vdot(t_{z'} - t_{z})} \right) \right)\\
& = \inf_{0\leq \lambda  \leq 1} \left\{
  Z \left| 
    \begin{array}{l}
      Z \geq \norm{z-s + \lambda \vdot(z'-z)}_{L_1}\text{ and }\\
      Z \geq \abs{t_z-t_s + \lambda \vdot(t_{z'} - t_{z})}
    \end{array}
    \right.
    \right\}
\end{align*}
This can be written as the following optimization problem.
\begin{alignat*}{2}
  \text{minimize  }\  &\  Z\\
  \text{subject to }\  & \left( \sum_{i=1}^n \abs{z_i-s_i + \lambda\vdot (z'-z)} \right) -Z\leq 0\\
  & \left( \abs{t_z-t_s + \lambda \vdot(t_{z'} - t_{z})}\right) -Z  \leq 0\\
  & 0\leq \lambda \leq 1\\
  & Z \geq 0 
   \end{alignat*}
The above optimization problem can then be written as:
\begin{alignat*}{2}
  \text{minimize  }\  &\  Z\\
  \text{subject to }\  & \left( \sum_{i=1}^n \abs{U_i} \right) -Z\leq 0\\
  & \left( \abs{U_t}\right) -Z  \leq 0\\
  & U_i = z_i-s_i + \lambda\vdot (z'-z)  & \qquad &  \text{ for } 1\leq i \leq n\\
  & U_t = t_z-t_s + \lambda \vdot(t_{z'} - t_{z})\\
  & 0\leq \lambda \leq 1\\
  & Z \geq 0
   \end{alignat*}

  We remove the absolute values in the constraints as follows.
We introduce $n+1$ new variables
$\absU_i$ for $1\leq i \leq n$, and $\absU_t$ such that
$\absU_i \geq U_i$ and $\absU_i \geq -U_i$ and similarly for
 $\absU_t$.
Consider the linear program
\begin{equation}
\label{LP:PointLineLOneSkoro}
\begin{alignedat}{2}
  \text{minimize  }\  &\  Z\\
 \text{subject to } \ &  
  \left(\sum_{i=1}^n \absU_i\right)  - Z\ \leq\  0\\
  & \absU -Z \leq 0\\
   & U_i - \absU_i \leq 0 & \qquad &  \text{ for } 1\leq i \leq n\\
   & -U_i - \absU_i \leq 0 & \qquad &  \text{ for } 1\leq i \leq n\\
   & U_t - \absU_t \leq 0\\
   & -U_t - \absU_t \leq 0\\
   &  \absU_i \geq 0 & \qquad &  \text{ for } 1\leq i \leq n\\
   & \absU_t \geq 0 \\
  & U_i = z_i-s_i + \lambda\vdot (z'-z)  & \qquad &  \text{ for } 1\leq i \leq n\\
  & U_t = t_z-t_s + \lambda \vdot(t_{z'} - t_{z})\\
   & 0\leq \lambda \leq 1\\
   & Z \geq 0
\end{alignedat}
\end{equation}

Note that $\absU_i \geq U_i$ and $\absU_i \geq -U_i$ only ensure
$\absU_i \geq \abs{U_i}$, equality is \emph{not} guaranteed. 
However, the new constraint
$\sum_{i=1}^n \absU_i  - Z \ \leq\  0$ 
ensures that $\absU_i$ can be taken to be 
$  |U_i|$ in the computation of the optimum in the linear
program.
The proof of correctness is as for the proof of correctness of the linear 
program~\ref{LP:Ball}.

The results of the proceeding discussion are summarized in the following
proposition.
\begin{proposition}[Distance of point to Line: $L_1^{\skoro}$]
\label{proposition:PointLineLOneSkoro}
Let $\tuple{s,t_s}, \tuple{z,t_z}, \tuple{z',t_{z'}}$ be points
with $s,z,z'\in \reals^n$ and $t\in \reals$.
For the $L_1^{\skoro}$-norm, the distance of the
point $\tuple{s,t_s}$ from the line
$\myline( \tuple{z,t_z}, \tuple{z',t_{z'}})$ is the
value of the linear program~\ref{LP:PointLineLOneSkoro}.
\qed
\end{proposition}

\smallskip\noindent\textbf{Computation of the least $\delta$ such that 
$\ball(\tuple{s_1,t_{s_1}} \delta) \cap \ball(\tuple{s_2,t_{s_1}}, \delta)  
\cap \myline( \tuple{z,t_z}, \tuple{z',t_{z'}})$
is non-empty.}\\
The value of the least $\delta$ is specified by the following
optimization problem.
\begin{alignat*}{2}
\text{minimize  }\  &   \delta \\ 
\text{subject to } \ & \norm{\tuple{z,t_z} - \tuple{s_1,t_{s_1}} +
  \lambda \vdot (\tuple{z',t_{z'}} -  \tuple{z,t_z})}_{L_1^{\skoro}} \leq \delta\\
& \norm{\tuple{z,t_z} - \tuple{s_2,t_{s_2}} +
  \lambda \vdot (\tuple{z',t_{z'}} -  \tuple{z,t_z})}_{L_1^{\skoro}} \leq \delta\\
 &  0 \leq \lambda \leq 1\\
  & 0 \leq \delta
\end{alignat*}

Expanding the $L_1^{\skoro}$ norm, and using Equation~\ref{equation:MaxSkoroLOne},
\begin{alignat*}{2}
\text{minimize  }\  &   \delta \\ 
\text{subject to } \ & \left(\sum_{i=1}^n\abs{z_i - s_{1,i}  +
  \lambda \vdot (z_i' -  z_i)} \right)- \delta \leq 0\\
& \abs{t_z-t_{s_1} + \lambda \vdot(t_{z'} - t_{z})}  -\delta \leq 0\\
&\left( \sum_{i=1}^n\abs{z_i - s_{2,i}  +
  \lambda \vdot (z_i' -  z_i)} \right) - \delta \leq 0\\
& \abs{t_z-t_{s_2} + \lambda \vdot(t_{z'} - t_{z})} -\delta \leq 0\\
 &  0 \leq \lambda \leq 1\\
  & 0 \leq \delta
\end{alignat*}

As was done in the case of the distance of a point from a line, we rewrite the above
constraint system as:
\begin{alignat*}{2}
\text{minimize  }\  &   \delta \\ 
\text{subject to } \ & \left(\sum_{i=1}^n\abs{U_i} \right)- \delta \leq 0\\
& \abs{U_t}  -\delta \leq 0\\
&\left(\sum_{i=1}^n\abs{U_i'} \right)- \delta \leq 0\\
& \abs{U_t'}  -\delta \leq 0\\
 & U_i = z_i - s_{1,i}  +  \lambda \vdot (z_i' -  z_i) & \qquad &  \text{ for } 1\leq i \leq n\\
 & U_t = t_z-t_{s_1} + \lambda \vdot(t_{z'} - t_{z})\\
 & U_i' = z_i - s_{1,i}  +  \lambda \vdot (z_i' -  z_i) & \qquad &  \text{ for } 1\leq i \leq n\\
 & U_t' = t_z-t_{s_2} + \lambda \vdot(t_{z'} - t_{z})\\
 &  0 \leq \lambda \leq 1\\
  & 0 \leq \delta
\end{alignat*}

We remove the absolute values using the same trick as before, by introducing
$2\vdot n +2$ new variables $\absU_i, \absU'_i$  for $1\leq i \leq n$, and $\absU_t$
and $\absU'_t$.
The proof of correctness of this transformation is exactly the same as it was for the
case of the distance of a point from a line.
\begin{equation}
\label{LP:BallsLineLOneSkoro}
\begin{alignedat}{2}
\text{minimize  }\  &   \delta \\ 
\text{subject to } \ & \left(\sum_{i=1}^n\absU_i \right)- \delta \leq 0\\
& \absU_t  -\delta \leq 0\\
&\left(\sum_{i=1}^n\absU_i' \right)- \delta \leq 0\\
& \absU_t'  -\delta \leq 0\\
 & U_i = z_i - s_{1,i}  +  \lambda \vdot (z_i' -  z_i) & \qquad &  \text{ for } 1\leq i \leq n\\
 & U_t = t_z-t_{s_1} + \lambda \vdot(t_{z'} - t_{z})\\
 & U_i' = z_i - s_{1,i}  +  \lambda \vdot (z_i' -  z_i) & \qquad &  \text{ for } 1\leq i \leq n\\
 & U_t' = t_z-t_{s_2} + \lambda \vdot(t_{z'} - t_{z})\\
   & U_i - \absU_i \leq 0 & \qquad &  \text{ for } 1\leq i \leq n\\
   & -U_i - \absU_i \leq 0 & \qquad &  \text{ for } 1\leq i \leq n\\
   & U_t - \absU_t \leq 0\\
   & -U_t - \absU_t \leq 0\\
   & U'_i - \absU'_i \leq 0 & \qquad &  \text{ for } 1\leq i \leq n\\
   & -U'_i - \absU'_i \leq 0 & \qquad &  \text{ for } 1\leq i \leq n\\
   & U'_t - \absU'_t \leq 0\\
   & -U_t - \absU'_t \leq 0\\
   &  \absU_i \geq 0 & \qquad &  \text{ for } 1\leq i \leq n\\
   & \absU_t \geq 0 \\
   &  \absU'_i \geq 0 & \qquad &  \text{ for } 1\leq i \leq n\\
   & \absU'_t \geq 0 \\
   &  0 \leq \lambda \leq 1\\
  & 0 \leq \delta
\end{alignedat}
\end{equation}

\begin{proposition}
\label{proposition:BallsIntersectLOneSkoro}
Let $\tuple{s_1,t_{s_1}},  \tuple{s_2,t_{s_2}}, \tuple{z,t_z}, \tuple{z',t_{z'}}$ be points
with $s_1,s_2, z,z'\in \reals^n$ and $t_{s_1}, t_{s_2}, ,t_z, t_{z'}\in \reals$.
The value of the least $\delta \geq 0$ such that
$\ball(\tuple{s_1,t_{s_1}} \delta) \cap \ball(\tuple{s_2,t_{s_1}}, \delta)  
\cap \myline( \tuple{z,t_z}, \tuple{z',t_{z'}})$
is non-empty is given by the linear program~\ref{LP:BallsLineLOneSkoro}.
\qed
\end{proposition}

\subsection{$L_2^{\skoro}$-norm}

In this section, we compute the geometric primitives for
the $L_2^{\skoro}$-norm defined by  $\norm{\tuple{s,t}}_ {L_2^{\skoro}}= 
\max\left(\norm{s}_{L_2}, \abs{t}\right)$ for $\tuple{s,t}\in \reals^n\times \reals$.
We note the following identity that we will use:
\begin{equation}
\label{equation:MaxSkoroLTwo}
\left(\norm{\tuple{s,t}}_ {L_2^{\skoro}} \leq \delta\right)
\ \text{ iff } \ 
\left(\norm{s}_{L_2} \leq \delta \text{ and } \abs{t} \leq \delta\right)
\end{equation}

\smallskip\noindent\textbf{Distance from a point to a Line.}
Let us be given the points $\tuple{s,t_s}, \tuple{z,t_z}, \tuple{z',t_{z'}}$
with $s,z,z'\in \reals^n$ and $t\in \reals$.
For the $L_2^{\skoro}$-norm, the distance
$\dist_{L_2^{\skoro}}\left(\tuple{s,t_s}, \myline(\tuple{z,t_z}, \tuple{z',t_{z'}})\right)$
is equal to the following using Equation~\ref{equation:MaxSkoroLTwo}:
\begin{equation}
\label{equation:PointLineLTwoSkoro}
\min_{\delta \geq 0} \left(\text{There exists } \lambda \in [0,1]
\text{ such that }
\left\{
\begin{array}{c}
\abs{t_z-t_s + \lambda\vdot(t_{z'}-t_{z})}\,  \leq \delta\\
\text{and}\\
\norm{\vec{z}-\vec{s} + \lambda\vdot(\vec{z'}-\vec{z})}_{L_2} \, \leq \delta
\end{array}
\right\}
\right)
\end{equation}
Note that the minimum exists because $\lambda\in [0,1]$, a closed set.
To compute this quantity, we  follow  a strategy similar to that for computing
the optimization problem in $L_2$ of two balls intersecting with a line --
we reason over $\lambda$ sets.
We have the following cases.

\noindent $\blacktriangleright$
If $\tuple{z,t_z} = \tuple{z',t_{z'}} $, the mininal $\delta$  is just
$\max\left(\norm{z-s}_{L_2}, \abs{t_z-t_s}\right)$.
So, assume $\tuple{z,t_z} \neq  \tuple{z',t_{z'}} $.

\noindent $\blacktriangleright$
If $\tuple{z,t_z} \neq  \tuple{z',t_{z'}} $, but $z=z'$, then we show
that Equation~\ref{equation:PointLineLTwoSkoro} has
the value $D= \max\left(\norm{z-s}_{L_2}, \dist\left(t_s, \myline(t_z, t_{z'})\right)\right)$
as follows.
\begin{compactitem}
\item[$D \geq \dist_{L_2^{\skoro}}\left(\tuple{s,t_s}, \myline(\tuple{z,t_z}, \tuple{z,t_{z'}})\right)$:]
We show there exists a $\lambda \in [0,1]$ such that both the clauses of
Equation~\ref{equation:PointLineLTwoSkoro} are satisfied.
By defintion, we have that there exists a $\lambda \in [0,1]$ such that
$\abs{t_z-t_s + \lambda\vdot(t_{z'}-t_{z})}\,  \leq D$.
Moreover, this value of $\lambda$ gives the point $z$ on the line
$\myline(z,z)$ (both endpoints are the same).
And by definition, we have $\norm{z-s}_{L_2} \leq D$.
This proves $D \geq \dist_{L_2^{\skoro}}\left(\tuple{s,t_s}, \myline(\tuple{z,t_z}, \tuple{z,t_{z'}})\right)$.
\item [$D \leq \dist_{L_2^{\skoro}}\left(\tuple{s,t_s}, \myline(\tuple{z,t_z}, \tuple{z,t_{z'}})\right)$:]
Let $\lambda$ be the value corresponding to the optimal $\delta$ according to
Equation~\ref{equation:PointLineLTwoSkoro} .
This means that $\delta \geq \norm{z-s}_{L_2}$ (from the second clause,
and noting $z+\lambda(z-z) = z$).
Also,  $\delta \geq  \dist\left(t_s, \myline(t_z, t_{z'})\right)$, othewise
the first clause cannot be satisfied for any $\lambda \in [0,1]$.
This proves $D \leq \dist_{L_2^{\skoro}}\left(\tuple{s,t_s}, \myline(\tuple{z,t_z}, \tuple{z,t_{z'}})\right)$.
\end{compactitem}

\medskip
\noindent $\blacktriangleright$
If $\tuple{z,t_z} \neq  \tuple{z',t_{z'}} $, and  $z\neq z'$, but
$t_z = t_{z'}$, we claim
that Equation~\ref{equation:PointLineLTwoSkoro} has
the value
$ \max\left( \dist_{L_2}\left(s, \myline(z,z')\right)\, ,\, \abs{t_s-t_z}\right)$.
The proof is similar to the previous case.

\noindent $\blacktriangleright$
Assume $z\neq z'$ and $t_z \neq  t_{z'}$.
First we compute 
$\dist_{L_2}\left(s, \myline(z,z')\right)$ as in 
Proposition~\ref{proposition:DistancePointLineTwo}.
We also compute $\lambda_p$, such that
the point $z+\lambda_p\vdot(z'-z)$ is $\dist_{L_2}\left(s, \myline{z,z'}\right)$ distance
away from $s$.
This can be done as done in the proof of 
 Proposition~\ref{proposition:DistancePointLineTwo}:
we compute $\frac{(\vec{z'}- \vec{z})  \odot(\vec{s} - \vec{z})}{\norm{ \vec{z'}- \vec{z} }^2}$; and 
(a)~if this value is in  the interval $[0,1]$, then it is $\lambda_p$;
(b)~If this value is less than $0$, then $\lambda_p=0$ (in this
case $z$ is the closest point to $s$);
(b)~If this value is greater than $1$, then $\lambda_p=1$ (in this
case $z'$ is the closest point to $s$).

Let $t_s = t_z + \lambda_{t_s} \vdot(t_{z'} - t_z)$.
Thus, $\lambda_{t_s} = \frac{t_s - t_z}{t_{z'} - t_z}$ (note that by 
assumption, $t_z \neq  t_{z'}$).
Note that $\lambda_{t_s} $ will in general \emph{not} be in
the interval $[0,1]$.
Assume $ \lambda_{t_s}  \geq  \lambda_p$ (otherwise, we just switch 
$\tuple{z,t_z}, \tuple{z,t_{z'}}$).
Thus on the real line, the point where $ \lambda_{t_s}$
lies is to the right of the minimum distance point $\lambda_p$ which
minimizes the $L_2$ norm for the $\reals^n$ components.
The value $\lambda_{t_s}$ can thus be   assumed to be in the interval $[0, \infty)$.

Let $h^1:  [\dist_{L_2}\left(s, \myline(z,z')\right), \infty) \mapsto [0,1]$ be the function
$h^1(\delta) =  \set{\lambda \mid 0\leq\lambda\leq 1 \text{ and a}\norm{z-s + \lambda \vdot (z'-z)}_{L_2} \leq \delta}$;
thus, $h^1()$ is the $h()$ function from Lemma~\ref{lemma:BallLineIntersect}.
The value $h_1(\delta)$ indicates the range of
$\lambda$ values in $[0,1]$ such that
for these $\lambda$ values, the points $z+ \lambda \vdot (z'-z)_{L_2}$ are at most
$\delta$ away (in the $L_2$ metric)  from $s$.
Note that for $\delta\geq \dist_{L_2}\left(s, \myline(z,z')\right)$, the set
$h^1(\delta)$ is never empty.
Let $h^t: \reals_+ \mapsto \reals$ be the function
$h^t(\delta) =  \set{\lambda \mid\abs{ t_z-t_s + \lambda\vdot(t_{z'}- t_z)} \leq \delta}$,
\emph{i.e.} the $\lambda$ values in $\reals$ such that the corresponding point
$ t_z + \lambda\vdot(t_{z'}- t_z)$ on the infinite line is at most $\delta$ away from
$t_s$.
Observe that the solution to Equation~\ref{equation:PointLineLTwoSkoro}
is $\min_{\delta \geq \dist_{L_2}\left(s, \myline(z,z')\right)} h^1(\delta)  \cap h^t(\delta) \neq \emptyset$.
We take $\delta \geq \dist_{L_2}\left(s, \myline(z,z')\right)$ because clearly
for $\delta <  \dist_{L_2}\left(s, \myline(z,z')\right)$
there is no $\lambda\in [0,1]$ such that
$\norm{\vec{z}-\vec{s} + \lambda\vdot(\vec{z'}-\vec{z})}_{L_2} \, \leq \delta$.
Let $h_{\max}^1(\delta)$ be as in Lemma~\ref{lemma:BallLineIntersect},
\emph{i.e.}  $h_{\max}^1(\delta) = \min h^1(\delta)$
 (the maximum exists since
$ h^1(\delta)$  is closed).
Let $h_{\min}^t(\delta) = \min h^t(\delta) $  (the minimum exists since
$ h^t(\delta)$  is closed).
We thus have the solution to Equation~\ref{equation:PointLineLTwoSkoro}
to be 
\begin{equation}
\label{equation:PointLineLTwoSkoroMinMax}
\min_{\delta \geq \dist_{L_2}\left(s, \myline(z,z')\right)}  \text{ such that }h_{\max}^1(\delta) \geq h_{\min}^t(\delta) .
\end{equation}

We have
$h^t(\delta)  $, which is defined to be 
$  \set{\lambda \mid\abs{ t_z-t_s + \lambda\vdot(t_{z'}- t_z)} \leq \delta}$,
to be equal to:
\begin{flalign*}
&\set{\lambda \mid\abs{ t_z- (t_z + \lambda_s\vdot(t_{z'}- t_z)) + \lambda\vdot(t_{z'}- t_z)} \leq \delta} \text{ since } t_s=t_z + \lambda_s\vdot(t_{z'}- t_z)&\\
&
=
\set{\lambda \mid\abs{ (\lambda-\lambda_s)\vdot (t_{z'}- t_z)} \leq \delta}&\\
& 
=
\set{\lambda \mid\abs{ (\lambda-\lambda_s)}\vdot \abs{(t_{z'}- t_z)} \leq \delta}&\\
& 
=
\set{\lambda \mid\abs{ (\lambda-\lambda_s) }\leq \frac{\delta}{\abs{(t_{z'}- t_z)}}}
\text{ note that by assumption } t_z \neq  t_{z'} &
\end{flalign*}
Thus, $h_{\min}^t(\delta) $ is $\lambda_s -  \frac{\delta}{\abs{t_{z'}- t_z}}$;
and since $\lambda_s = \frac{t_s - t_z}{t_{z'} - t_z}$, we have
\begin{equation}
\label{equation:DeltaMinTime}
h_{\min}^t(\delta)  =  \frac{t_s - t_z}{t_{z'} - t_z} - \frac{\delta}{\abs{t_{z'}- t_z}}
\end{equation}
Observe that $h_{\min}^t(\delta)$ is continuous and strictly decreasing;
and that if $h_{\min}^t(\delta) = \lambda_{\delta}$, then
$\abs{t_z-t_s + \lambda_{\delta}\vdot(t_{z'}- t_z)} = \delta$.
We have the following cases.
\begin{enumerate}[(a)]
\item $h_{\min}^t\left(\dist_{L_2}\left(s, z'\right) \right) > 1$.
This means that  
$\min_{\delta \geq \dist_{L_2}\left(s, \myline(z,z')\right)}  h_{\max}^1(\delta) \geq h_{\min}^t(\delta) $
is equal to
$\min_{\delta > \dist_{L_2}\left(s, z'\right)}  h_{\max}^1(\delta) \geq h_{\min}^t(\delta) $
(since the maximum value of $h_{\max}^1()$ is $1$, and $h_{\min}^t$ is strictly
decreasing).
Since $h_{\max}^1(\delta) = 1 $ for all  $\delta \geq \dist_{L_2}\left(s, z'\right)$,
we have:
\[
\left(\min_{\delta > \dist_{L_2}\left(s, z'\right)} \text{ such that } h_{\max}^1(\delta) \geq h_{\min}^t(\delta)  \right)\ = \ 
\left( \min_{\delta \geq 0} \text{ such that } h_{\min}^t(\delta) \leq 1\right).
\]
Using Equation~\ref{equation:DeltaMinTime}, we get the above minimum to be
when $ \frac{t_s - t_z}{t_{z'} - t_z} - \frac{\delta}{\abs{t_{z'}- t_z}} = 1$,
\emph{i.e.}, the minimum value for $\delta$ is
$\delta =  \abs{t_{z'}- t_z}\cdot \left(\frac{t_s - t_z}{t_{z'} - t_z} - 1\right)$.

\item $h_{\min}^t\left(\dist_{L_2}\left(s, z'\right) \right)  \leq 1$.
Note that $h_{\max}^1\left(\dist_{L_2}\left(s, z'\right)\right)  = 1$,
thus, the solution to
Equation~\ref{equation:PointLineLTwoSkoroMinMax} must be $\delta \leq 1$.
Hence, we can restrict the range of 
Equation~\ref{equation:PointLineLTwoSkoroMinMax} as follows.
\begin{equation}
\label{equation:PointLineLTwoSkoroMinMaxRestrictOne} 
\min_{\dist_{L_2}\left(s, z'\right)\geq \delta \geq \dist_{L_2}\left(s, \myline(z,z')\right)}  \text{ such that }
h_{\max}^1(\delta) \geq h_{\min}^t(\delta) .
\end{equation}
Suppose $h_{\min}^t\left(\dist_{L_2}\left(s, \myline(z,z')\right)\right) $
has the value $\lambda$ such that $\lambda\leq \lambda_p$ where
$\lambda_p$ is such that
$\dist_{L_2}\left(z-s + \lambda_p\vdot(z'-z)\right) $
is equal to  $\dist_{L_2}\left(s, \myline(z,z')\right)$, that is
$\lambda_p$  gives the point on the line $ \myline(z,z')$ which
minimizes the distance to $s$.
In this case, the minimum value of $\delta$  satisfying 
Equation~\ref{equation:PointLineLTwoSkoroMinMaxRestrictOne} is
$\dist_{L_2}\left(s, \myline(z,z')\right)$.

Suppose   $h_{\min}^t\left(\dist_{L_2}\left(s, \myline(z,z')\right)\right) $
has the value $\lambda$ such that $\lambda> \lambda_p$ where
$\lambda_p$ is as before.
Over the range 
$\dist_{L_2}\left(s, z'\right)\geq \delta \geq \dist_{L_2}\left(s, \myline(z,z')\right)$,
the function  $h_{\max}^1()$ is continuous and strictly increasing by
Lemma~\ref{lemma:BallLineIntersect}.
Also, the function $ h_{\min}^t()$ is continuous and strictlly decreasing.
At $\delta = \dist_{L_2}\left(s, \myline(z,z')\right)$, the value of 
$ h_{\min}^t(\delta)$ is greater than that of $h_{\max}^1(\delta)$.
And at $\delta'= \dist_{L_2}\left(s, z'\right) $,
 the value of 
$ h_{\min}^t(\delta')$ is greater than that of $h_{\max}^1(\delta')$.
Thus, $h_{\min}^t()$ and $h_{\max}^1()$ must intersect at some point 
$\delta^*$ in between $\delta $ and $\delta'$.
Let $h_{\min}^t(\delta^*) = \lambda^* = h_{\max}^1(\max)$.
Moreover, by the properties of $h_{\min}^t()$ and $h_{\max}^1()$,
we have the following two
equalities:
$\dist_{L_2}\left(z-s + \lambda^*\vdot(z'-z)\right) = \delta^*$;
and
$\abs{t_z-t_s + \lambda^*\vdot(t_{z'}- t_z)} = \delta^*$.
Thus,
\begin{flalign*}
& \norm{\left(z-s + \lambda^*\vdot(z'-z)\right) }^2 = 
\left(t_z-t_s + \lambda^*\vdot(t_{z'}- t_z)\right)^2&\\
& \text{Expanding, }\ 
\sum_{k=1}^n\Big(
(z_k-s_k)^2 +(\lambda^*)^2\vdot (z'_k-z_k)^2  + \lambda^*\vdot 2 \vdot (z_k-s_k)\vdot (z'_k-z_k) \Big)\ =\ &\\
&\hspace{5cm}
\sum_{k=1}^n
\Big(
(t_{z,k}- t_{s,k})^2 +  (\lambda^*)^2\vdot (t_{z',k}- t_{z,k})^2 + 
\lambda^*\vdot 2\vdot (t_{z,k}- t_{s,k})\vdot (z'_k-z_k) 
\Big)&\\
& \text{Rearranging, }\ 
(\lambda^*)^2\vdot \sum_{k=1}^n \big((z'_k-z_k)^2 - (t_{z',k}- t_{z,k})^2\big) \quad + &\\
&\hspace{4cm}
\lambda^* \vdot 2\vdot \sum_{k=1}^n \big((z_k-s_k)\vdot (z'_k-z_k) -  (t_{z,k}- t_{s,k})\vdot (z'_k-z_k) \big) \quad +&\\
&\hspace{8cm}
\sum_{k=1}^n \big((z_k-s_k)^2
- (t_{z,k}- t_{s,k})^2 \big) \quad =\  0 &
\end{flalign*}
We solve the quadratic equation above to obtain $\lambda^*$.
The assumed conditions imply that there will be exactly one root in the interval 
$[0,1]$ (otherwise we would get two intersections of
$h_{\min}^t()$ and $h_{\max}^1()$, which is not possible since the first one is
strictly decreasing, and the second one strictly increasing.
The optimal value $\delta^*$ is then $\abs{t_z-t_s + \lambda^*\vdot(t_{z'}- t_z)} $.

\end{enumerate}

We put everything together in Function~\ref{function:PointLineSkoroLTwo}.

\begin{function}
  \SetKwInOut{Input}{Input}
  \SetKwInOut{Output}{Output}
  \Input{Points $\tuple{s,t_s}, \tuple{z,t_z}, \tuple{z,t_{z'}}$ in $\reals^n\times\reals$}
  \Output{$\dist_{L_2^{\skoro}}\left(\tuple{s,t_s}, \myline(\tuple{z,t_z}, \tuple{z,t_{z'}})\right)$}
  \Switch{$\tuple{z,t_z}, \tuple{z,t_{z'}}$}{
        \lCase{ $\tuple{z,t_z} = \tuple{z',t_{z'}} $ }{
          \Return{$\max\left(\norm{z-s}_{L_2}, \abs{t_z-t_s}\right)$\;}
        }
        \lCase{ $\tuple{z,t_z} = \tuple{z',t_{z'}} $ and $z=z'$}{
          \Return{$\max\left(\norm{z-s}_{L_2}, \dist\left(t_s, \myline(t_z, t_{z'})\right)\right)$\;}
        }
        \lCase{$\tuple{z,t_z} \neq  \tuple{z',t_{z'}} $, and  $z\neq z'$, and
          $t_z = t_{z'}$}{
          \Return{$ \max\left( \dist_{L_2}\left(s, \myline(z,z')\right)\, ,\, \abs{t_s-t_z}\right)$\;}
        }
        \Case{$z\neq z'$, and $t_z \neq t_{z'}$} {
          $\lambda_p:=   \frac{(\vec{z'}- \vec{z})  \odot(\vec{s} - \vec{z})}{\norm{ \vec{z'}- \vec{z} }^2}$\;
          \lIf{$\lambda_p < 0$}{$\lambda_p :=0$\;}
          \lElseIf{$\lambda_p > 1$}{$\lambda_p :=1$\;}
          $\lambda_{t_s}: = \frac{t_s - t_z}{t_{z'} - t_z}$ \;
          \If{$\lambda_{t_s}  < \lambda_p$}{
            swap  $\tuple{z,t_z}$  and  $\tuple{z,t_{z'}}$ 
            \tcp*{Make $\lambda_{t_s}  \geq  \lambda_p$}
            $\lambda_p:= 1-\lambda_p$\;
            $\lambda_{t_s}:=1-\lambda_{t_s}$\;
          }
          $\alpha:= \left( \frac{t_s - t_z}{t_{z'} - t_z} - \frac{\norm{s - z'}_{L_2}}{\abs{t_{z'}- t_z}}\right)$
          \tcp*[r]{$\alpha= h_{\min}^t\left(\dist_{L_2}\left(s, z'\right) \right) $}
          \lIf{$\alpha> 1$}
          {
            \Return{$ \abs{t_{z'}- t_z}\cdot \left(\frac{t_s - t_z}{t_{z'} - t_z} - 1\right)$\;}
          }
          \Else{
            \lIf{$\alpha \leq  \lambda_p$}{
              \Return{$\dist_{L_2}\left(s,\myline(z,z')\right)$}\;
            }
            \Else{
              Solve the following quadratic equation for $\lambda^*\in [0,1]$
              \begin{flalign*}
               &\hspace{1cm} (\lambda^*)^2\vdot \sum_{k=1}^n \big((z'_k-z_k)^2 - (t_{z',k}- t_{z,k})^2\big) \quad + &\\
                &\hspace{2cm}
                \lambda^* \vdot 2\vdot \sum_{k=1}^n \big((z_k-s_k)\vdot (z'_k-z_k) -  (t_{z,k}- t_{s,k})\vdot (z'_k-z_k) \big) \quad +&\\
                &\hspace{4cm}
                \sum_{k=1}^n \big((z_k-s_k)^2
                - (t_{z,k}- t_{s,k})^2 \big) \quad =\  0 \;&
              \end{flalign*}
              \Return{$\abs{t_z-t_s + \lambda^*\vdot(t_{z'}- t_z)} $\;}

            }
        }
      }
}

 \caption{
$\Psi_{L_2^{\skoro}}$($\tuple{s,t_s},\tuple{z,t_z}, \tuple{z,t_{z'}}$)
}
  \label{function:PointLineSkoroLTwo}
\end{function}

\begin{proposition}[Distance of point to line: $L_2^{\skoro}$]
\label{proposition:PointLineLTwoSkoro}
Given points $\tuple{s,t_s}, \tuple{z,t_z}, \tuple{z,t_{z'}}$ in $\reals^n\times\reals$,
Function~\ref{function:PointLineSkoroLTwo} computes
the minimum distance of the point $\tuple{s,t_s}$ from the line
$\myline(\tuple{z,t_z}, \tuple{z,t_{z'}})$ for the
$L_2^{\skoro}$-norm.
\qed
\end{proposition}

\smallskip\noindent\textbf{Computation of the least $\delta$ such that 
$\ball(\tuple{s_1,t_{s_1}} \delta) \cap \ball(\tuple{s_2,t_{s_1}}, \delta)  
\cap \myline( \tuple{z,t_z}, \tuple{z',t_{z'}})$
is non-empty.}\\
We use the following result to solve the optimization problem.

\begin{proposition}[Helley's theorem~\cite{HandbookConvex}]
\label{proposition:Helley}
Let $X_1, \dots, X_r$ be a finite collection of convex subsets of $\reals^d$ with $r > d$.
If the intersection of every $d+1$ of these sets is nonempty, then 
$\displaystyle\bigcap_{i=1}^r X_i\  \neq \emptyset$. 
\qed

\end{proposition}

Let us be given the points 
$\tuple{s_1,t_{s_1}},  \tuple{s_2,t_{s_2}}, \tuple{z,t_z}, \tuple{z',t_{z'}}$
with $s_1,s_2, z,z'\in \reals^n$ and $t_{s_1}, t_{s_2}, ,t_z, t_{z'}\in \reals$.
For the $L_2^{\skoro}$-norm,  given a $\delta\geq 0$, we have
$\ball(\tuple{s_1,t_{s_1}}, \delta) \cap \ball(\tuple{s_2,t_{s_2}}, \delta)  \cap 
\myline(\tuple{z,t_z}, \tuple{z,t_{z'}})$ to be non-empty iff
there is some point $\tuple{l,t_l}$ on the line $\myline(\tuple{z,t_z}, \tuple{z,t_{z'}})$
such that both the following two conditions hold: 
(i)~$\norm{\tuple{s_1,t_{s_1}} - \tuple{l,t_l}}_{L_2^{\skoro}} \leq \delta$; and
(ii)~$\norm{\tuple{s_2,t_{s_2}} - \tuple{l,t_l}}_{L_2^{\skoro}} \leq \delta$; and
Using~\ref{equation:MaxSkoroLOne}, the least $\delta$ such that the previous two 
conditions hold  is thus
given by the following expression:
\begin{equation}
\label{equation:BallsLineLTwoSkoro}
\min_{\delta \geq 0} \left(\text{There exists } \lambda \in [0,1]
\text{ such that }
\left\{
\begin{array}{l}
\norm{\vec{z}-\vec{s_1} + \lambda\vdot(\vec{z'}-\vec{z})}_{L_2} \, \leq \delta; 
\text{ and}\\
\abs{t_z-t_{s_1} + \lambda\vdot(t_{z'}-t_{z})}\,  \leq \delta; \text{ and}\\
\norm{\vec{z}-\vec{s_2} + \lambda\vdot(\vec{z'}-\vec{z})}_{L_2} \, \leq \delta; 
\text{ and}\\
\abs{t_z-t_{s_2} + \lambda\vdot(t_{z'}-t_{z})}\,  \leq \delta
\end{array}
\right\}
\right)
\end{equation}

A possible approach to solve the above optimization problem is to proceed as in $L_2$, 
and reason over $\lambda$ sets.
However,  in this case we have $4$ $\lambda$ sets at play, so at first glance 
we cannot just
use the old approach.
However, we show that  Helly's theorem~\ref{proposition:Helley} allows us to
restrict our attention to only two $\lambda$ sets at a time.
We proceed as follows.

\noindent\underline{\textsf{Symmetry}}:  Observe from 
Equation~\ref{equation:BallsLineLTwoSkoro} that we can swap $t_{s_1}$ and $t_{s_2}$
without changing the value of the least $\delta$, \emph{i.e.}, the minimum $\delta$
such that
$\ball(\tuple{s_1,t_{s_1}}, \delta) \cap \ball(\tuple{s_2,t_{s_2}}, \delta)  \cap 
\myline(\tuple{z,t_z}, \tuple{z',t_{z'}})$ is non-empty equals the
the minimum $\delta$
such that
$\ball(\tuple{s_1,t_{s_2}}, \delta) \cap \ball(\tuple{s_2,t_{s_1}}, \delta)  \cap 
\myline(\tuple{z,t_{z}}, \tuple{z',t_{z'}})$  is non-empty.
We will use this symmetry to reduce the number of cases that need to be considered.

\noindent\underline{\textsf{Swapping $\tuple{z,t_z} $ with $ \tuple{z',t_{z'}}$}}:
Clearly we can swap $\tuple{z,t_z} $  with   $ \tuple{z',t_{z'}}$, as the
line between the two points remains the same.
Observe that if $t= t_z + \lambda \vdot (t_{z'} - t_{z})$,
then we can rewrite the expression as
$t= t_{z'} + (1-\lambda)\vdot (t_{z} - t_{z'})$ (and similarly for $z,z'$), 
thus after the swap, the new
$\lambda$ values are $1$ minus the old ones.

\noindent\underline{\textsf{The function $h^{\infty}$}}:
Recall the function $h(\delta)$ from Lemma~\ref{lemma:BallLineIntersect}.
We define a similar function which can have all elements from reals, not just from
$[0,1]$.
Suppose we are given a infinite line $\myline_{\infty}(l_1, l_2)$, passing through
$l_1$ and $l_2$; and a point $l$ (not necessarily on that line).
Givem $\delta \geq 0$, we want a representation of all points on the line
that are at most $\delta$ away from $l$.
\[
h^{\infty, l_1, l_2}_l(\delta) = \set{\lambda \mid \norm{l_1-l + \lambda\vdot(l_2-l_1)}\leq \delta}
\]
We omit the annotations $l_1, l_2$ which specify the line when the line is clear from
context.
For the norms $L_2$ and $\abs{\cdot}$, it can be show that
for $\delta \geq \dist(, \myline(l_1, l_2)$, 
(a)~the function $\min(h_l^{\infty})$ is strictly 
decreasing, and the function  $\max(h_l^{\infty})$ is strictly increasing; and
(b)~if  $\max(h_l^{\infty})(\delta) = \lambda$, then
$\norm{ l_1-l + \lambda\vdot(l_2-l_1)} = \delta$; and similarly for
 $\min(h_l^{\infty})$.

For the (one dimensional)  line $\myline(t_{z'} , t_z )$,  where $t_z\neq t_{z'}$,
and a point $t$,
we have $h^{\infty}_l(\delta)$ to contains $\lambda$ such that
$\abs{t_z-t + \lambda\vdot(t_{z'} - t_z)} \leq \delta$.
This holds iff both the following hold:
(a)~$\lambda\vdot(t_{z'} - t_z) \leq \delta+t-t_z$; and
(b)~$-\lambda\vdot(t_{z'} - t_z) \leq  \delta - (t-t_z)$ which
is equilvalent to ~$\lambda\vdot(t_{z'} - t_z) \geq  -\left(\delta - (t-t_z)\right)$.
Thus,
\begin{equation}
\label{equation:HInftyBoundaryOneDim}
h^{\infty}_{t}(\delta) = 
\begin{cases} 
[-\frac{\delta - (t-t_z) }{t_{z'} - t_z}\, , \frac{\delta+t-t_z}{t_{z'} - t_z} ] & \text{ if }
(t_{z'} - t_z) > 0\\
[\frac{\delta+t-t_z}{t_{z'} - t_z}\, , -\frac{\delta - (t-t_z)}{t_{z'} - t_z} ] & \text{ if }
(t_{z'} - t_z) < 0
\end{cases}
\end{equation}

We obtain the value of the minimum $\delta$ (denoted as $\delta^*$) as follows.
We let $\lambda^*$ be the $\lambda$ value when we have $\delta^*$
in Equation~\ref{equation:BallsLineLTwoSkoro} (it can be shown that this
$\lambda^*$ is unique).
We define the following parameters, and use them in the computation of the
least $\delta$.
\begin{compactitem}
\item $\lambda_p^{s_1}$: this is the value, defined when  $z\neq z'$,  such that 
the point $z + \lambda_p^{s_1}\vdot(z'-{z})$ on the line $\myline(z,z')$ 
is the least distance away from $s_1$ in the $L_2$ norm.
Proposition~\ref{proposition:DistancePointLineTwo} gives us 
$\lambda_p^{s_1} = \frac{(\vec{z'}- \vec{z})  \odot(\vec{s_1} - \vec{z})}{\norm{ \vec{z'}- \vec{z} }^2}$ 

\item  $\lambda_{t_{s_1}}$: this is the value, defined when $ t_{z'} \neq  t_z$,
 such that
$t_{s_1} = t_z +  \lambda_{t_{s_1}}\vdot ( t_{z'} -  t_z )$. It
 can be seen that
$ \lambda_{t_{s_1}} = \frac{t_{s_1} - t_z }{t_{z'} -  t_z }$.
\end{compactitem}
We also similarly  have the parameters $\lambda_p^{s_2}$ and $\lambda_{t_{s_1}}$.

\noindent $\blacktriangleright$ 
Suppose $\tuple{s_1,t_{s_1}} =  \tuple{s_2,t_{s_2}} $.
Then  $\delta^*$ is 
$\dist_{L_2^{\skoro}}\left(\tuple{s_1,t_{s_1}} , \myline( \tuple{z,t_z}, \tuple{z',t_{z'}})\right)$.
This can be computed using Proposition~\ref{proposition:PointLineLTwoSkoro}.

\noindent $\blacktriangleright$ 
Suppose $\tuple{z,t_z} = \tuple{z',t_{z'}}$.
Then $\delta^* = \max\left(\norm{\tuple{s_1,t_{s_1}- \tuple{z,t_z}}}_{L_2^{\skoro}}\, , \,
\norm{\tuple{s_2,t_{s_2}- \tuple{z,t_z}}}_{L_2^{\skoro}}\right)$.
Thus 
$\delta^* = \max\left(\norm{s_1-z}_{L_2}, \abs{t_{s_1}- t_z}, 
\norm{s_2-z}_{L_2}, \abs{t_{s_2}- t_z}\right)$.

\noindent $\blacktriangleright$ 
Suppose $\tuple{z,t_z} \neq  \tuple{z',t_{z'}}$, but
$z= z'$ (and thus $ t_{z'} \neq  t_z$).
Then,
\[\delta^* =
\min_{\delta \geq 0} \left(\text{There exists } \lambda \in [0,1]
\text{ such that }
\left\{
\begin{array}{l}
\norm{z-s_1}_{L_2} \leq \delta; \text{ and}\\
\abs{t_z-t_{s_1} + \lambda\vdot(t_{z'}-t_{z})}\,  \leq \delta; \text{ and}\\
\norm{z-s_2}_{L_2} \leq \delta; \text{ and}\\
\abs{t_z-t_{s_2} + \lambda\vdot(t_{z'}-t_{z})}\,  \leq \delta
\end{array}
\right\}
\right).
\]
This can be simplified to:
\[\delta^* =
\min_{\delta \geq \max(\norm{z-s_1}_{L_2}, \norm{z-s_2}_{L_2})} \left(\text{There exists } \lambda \in [0,1]
\text{ such that }
\left\{
\begin{array}{l}
\abs{t_z-t_{s_1} + \lambda\vdot(t_{z'}-t_{z})}\,  \leq \delta; \text{ and}\\
\abs{t_z-t_{s_2} + \lambda\vdot(t_{z'}-t_{z})}\,  \leq \delta
\end{array}
\right\}
\right).
\]

Using $h^{\infty}$ functions, the above can be written as
\begin{equation}
\label{equation:OnlyTwoLinesLTwoSkoroPre}
\delta^* = \min_{\delta \geq \max(\norm{z-s_1}_{L_2}, \norm{z-s_2}_{L_2})} \left(
h^{\infty}_{ t_{s_1}}(\delta) \, \cap\,  h^{\infty}_{ t_{s_2}}(\delta) \cap [0,1] \neq \emptyset\right).
\end{equation}

The above equation is equivalent to the following:
\begin{gather}
  \delta^{\dagger} \left( t_{s_1},  t_{s_2}, \myline(t_{z'},t_{z})\right) = 
  \min_{\delta \geq 0} \left(
    h^{\infty}_{ t_{s_1}}(\delta) \, \cap\,  h^{\infty}_{ t_{s_2}}(\delta) \cap [0,1] \neq \emptyset
  \right)\label{gather:OnlyTwoLinesSkoroDagger}\\
  \delta^* = \max\left( 
  \norm{z-s_1}_{L_2}, \norm{z-s_2}_{L_2}, 
  \delta^{\dagger}\left( t_{s_1},  t_{s_2}, \myline(t_{z'},t_{z})
  \right) \right)
\end{gather}

We now show how to compute $\delta^{\dagger}$ (we omit the arguments
$t_{s_1},  t_{s_2}, \myline(t_{z'},t_{z})$ for simplicity).
The intervals around $ \lambda_{t_{s_1}}$ and $  \lambda_{t_{s_2}}$  keep getting bigger
as $\delta$ gets bigger -- we want the least $\delta$ such that there
is an intersection point in $[0,1]$.
Moreover, the rate at which the boundaries of the intervals increase or decrease
 is the same
from Equation~\ref{equation:HInftyBoundaryOneDim}.
We assume $ \lambda_{t_{s_2}}\geq 0$ (otherwise, we swap
$\tuple{z,t_z} $ with $ \tuple{z',t_{z'}}$ to make this so).
We also assume $ \lambda_{t_{s_2}}\geq \lambda_{t_{s_1}}$ (this can be ensured
by swapping  $t_{s_1}$ and $t_{s_2}$ if necessary).
The following cases arise (see Figure~\ref{figure:oneLTwoSkoro}
for a pictorial representation of the $\lambda$ value placements).

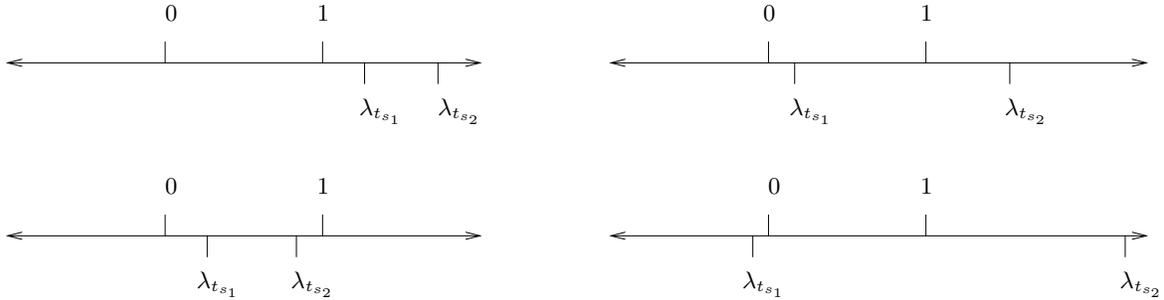
\begin{figure}[h]
\strut\centerline{\setlength{\unitlength}{0.00061242in}
\begingroup\makeatletter\ifx\SetFigFont\undefined%
\gdef\SetFigFont#1#2#3#4#5{%
  \reset@font\fontsize{#1}{#2pt}%
  \fontfamily{#3}\fontseries{#4}\fontshape{#5}%
  \selectfont}%
\fi\endgroup%
{\renewcommand{\dashlinestretch}{30}
\begin{picture}(9789,2547)(0,-10)
\path(132.000,2043.000)(12.000,2013.000)(132.000,1983.000)
\path(12,2013)(4062,2013)
\path(3942.000,1983.000)(4062.000,2013.000)(3942.000,2043.000)
\path(1362,2193)(1362,2013)
\path(2712,2193)(2712,2013)
\path(132.000,558.000)(12.000,528.000)(132.000,498.000)
\path(12,528)(4062,528)
\path(3942.000,498.000)(4062.000,528.000)(3942.000,558.000)
\path(1362,708)(1362,528)
\path(2712,708)(2712,528)
\path(3702,2013)(3702,1833)
\path(3072,2013)(3072,1833)
\path(1722,528)(1722,348)
\path(2487,528)(2487,348)
\path(8607,2013)(8607,1833)
\path(6402,528)(6402,348)
\path(6537,708)(6537,528)
\path(7887,708)(7887,528)
\path(6537,2193)(6537,2013)
\path(7887,2193)(7887,2013)
\path(5307.000,558.000)(5187.000,528.000)(5307.000,498.000)
\path(5187,528)(9777,528)
\path(9657.000,498.000)(9777.000,528.000)(9657.000,558.000)
\path(5307.000,2043.000)(5187.000,2013.000)(5307.000,1983.000)
\path(5187,2013)(9777,2013)
\path(9657.000,1983.000)(9777.000,2013.000)(9657.000,2043.000)
\path(9597,528)(9597,348)
\path(6762,2013)(6762,1833)
\put(1362,2373){\makebox(0,0)[lb]{\smash{{\SetFigFont{9}{10.8}{\familydefault}{\mddefault}{\updefault}$0$}}}}
\put(2667,2373){\makebox(0,0)[lb]{\smash{{\SetFigFont{9}{10.8}{\familydefault}{\mddefault}{\updefault}$1$}}}}
\put(1362,888){\makebox(0,0)[lb]{\smash{{\SetFigFont{9}{10.8}{\familydefault}{\mddefault}{\updefault}$0$}}}}
\put(2667,888){\makebox(0,0)[lb]{\smash{{\SetFigFont{9}{10.8}{\familydefault}{\mddefault}{\updefault}$1$}}}}
\put(3027,1563){\makebox(0,0)[lb]{\smash{{\SetFigFont{9}{10.8}{\familydefault}{\mddefault}{\updefault}$\lambda_{t_{s_1}}$}}}}
\put(3702,1563){\makebox(0,0)[lb]{\smash{{\SetFigFont{9}{10.8}{\familydefault}{\mddefault}{\updefault}$\lambda_{t_{s_2}}$}}}}
\put(1632,78){\makebox(0,0)[lb]{\smash{{\SetFigFont{9}{10.8}{\familydefault}{\mddefault}{\updefault}$\lambda_{t_{s_1}}$}}}}
\put(6312,78){\makebox(0,0)[lb]{\smash{{\SetFigFont{9}{10.8}{\familydefault}{\mddefault}{\updefault}$\lambda_{t_{s_1}}$}}}}
\put(2442,78){\makebox(0,0)[lb]{\smash{{\SetFigFont{9}{10.8}{\familydefault}{\mddefault}{\updefault}$\lambda_{t_{s_2}}$}}}}
\put(6492,2373){\makebox(0,0)[lb]{\smash{{\SetFigFont{9}{10.8}{\familydefault}{\mddefault}{\updefault}$0$}}}}
\put(6537,888){\makebox(0,0)[lb]{\smash{{\SetFigFont{9}{10.8}{\familydefault}{\mddefault}{\updefault}$0$}}}}
\put(7842,888){\makebox(0,0)[lb]{\smash{{\SetFigFont{9}{10.8}{\familydefault}{\mddefault}{\updefault}$1$}}}}
\put(7842,2373){\makebox(0,0)[lb]{\smash{{\SetFigFont{9}{10.8}{\familydefault}{\mddefault}{\updefault}$1$}}}}
\put(9552,78){\makebox(0,0)[lb]{\smash{{\SetFigFont{9}{10.8}{\familydefault}{\mddefault}{\updefault}$\lambda_{t_{s_2}}$}}}}
\put(8562,1563){\makebox(0,0)[lb]{\smash{{\SetFigFont{9}{10.8}{\familydefault}{\mddefault}{\updefault}$\lambda_{t_{s_2}}$}}}}
\put(6717,1563){\makebox(0,0)[lb]{\smash{{\SetFigFont{9}{10.8}{\familydefault}{\mddefault}{\updefault}$\lambda_{t_{s_1}}$}}}}
\end{picture}
}}
 \caption{$\lambda$ positions for $ t_{z'} \neq  t_z$.}
\label{figure:oneLTwoSkoro}
\end{figure}

\begin{compactenum}
\item $ \lambda_{t_{s_1}},  \lambda_{t_{s_2}}$ both $\geq 1$.\\
Since the rate of decrease of $\min(h^{\infty}_{\lambda_{t_{s_1}}})$ is the same
as that of  $\min(h^{\infty}_{\lambda_{t_{s_2}}})$, in this case
$\delta^{\dagger}$ is when $\min(h^{\infty}_{\lambda_{t_{s_2}}})(\delta^{\dagger}) = 1$.
That is, when $\sphere(t_{s_2}, \delta^{\dagger}) = t_{z'}$.
Solving, we get $\delta^{\dagger} = \abs{t_{s_2} - t_{z'}}$.

\item $ \lambda_{t_{s_1}},  \lambda_{t_{s_2}}$ both between $0$ and $1$.\\
In this, the minimum $\delta$ is obtained when
 $\min(h^{\infty}_{\lambda_{t_{s_2}}}) =  \max(h^{\infty}_{\lambda_{t_{s_1}}})$.
Using Equation~\ref{equation:HInftyBoundaryOneDim}, we have two cases.
\begin{compactenum}
\item  $ (t_{z'} - t_z) > 0$. In this case
  $-\frac{\delta^{\dagger} - (t_{s_2}-t_z) }{t_{z'} - t_z} \, =\, 
  \frac{\delta^{\dagger} + (t_{s_1}-t_z) }{t_{z'} - t_z} $.
  Simplifying, $2\vdot\delta^{\dagger} =  (t_{s_2}-t_z)  -  (t_{s_1}-t_z)$, thus,
  $\delta^{\dagger} = (t_{s_2}-t_{s_1})/2$.
\item
  $(t_{z'} - t_z) < 0$. In this case
 $\frac{\delta^{\dagger} + (t_{s_2}-t_z) }{t_{z'} - t_z} \, =\, 
  -\frac{\delta^{\dagger} - (t_{s_1}-t_z) }{t_{z'} - t_z} $.
  Solving, we get  $\delta^{\dagger} = (t_{s_1}-t_{s_2})/2$.
\end{compactenum}
Putting the two cases, together,
\[\delta^{\dagger} =
\begin{cases} 
  (t_{s_2}-t_{s_1})/2& \text{ if } (t_{z'} - t_z) > 0\\
 (t_{s_1}-t_{s_2})/2& \text{ if } (t_{z'} - t_z) < 0.
\end{cases}
\]

Or, equivalently, $\delta^{\dagger} =  \abs{t_{s_2}-t_{s_1}}/2$.
\item $ \lambda_{t_{s_2}} > 1, $ and $0\leq \lambda_{t_{s_1}} < 1$\\
We need to consider two sub-case.
  \begin{compactenum}
  \item 
    If $\delta_1^1$ such that $\max(h^{\infty}_{\lambda_{t_{s_1}}}) (\delta_1^1) = 1$
    is less than, or equal to 
    $\delta_2^1$ such that $\min(h^{\infty}_{\lambda_{t_{s_2}}}) (\delta_1^1) = 1$,
    then, when the two intervals around $ \lambda_{t_{s_1}}$ and $ \lambda_{t_{s_2}}$
    intersect, the intersection point is going to be in $[0,1]$.
    Intuitively, the min boundary of $h^{\infty}_{\lambda_{t_{s_2}}}$ hits $1$
    before the max boundary of $h^{\infty}_{\lambda_{t_{s_1}}}$ does.
    In this case, the least $\delta$ is when
    the boundaries $\min(h^{\infty}_{\lambda_{t_{s_2}}})$
    and $  \max(h^{\infty}_{\lambda_{t_{s_1}}})$ intersect.
    The $\delta^{\dagger}$ value   can then be extracted as in the previous case.
    We have  $\sphere(t_{s_2}, \delta_1^1) = t_{z'}$; and
    $\sphere(t_{s_1}, \delta_1^1) = t_{z'}$; \emph{i.e.}
    $\delta_1^1 = \abs{t_{s_1} -  t_{z'}}$ and
    $\delta_2^1 = \abs{t_{s_2} -  t_{z'}}$.
    Thus, if $\abs{t_{s_1} -  t_{z'}} \leq  \abs{t_{s_2} -  t_{z'}}$, then
    $\delta^{\dagger} =  \abs{t_{s_2}-t_{s_1}}/2$.

  \item If $\delta_1^1$ is greater than $\delta_2^1$, where these values are as
    defined in the previous case.
    In this case, $  \max(h^{\infty}_{\lambda_{t_{s_1}}})$  hits
    $1$ before $\min(h^{\infty}_{\lambda_{t_{s_2}}})$ does, \emph{i.e.}
     $  \max(h^{\infty}_{\lambda_{t_{s_1}}})$  hits
     $1$ before  $\min(h^{\infty}_{\lambda_{t_{s_2}}})$ reaches 
     $  \max(h^{\infty}_{\lambda_{t_{s_1}}})$.
     Thus, the least $\delta$ occurs when
     $\min(h^{\infty}_{\lambda_{t_{s_2}}}) (\delta^{\dagger})= 1$.
     That is, when $\sphere(t_{s_2}, \delta^{\dagger}) = t_{z'}$.
     Thus, $\delta^{\dagger} = \abs{t_{s_2} - t_{z'}}$.
    
   \end{compactenum}

Putting the two cases together, we get
 \[\delta^{\dagger} =
\begin{cases}
 \abs{t_{s_2} - t_{z'}} & \text{ if } \abs{t_{s_1} -  t_{z'}} >  \abs{t_{s_2} -  t_{z'}}\\
 \abs{t_{s_2}-t_{s_1}}/2
    \text{ otherwise. }
  \end{cases}
  \]

\item $ \lambda_{t_{s_2}} > 1, $ and $\lambda_{t_{s_1}} < 0 $\\
  The analysis of this case depends on when the boundaries of 
  $h^{\infty}_{\lambda_{t_{s_1}}}$ and $h^{\infty}_{\lambda_{t_{s_2}}}$ hit
  $0$ and $1$.
  \begin{compactenum}
  \item 
    Suppose the min boundary of $h^{\infty}_{\lambda_{t_{s_2}}}$ hits
    $0$ before the max boundary of $h^{\infty}_{\lambda_{t_{s_1}}}$ hits $0$.
    Let us call the respective $\delta$ values $\delta_1^0$ and $\delta_2^0$.
    In this case, the least desired $\delta$ is going to be when
    the max boundary of $h^{\infty}_{\lambda_{t_{s_1}}}$ hits $0$.
    Thus, if $ \abs{t_{s_1} -  t_{z} }<  \abs{t_{s_2} -  t_{z}}$, then
    $\delta^{\dagger} =  \abs{t_{s_1} -  t_{z} }$.
  \item 
    If the previous case does not hold, and if 
    the min boundary of $h^{\infty}_{\lambda_{t_{s_2}}}$ hits
    $1$ before the max boundary of $h^{\infty}_{\lambda_{t_{s_1}}}$ hits $1$,
    then the intersection point of the two intervals in going to be in $[0,1]$.
    The value can be obtained from one of the previous case.
    Thus, if $ \abs{t_{s_1} -  t_{z} }\geq  \abs{t_{s_2} -  t_{z}}$, and
    $ \abs{t_{s_1} -  t_{z'}} \leq  \abs{t_{s_2} -  t_{z'}}$, then
 $\delta^{\dagger} =  \abs{t_{s_2}-t_{s_1}}/2$.
  \item If the previous two case do not hold, \emph{i.e.}
    $\abs{t_{s_1} -  t_{z} }\geq  \abs{t_{s_2} -  t_{z}}$, and
    $ \abs{t_{s_1} -  t_{z'}} >  \abs{t_{s_2} -  t_{z'}}$, then
    the minimum $\delta$ occurs when the min boundary 
    of $h^{\infty}_{\lambda_{t_{s_2}}}$ hits $1$, \emph{i.e.} when
    $\delta^{\dagger} =  \abs{t_{s_2} - t_{z'}}$.
  \end{compactenum}
  Putting the above three cases together,
  \[
  \delta^{\dagger}\left( t_{s_1},  t_{s_2}, \myline(t_{z'},t_{z})\right)  =
  \begin{cases}
    \abs{t_{s_1} -  t_{z} } & \text{ if }  \abs{t_{s_1} -  t_{z} }<  \abs{t_{s_2} -  t_{z}}\\
     \abs{t_{s_2}-t_{s_1}}/2
    & \text{ if }  \abs{t_{s_1} -  t_{z} }\geq   \abs{t_{s_2} -  t_{z}} \text{ and }
    \abs{t_{s_1} -  t_{z'}} \leq  \abs{t_{s_2} -  t_{z'}}\\
    \abs{t_{s_2} - t_{z'}} & 
    \text{ otherwise. }
  \end{cases}
  \]

\end{compactenum} 

\medskip
\noindent $\blacktriangleright$ 
Suppose $\tuple{z,t_z} \neq  \tuple{z',t_{z'}}$, but
$t_z= t_{z'}$ (and thus $ z' \neq  z$).
Then, 
\[\delta^* =
\min_{\delta \geq 0} \left(\text{There exists } \lambda \in [0,1]
\text{ such that }
\left\{
\begin{array}{l}
\norm{\vec{z}-\vec{s_1} + \lambda\vdot(\vec{z'}-\vec{z})}_{L_2} \, \leq \delta; 
\text{ and}\\
\abs{t_z-t_{s_1}}\leq  \delta; \text{ and}\\
\norm{\vec{z}-\vec{s_2} + \lambda\vdot(\vec{z'}-\vec{z})}_{L_2} \, \leq \delta; 
\text{ and}\\
\abs{t_z-t_{s_2}}\leq  \delta
\end{array}
\right\}
\right)
\]
The above is equivalent to:
\[
\delta^* =
\min_{\delta \geq \max\left(\abs{t_z-t_{s_1}}, \abs{t_z-t_{s_2}}\right)} \left(\text{There exists } \lambda \in [0,1]
\text{ such that }
\left\{
\begin{array}{l}
\norm{\vec{z}-\vec{s_1} + \lambda\vdot(\vec{z'}-\vec{z})}_{L_2} \, \leq \delta; 
\text{ and}\\
\norm{\vec{z}-\vec{s_2} + \lambda\vdot(\vec{z'}-\vec{z})}_{L_2} \, \leq \delta; 
\end{array}
\right\}
\right)
\]
The above is equilavalent to
\[
\delta^* = \max\left(
\abs{t_z-t_{s_1}}, \abs{t_z-t_{s_2}}, \,
 \min_{\delta \geq 0}\set{\delta \mid 
      \ball_{L_2}(s_1,\delta) \cap \ball_{L_2}(s_2, \delta) \cap \myline(z,z') \neq\emptyset}
\right)
\]
This can be computed using the algorithms for the $L_2$-norm.

\medskip
\noindent $\blacktriangleright$ 
Suppose $\tuple{z,t_z} \neq  \tuple{z',t_{z'}}$, and 
$t_z\neq t_{z'}$, and  $ z' \neq  z$.
The analysis of this case depends on the relative positions of
$\lambda_p^{s_1}, \lambda_{t_{s_1}}, \lambda_p^{s_2}$, and  $\lambda_{t_{s_1}}$.
Figure~\ref{figure:oneLTwoSkoroGeneral} illustrates 
some of the $\lambda$ value placements).
Note that  $\lambda_p^{s_1} $ and  $\lambda_p^{s_2} $ belong to $[0,1]$
by definition.

\begin{figure}[h]
\strut\centerline{\setlength{\unitlength}{0.00061242in}
\begingroup\makeatletter\ifx\SetFigFont\undefined%
\gdef\SetFigFont#1#2#3#4#5{%
  \reset@font\fontsize{#1}{#2pt}%
  \fontfamily{#3}\fontseries{#4}\fontshape{#5}%
  \selectfont}%
\fi\endgroup%
{\renewcommand{\dashlinestretch}{30}
\begin{picture}(9684,2877)(0,-10)
\path(3710,2317)(3710,2137)
\path(3080,2317)(3080,2137)
\path(140.000,2347.000)(20.000,2317.000)(140.000,2287.000)
\path(20,2317)(4070,2317)
\path(3950.000,2287.000)(4070.000,2317.000)(3950.000,2347.000)
\path(2720,2497)(2720,2317)
\path(920,2497)(920,2317)
\blacken\path(1364,2328)(1325,2163)(1403,2163)
	(1364,2320)(1364,2328)
\path(1364,2328)(1325,2163)(1403,2163)
	(1364,2320)(1364,2328)
\blacken\path(2391,2336)(2352,2171)(2430,2171)
	(2391,2328)(2391,2336)
\path(2391,2336)(2352,2171)(2430,2171)
	(2391,2328)(2391,2336)
\path(3702,528)(3702,348)
\path(132.000,558.000)(12.000,528.000)(132.000,498.000)
\path(12,528)(4062,528)
\path(3942.000,498.000)(4062.000,528.000)(3942.000,558.000)
\path(2712,708)(2712,528)
\path(912,708)(912,528)
\blacken\path(1356,539)(1317,374)(1395,374)
	(1356,531)(1356,539)
\path(1356,539)(1317,374)(1395,374)
	(1356,531)(1356,539)
\blacken\path(2383,547)(2344,382)(2422,382)
	(2383,539)(2383,547)
\path(2383,547)(2344,382)(2422,382)
	(2383,539)(2383,547)
\path(1910,528)(1910,348)
\path(9312,2343)(9312,2163)
\path(5742.000,2373.000)(5622.000,2343.000)(5742.000,2313.000)
\path(5622,2343)(9672,2343)
\path(9552.000,2313.000)(9672.000,2343.000)(9552.000,2373.000)
\path(8322,2523)(8322,2343)
\path(6522,2523)(6522,2343)
\blacken\path(7993,2362)(7954,2197)(8032,2197)
	(7993,2354)(7993,2362)
\path(7993,2362)(7954,2197)(8032,2197)
	(7993,2354)(7993,2362)
\path(6635,2358)(6635,2178)
\blacken\path(7158,2351)(7119,2186)(7197,2186)
	(7158,2343)(7158,2351)
\path(7158,2351)(7119,2186)(7197,2186)
	(7158,2343)(7158,2351)
\blacken\path(881,2524)(920,2524)(920,2305)
	(881,2305)(881,2524)
\path(881,2524)(920,2524)(920,2305)
	(881,2305)(881,2524)
\blacken\path(888,748)(927,748)(927,529)
	(888,529)(888,748)
\path(888,748)(927,748)(927,529)
	(888,529)(888,748)
\blacken\path(2688,755)(2727,755)(2727,536)
	(2688,536)(2688,755)
\path(2688,755)(2727,755)(2727,536)
	(2688,536)(2688,755)
\blacken\path(8299,2545)(8338,2545)(8338,2326)
	(8299,2326)(8299,2545)
\path(8299,2545)(8338,2545)(8338,2326)
	(8299,2326)(8299,2545)
\blacken\path(2712,2515)(2751,2515)(2751,2296)
	(2712,2296)(2712,2515)
\path(2712,2515)(2751,2515)(2751,2296)
	(2712,2296)(2712,2515)
\blacken\path(6507,2538)(6546,2538)(6546,2319)
	(6507,2319)(6507,2538)
\path(6507,2538)(6546,2538)(6546,2319)
	(6507,2319)(6507,2538)
\blacken\path(6477,746)(6516,746)(6516,527)
	(6477,527)(6477,746)
\path(6477,746)(6516,746)(6516,527)
	(6477,527)(6477,746)
\blacken\path(8276,755)(8315,755)(8315,536)
	(8276,536)(8276,755)
\path(8276,755)(8315,755)(8315,536)
	(8276,536)(8276,755)
\path(5720.000,588.000)(5600.000,558.000)(5720.000,528.000)
\path(5600,558)(9650,558)
\path(9530.000,528.000)(9650.000,558.000)(9530.000,588.000)
\path(8300,738)(8300,558)
\path(6500,738)(6500,558)
\path(7182,565)(7182,385)
\blacken\path(6764,565)(6725,400)(6803,400)
	(6764,557)(6764,565)
\path(6764,565)(6725,400)(6803,400)
	(6764,557)(6764,565)
\blacken\path(7634,573)(7595,408)(7673,408)
	(7634,565)(7634,573)
\path(7634,573)(7595,408)(7673,408)
	(7634,565)(7634,573)
\path(8075,565)(8075,385)
\put(3710,1867){\makebox(0,0)[lb]{\smash{{\SetFigFont{9}{10.8}{\familydefault}{\mddefault}{\updefault}$\lambda_{t_{s_2}}$}}}}
\put(2675,2677){\makebox(0,0)[lb]{\smash{{\SetFigFont{9}{10.8}{\familydefault}{\mddefault}{\updefault}$1$}}}}
\put(830,2677){\makebox(0,0)[lb]{\smash{{\SetFigFont{9}{10.8}{\familydefault}{\mddefault}{\updefault}$0$}}}}
\put(1262,1889){\makebox(0,0)[lb]{\smash{{\SetFigFont{9}{10.8}{\familydefault}{\mddefault}{\updefault}$\lambda_p^{s_1}$}}}}
\put(3035,1867){\makebox(0,0)[lb]{\smash{{\SetFigFont{9}{10.8}{\familydefault}{\mddefault}{\updefault}$\lambda_{t_{s_1}}$}}}}
\put(2345,1877){\makebox(0,0)[lb]{\smash{{\SetFigFont{9}{10.8}{\familydefault}{\mddefault}{\updefault}$\lambda_p^{s_2}$}}}}
\put(2667,888){\makebox(0,0)[lb]{\smash{{\SetFigFont{9}{10.8}{\familydefault}{\mddefault}{\updefault}$1$}}}}
\put(822,888){\makebox(0,0)[lb]{\smash{{\SetFigFont{9}{10.8}{\familydefault}{\mddefault}{\updefault}$0$}}}}
\put(1254,100){\makebox(0,0)[lb]{\smash{{\SetFigFont{9}{10.8}{\familydefault}{\mddefault}{\updefault}$\lambda_p^{s_1}$}}}}
\put(3702,78){\makebox(0,0)[lb]{\smash{{\SetFigFont{9}{10.8}{\familydefault}{\mddefault}{\updefault}$\lambda_{t_{s_2}}$}}}}
\put(1826,78){\makebox(0,0)[lb]{\smash{{\SetFigFont{9}{10.8}{\familydefault}{\mddefault}{\updefault}$\lambda_{t_{s_1}}$}}}}
\put(2344,89){\makebox(0,0)[lb]{\smash{{\SetFigFont{9}{10.8}{\familydefault}{\mddefault}{\updefault}$\lambda_p^{s_2}$}}}}
\put(8277,2703){\makebox(0,0)[lb]{\smash{{\SetFigFont{9}{10.8}{\familydefault}{\mddefault}{\updefault}$1$}}}}
\put(6432,2703){\makebox(0,0)[lb]{\smash{{\SetFigFont{9}{10.8}{\familydefault}{\mddefault}{\updefault}$0$}}}}
\put(6582,1897){\makebox(0,0)[lb]{\smash{{\SetFigFont{9}{10.8}{\familydefault}{\mddefault}{\updefault}$\lambda_{t_{s_1}}$}}}}
\put(7111,1909){\makebox(0,0)[lb]{\smash{{\SetFigFont{9}{10.8}{\familydefault}{\mddefault}{\updefault}$\lambda_p^{s_1}$}}}}
\put(7948,1875){\makebox(0,0)[lb]{\smash{{\SetFigFont{9}{10.8}{\familydefault}{\mddefault}{\updefault}$\lambda_p^{s_2}$}}}}
\put(9259,1885){\makebox(0,0)[lb]{\smash{{\SetFigFont{9}{10.8}{\familydefault}{\mddefault}{\updefault}$\lambda_{t_{s_2}}$}}}}
\put(8255,918){\makebox(0,0)[lb]{\smash{{\SetFigFont{9}{10.8}{\familydefault}{\mddefault}{\updefault}$1$}}}}
\put(6410,918){\makebox(0,0)[lb]{\smash{{\SetFigFont{9}{10.8}{\familydefault}{\mddefault}{\updefault}$0$}}}}
\put(8082,115){\makebox(0,0)[lb]{\smash{{\SetFigFont{9}{10.8}{\familydefault}{\mddefault}{\updefault}$\lambda_{t_{s_2}}$}}}}
\put(6692,123){\makebox(0,0)[lb]{\smash{{\SetFigFont{9}{10.8}{\familydefault}{\mddefault}{\updefault}$\lambda_p^{s_1}$}}}}
\put(7129,123){\makebox(0,0)[lb]{\smash{{\SetFigFont{9}{10.8}{\familydefault}{\mddefault}{\updefault}$\lambda_{t_{s_1}}$}}}}
\put(7602,135){\makebox(0,0)[lb]{\smash{{\SetFigFont{9}{10.8}{\familydefault}{\mddefault}{\updefault}$\lambda_p^{s_2}$}}}}
\end{picture}
}}
 \caption{A few of the $\lambda$ positions for $ z' \neq  z$ and $ t_{z'} \neq  t_z$.}
\label{figure:oneLTwoSkoroGeneral}
\end{figure}
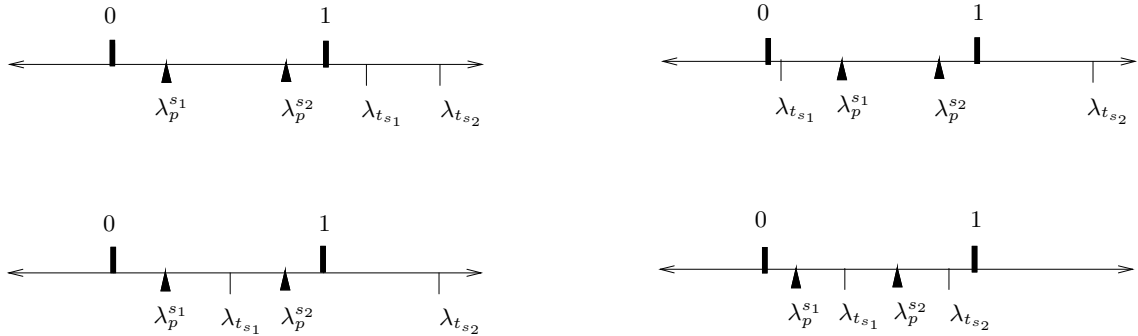

Consider Equation~\ref{equation:BallsLineLTwoSkoro}.
There are four innermost clauses.
Suppose for each $\delta$ we define a set $H(\delta)$ which
denote the $\lambda$ points in $[0,1]$ that satisfy all the four contraints 
of  Equation~\ref{equation:BallsLineLTwoSkoro} as follows:
\[
H(\delta) = 
\left\{
\lambda \ \left|\ 
\begin{array}{l}
\norm{\vec{z}-\vec{s_1} + \lambda\vdot(\vec{z'}-\vec{z})}_{L_2} \, \leq \delta; 
\text{ and}\\
\abs{t_z-t_{s_1} + \lambda\vdot(t_{z'}-t_{z})}\,  \leq \delta; \text{ and}\\
\norm{\vec{z}-\vec{s_2} + \lambda\vdot(\vec{z'}-\vec{z})}_{L_2} \, \leq \delta; 
\text{ and}\\
\abs{t_z-t_{s_2} + \lambda\vdot(t_{z'}-t_{z})}\,  \leq \delta;
\text{ and}\\
\lambda \in [0,1]
\end{array}
\right .
\right\}.
\]
Then, $\delta^*$, the optimal value of the
solution of Equation~\ref{equation:BallsLineLTwoSkoro} is:
\[
\delta^* = \min_{\delta} \left(H(\delta) \neq \emptyset\right).\]
We can break down $H(\delta)$ into sets defined by individual constraints as follows:
\begin{align*}
H_1(\delta) & = \set{\lambda \mid \norm{\vec{z}-\vec{s_1} + \lambda\vdot(\vec{z'}-\vec{z})}_{L_2} \, \leq \delta}\\
H_2(\delta) & = \set{\lambda \mid \abs{t_z-t_{s_1} + \lambda\vdot(t_{z'}-t_{z})}\,  \leq \delta}\\
H_3(\delta) & = \set{\lambda \mid \norm{\vec{z}-\vec{s_2} + \lambda\vdot(\vec{z'}-\vec{z})}_{L_2} \, \leq \delta}\\
H_4(\delta) & = \set{\lambda \mid \abs{t_z-t_{s_2} + \lambda\vdot(t_{z'}-t_{z})}\,  \leq \delta}\\
H(\delta) & = H_1(\delta) \, \cap \,  H_2(\delta) \, \cap \,  H_3(\delta) \, \cap \,  H_4(\delta) \, \cap \,  [0,1]
\end{align*}
For $1\leq i < j \leq 4$, consider the equation defined by:
\[
\delta^*_{i,j} = \min_{\delta} \big(H_i(\delta)\, \cap\, H_j(\delta)\, \cap\, [0,1]
\ \neq \emptyset\big)
\]
It is clear that
\begin{equation}
 \label{Equation:DeltaStarGreaterMax}
\delta^* \geq \max_{1\leq i < j \leq 4} \delta^*_{i,j}.
\end{equation}
This is because if 
$\lambda \in H(\delta) $, then $\lambda \in H_i(\delta)\, \cap\, H_j(\delta)\, \cap\, [0,1]$.
We show that we in fact have equality:
\begin{equation}
 \label{Equation:DeltaStarEqualMax}
\delta^* =  \max_{1\leq i < j \leq 4} \delta^*_{i,j}.
\end{equation}
Observe that for any $\delta \geq  \max_{1\leq i < j \leq 4} \delta^*_{i,j}$, we have:
\begin{compactitem}
\item $\big(H_i(\delta)  \cap [0,1]\big)\,  \cap\, \big(H_j(\delta)\, \cap\, [0,1]\big)
\ \neq \emptyset$ for all $1\leq i < j \leq 4$ since 
$\delta >  \delta^*_{i,j}$.
\item $H_i(\delta)  \cap [0,1]$ is an interval $[\alpha_i, \alpha_i']$ with
$\alpha_i' \geq \alpha_i$ for all $1\leq i\leq 4$.
\end{compactitem}
The second fact, which states that
the set of $\lambda$ values on a line such that the corresponding line points are at
most $\delta$ away from some some point is a closed interval,
 follows from
Equation~\ref{equation:HInftyBoundaryOneDim} and
Lemma~\ref{lemma:BallLineIntersect}.
Thus, from Helley's theorem (Proposition~\ref{proposition:Helley}) with $d=1$, we have that
$\displaystyle\cap_{i=1}^4\left(H_i(\delta)  \cap [0,1]\right)$ is not
empty.
This means that $\delta^* \leq \max_{1\leq i < j \leq 4} \delta^*_{i,j}$.
Using Equation~\ref {Equation:DeltaStarGreaterMax}, we get
$\delta^* =  \max_{1\leq i < j \leq 4} \delta^*_{i,j}$.

We now observe the following facts.
\begin{compactitem}
\item 
$ \min_{\delta} \big(H_1(\delta)\, \cap\, H_2(\delta)\, \cap\, [0,1]
\ \neq \emptyset\big)$ which equals
\[
\min_{\delta \geq 0} \left(\text{There exists } \lambda \in [0,1]
\text{ such that }
\left\{
\begin{array}{l}
\norm{\vec{z}-\vec{s_1} + \lambda\vdot(\vec{z'}-\vec{z})}_{L_2} \, \leq \delta; 
\text{ and}\\
\abs{t_z-t_{s_1} + \lambda\vdot(t_{z'}-t_{z})}\,  \leq \delta
\end{array}
\right\}
\right)
\]
is equal to
$\dist_{L_2^{\skoro}}\big(\tuple{s_1,t_{s_1}}, \myline(\tuple{z,t_z}, \tuple{z,t_{z'}})\big)$
from Equation~\ref{equation:PointLineLTwoSkoro}.

\item $ \min_{\delta} \big(H_1(\delta)\, \cap\, H_3(\delta)\, \cap\, [0,1]
\ \neq \emptyset\big)$ which equals
\[
\min_{\delta \geq 0} \left(\text{There exists } \lambda \in [0,1]
\text{ such that }
\left\{
\begin{array}{l}
\norm{\vec{z}-\vec{s_1} + \lambda\vdot(\vec{z'}-\vec{z})}_{L_2} \, \leq \delta; 
\text{ and}\\
\norm{\vec{z}-\vec{s_2} + \lambda\vdot(\vec{z'}-\vec{z})}_{L_2} \, \leq \delta; 
\end{array}
\right\}
\right)
\]
is equal to $\min_{\delta \geq 0}\set{\delta \mid 
 \ball_{L_2}(s_1,\delta) \cap \ball_{L_2}(s_2, \delta) \cap \myline(z,z') \neq\emptyset}$.

\item $ \min_{\delta} \big(H_1(\delta)\, \cap\, H_4(\delta)\, \cap\, [0,1]
\ \neq \emptyset\big)$ which equals
\[
\min_{\delta \geq 0} \left(\text{There exists } \lambda \in [0,1]
\text{ such that }
\left\{
\begin{array}{l}
\norm{\vec{z}-\vec{s_1} + \lambda\vdot(\vec{z'}-\vec{z})}_{L_2} \, \leq \delta; 
\text{ and}\\
\abs{t_z-t_{s_2} + \lambda\vdot(t_{z'}-t_{z})}\,  \leq \delta
\end{array}
\right\}
\right)
\]
is equal to
$\dist_{L_2^{\skoro}}\big(\tuple{s_1,t_{s_2}}, \myline(\tuple{z,t_z}, \tuple{z,t_{z'}})\big)$
from Equation~\ref{equation:PointLineLTwoSkoro}.

\item $ \min_{\delta} \big(H_2(\delta)\, \cap\, H_3(\delta)\, \cap\, [0,1]
\ \neq \emptyset\big)$ which equals
\[
\min_{\delta \geq 0} \left(\text{There exists } \lambda \in [0,1]
\text{ such that }
\left\{
\begin{array}{l}
\abs{t_z-t_{s_1} + \lambda\vdot(t_{z'}-t_{z})}\,  \leq \delta; \text{ and}\\
\norm{\vec{z}-\vec{s_2} + \lambda\vdot(\vec{z'}-\vec{z})}_{L_2} \, \leq \delta
\end{array}
\right\}
\right)
\]
which is equal to
$\dist_{L_2^{\skoro}}\big(\tuple{s_2,t_{s_1}}, \myline(\tuple{z,t_z}, \tuple{z,t_{z'}})\big)$
from Equation~\ref{equation:PointLineLTwoSkoro}.

\item 
$ \min_{\delta} \big(H_2(\delta)\, \cap\, H_4(\delta)\, \cap\, [0,1]
\ \neq \emptyset\big)$ which equals
\[
\min_{\delta \geq 0} \left(\text{There exists } \lambda \in [0,1]
\text{ such that }
\left\{
\begin{array}{l}
\abs{t_z-t_{s_1} + \lambda\vdot(t_{z'}-t_{z})}\,  \leq \delta; \text{ and}\\
\abs{t_z-t_{s_2} + \lambda\vdot(t_{z'}-t_{z})}\,  \leq \delta
\end{array}
\right\}
\right)
\]
is equal to $\delta^{\dagger}\left( t_{s_1},  t_{s_2}, \myline(t_{z'},t_{z})\right)  $ from
Equation~\ref{gather:OnlyTwoLinesSkoroDagger}.

\item 
$ \min_{\delta} \big(H_3(\delta)\, \cap\, H_4(\delta)\, \cap\, [0,1]
\ \neq \emptyset\big)$ which equals
\[
\min_{\delta \geq 0} \left(\text{There exists } \lambda \in [0,1]
\text{ such that }
\left\{
\begin{array}{l}
\norm{\vec{z}-\vec{s_2} + \lambda\vdot(\vec{z'}-\vec{z})}_{L_2} \, \leq \delta; 
\text{ and}\\
\abs{t_z-t_{s_2} + \lambda\vdot(t_{z'}-t_{z})}\,  \leq \delta
\end{array}
\right\}
\right)
\]
is equal to
$\dist_{L_2^{\skoro}}\big(\tuple{s_2,t_{s_2}}, \myline(\tuple{z,t_z}, \tuple{z,t_{z'}})\big)$
from Equation~\ref{equation:PointLineLTwoSkoro}.

\end{compactitem}

Thus, in case
$\tuple{z,t_z} \neq  \tuple{z',t_{z'}}$, and 
$t_z\neq t_{z'}$, and  $ z' \neq  z$, putting everything together, we have:
\[
\delta^* = \max\left(
\begin{array}{c}
\dist_{L_2^{\skoro}}\big(\tuple{s_1,t_{s_1}}, \myline(\tuple{z,t_z}, \tuple{z,t_{z'}})\big)\, , \\
\min_{\delta \geq 0}\set{\delta \mid 
 \ball_{L_2}(s_1,\delta) \cap \ball_{L_2}(s_2, \delta) \cap \myline(z,z') \neq\emptyset}\, ,\\
 \dist_{L_2^{\skoro}}\big(\tuple{s_1,t_{s_2}}, \myline(\tuple{z,t_z}, \tuple{z,t_{z'}})\big)\, , \\
 \dist_{L_2^{\skoro}}\big(\tuple{s_2,t_{s_1}}, \myline(\tuple{z,t_z}, \tuple{z,t_{z'}})\big)\, , \\
\delta^{\dagger}\left( t_{s_1},  t_{s_2}, \myline(t_{z'},t_{z})\right)\, ,\\
\dist_{L_2^{\skoro}}\big(\tuple{s_2,t_{s_2}}, \myline(\tuple{z,t_z}, \tuple{z,t_{z'}})
\end{array}
\right).
\]
We note that we have shown how to compute 
all the individual values inside the max set before.

The following 
function~\ref{function:BallsSkoroLTwo} 
combines everything together giving us 
Proposition~\ref{proposition:BallIntersectFinalLTwoSkoro}.

\begin{function}
  \SetKwInOut{Input}{Input}
  \SetKwInOut{Output}{Output}
  \Input{Points $\tuple{s_1,t_{s_1}}, \tuple{s_1,t_{s_2}},  \tuple{z,t_z}, \tuple{z',t_{z'}}$ in $\reals^n\times\reals$}
  \Output{Least $\delta$ such that 
    $\ball(\tuple{s_1,t_{s_1}} \delta) \cap \ball(\tuple{s_2,t_{s_1}}, \delta)  
    \cap \myline( \tuple{z,t_z}, \tuple{z',t_{z'}})$ is non-empty in $L_2^{\skoro}$ norm}
  \If(\tcp*[f]{ensure $ \lambda_{t_{s_2}} \geq 0$ and
      $ \lambda_{t_{s_2}} \geq\lambda_{t_{s_1}}$})
    {$ t_{z'} \neq  t_z$ }{
      $ \lambda_{t_{s_1}} = \frac{t_{s_1} - t_z }{t_{z'} -  t_z }$\;
      $ \lambda_{t_{s_2}} = \frac{t_{s_2} - t_z }{t_{z'} -  t_z }$\;
      \If{$ \lambda_{t_{s_2}} < 0$}{
        swap $ \tuple{z,t_z}$ with $\tuple{z',t_{z'}}$\;
        $ \lambda_{t_{s_1}} = 1-  \lambda_{t_{s_1}}$\;
        $ \lambda_{t_{s_2}} = 1-  \lambda_{t_{s_2}}$\;
      }
      \If{$ \lambda_{t_{s_2}} <  \lambda_{t_{s_1}}$}{
        swap $t_{s_2}$ with $t_{s_1}$\;
        swap $ \lambda_{t_{s_2}} $ with $  \lambda_{t_{s_1}}$\;
      } 
    }
    \If(\tcp*[f]{we need $\delta^{\dagger}$   twice})
    {$ t_{z'} \neq  t_z$ }{
      $\delta^{\dagger}:=  \begin{cases}
        \abs{t_{s_1} -  t_{z} } & \text{ if }  \abs{t_{s_1} -  t_{z} }<  \abs{t_{s_2} -  t_{z}}\\
        \abs{t_{s_2}-t_{s_1}}/2
        & \text{ if }  \abs{t_{s_1} -  t_{z} }\geq   \abs{t_{s_2} -  t_{z}} \text{ and }
        \abs{t_{s_1} -  t_{z'}} \leq  \abs{t_{s_2} -  t_{z'}}\\
        \abs{t_{s_2} - t_{z'}} & 
        \text{ otherwise. }
      \end{cases}$\;
    }
    \Switch{ $\tuple{s_1,t_{s_1}}, \tuple{s_1,t_{s_2}},  \tuple{z,t_z}, \tuple{z',t_{z'}}$}{
      \lCase{$\tuple{s_1,t_{s_1}} =  \tuple{s_1,t_{s_2}}$}{
        \Return{$\dist_{L_2^{\skoro}}\left(\tuple{s_1,t_{s_1}} , \myline( \tuple{z,t_z}, \tuple{z',t_{z'}})\right)$\;}
      }
      \lCase{ $\tuple{z,t_z} = \tuple{z',t_{z'}}$}{
        \Return{ $\max\big(\norm{s_1-z}_{L_2}, \abs{t_{s_1}- t_z}, 
          \norm{s_2-z}_{L_2}, \abs{t_{s_2}- t_z}\big)$\;}}
      \lCase{$z= z'$ and $ t_{z'} \neq  t_z$}{
        \Return{ $\max\big( \norm{z-s_1}_{L_2}, \norm{z-s_2}_{L_2}, \delta^{\dagger}\big)$\;}
      }
      \uCase{$z\neq  z'$ and $ t_{z'} = t_z$}{
        $ \delta_{L_2}:= \min_{\delta \geq 0}\set{\delta \mid 
          \ball_{L_2}(s_1,\delta) \cap \ball_{L_2}(s_2, \delta) \cap \myline(z,z') \neq\emptyset}$\;
        \tcp*[f]{ $\delta_{L_2}$ can be computed using procedure for $L_2$}\\
        \Return{$\max\big(\abs{t_z-t_{s_1}}, \abs{t_z-t_{s_2}}, \delta_{L_2}\big)$\;}
 }
 \uCase{$z\neq  z'$ and $ t_{z'} \neq t_z$}{
   $ \delta_{L_2}:= \min_{\delta \geq 0}\set{\delta \mid 
     \ball_{L_2}(s_1,\delta) \cap \ball_{L_2}(s_2, \delta) \cap \myline(z,z') \neq\emptyset}$\;
   \Return{$
     \max\left(
       \begin{array}{c}
         \dist_{L_2^{\skoro}}\big(\tuple{s_1,t_{s_1}}, \myline(\tuple{z,t_z}, \tuple{z,t_{z'}})\big)\, , \\
         \delta_{L_2}\, \\
         \dist_{L_2^{\skoro}}\big(\tuple{s_1,t_{s_2}}, \myline(\tuple{z,t_z}, \tuple{z,t_{z'}})\big)\, , \\
         \dist_{L_2^{\skoro}}\big(\tuple{s_2,t_{s_1}}, \myline(\tuple{z,t_z}, \tuple{z,t_{z'}})\big)\, , \\
          \delta^{\dagger}\, ,\\
          \dist_{L_2^{\skoro}}\big(\tuple{s_2,t_{s_2}}, \myline(\tuple{z,t_z}, \tuple{z,t_{z'}})
        \end{array}
      \right)$\;}

  }
 
  }
  \caption{$\Phi_{L_2^{\skoro}}$($\tuple{s_1,t_{s_1}}, \tuple{s_1,t_{s_2}},  \tuple{z,t_z}, \tuple{z',t_{z'}}$)}
  \label{function:BallsSkoroLTwo}
\end{function}

\begin{proposition}[Computation of least $\delta$ such that 
$\ball(\tuple{s_1,t_{s_1}} \delta) \cap \ball(\tuple{s_2,t_{s_1}}, \delta)  
\cap \myline( \tuple{z,t_z}, \tuple{z',t_{z'}})$ is non-empty: $L_2^{\skoro}$]
\label{proposition:BallIntersectFinalLTwoSkoro}
Given points $\tuple{s_1,t_{s_1}},  \tuple{s_2,t_{s_2}}, \tuple{z,t_z}, \tuple{z',t_{z'}}$
with $s_1,s_2, z,z'\in \reals^n$ and $t_{s_1}, t_{s_2}, ,t_z, t_{z'}\in \reals$.
Function~\ref{function:BallsSkoroLTwo} computes
\[\min_{\delta \geq 0}\set{\delta \mid 
\ball(\tuple{s_1,t_{s_1}}, \delta) \cap \ball(\tuple{s_2,t_{s_2}}, \delta)  \cap 
\myline(\tuple{z,t_z}, \tuple{z,t_{z'}}) \neq \emptyset}\]
for the $L_2^{\skoro}$ norm.
\qed
\end{proposition}

\section{The Skorokhod Distance Algorithm}
\label{section:SkorokhodFinal}

Using the reduction from Proposition~\ref{proposition:SkoroToFrechet}
and the algorithm from Theorem~\ref{theorem:FrechetAlgo}, 
together with the procedures from
Propositions~\ref{proposition:PointLineLOneSkoro}, \ref{proposition:BallsIntersectLOneSkoro}, \ref{proposition:PointLineLTwoSkoro}, \ref{proposition:BallIntersectFinalLTwoSkoro}, \ref{proposition:DistancePointLineInfty} and~\ref{proposition:BallsIntersectLInfty}
for computing geometric primitives,
we obtain
the complete algorithm for computing the continuous time  Skorokhod distance between
two polygonal traces for the $L_1, L_2$ and $L_{\infty}$ norms.

\begin{theorem}
\label{theorem:SkoroFinal}
Let $x : [0,T_x] \mapsto \reals^n$  and $y : [0,T_y] \mapsto \reals^n$ be two polygonal  traces  with $m_x$ and $m_y$ affine segments respectively.
The continuous time Skorokhod distance between them, denoted
$\skoro_{\chi}(x, y) $ for the norm $\chi\in \set{L_1, L_2, L_{\infty}}$ can be computed in time:
\[
O\Big(
\left(m_y\vdot m_x^2 + m_x\vdot m_y^2\right)\vdot
\big( P(\chi)  + \log(m_y\vdot m_x) \big)\ 
+\ 
m_x\vdot m_y \vdot H(\chi) 
 \Big)
\]
where
\begin{compactitem}
\item for $L_1$, we have $P(L_1) = \LP(n)$ and $H(L_1) = n^2$.
\item for $L_2$, we have  $P(L_2) = n $ and   $H(L_2) = n$.
\item for $L_{\infty}$, we have $P(L_{\infty}) = \LP(n)$ and $H(L_{\infty}) = n$.
\end{compactitem}
and where $\LP(n)$ is the (polynomial-time) upper bound for linear programming.
The corresponding decision problem can be solved in time
$O(m_x\vdot m_y \vdot H(\chi) )$.
\qed
\end{theorem}

\smallskip\noindent\textbf{The Skorokhod  Distance with Windows.}
As for the \frechet distance, it is often of interest to
apply a window of $W$ in the the retimings.
The complexity of computing the value of the Skorokhod distance  in this case is 
$O\Big(
W^2\vdot M \vdot \big(P(\chi) + \log\left(W\vdot M \right)\big)+
W\vdot M\vdot \log\left(W\vdot M \right)\cdot H(\chi) 
\Big)
$, where $M= \max(m_f, m_g )$; and
$H(\chi),$ and $P(\chi)$ are as in Theorem~\ref{theorem:SkoroFinal}.
The corresponding decision problem runs in time $O(M\vdot W^2\vdot H(\chi) )$.
If $W$ can be taken to be a constant, the value computation
algorithm takes 
$O\Big(  M \cdot P(\chi)  + M\cdot \log(M)\cdot H(\chi) \Big)$ time,
and the decision problem can be solved in 
$O(M \vdot H(\chi) )$ time.

\section{Conclusion}
\label{section:Conclusion}

Our work presents the first algorithm for computing the
Skorokhod distance between polygonal traces in $\reals^n$; such traces
arise when sampled-time traces are completed by linear interpolation.
The individual dimensional values, and also the time, can be scaled by different constants
to suit the needs of a particular problem.
Our  algorithm is fully polynomial time for the three norms
$L_1, L_2$ and $ L_{\infty}$, and runs in $O\left(m^3\vdot\log(m) \vdot \poly(n)\right)$ time,
where values come from $\reals^n$ and $m$ is the number of affine segments. 
\footnote{The algorithm can be improved to run in 
$O\left(m^2\vdot \log(m)\vdot \poly(n)\right)$ time using parametric search
and parallel sorting, as mentioned in~\cite{AltG95}  for \frechet metrics in $\reals^2$; 
however these improvements (which can be added on top
of the work presented) are not easily implementable and often involve huge
constants.}
The polynomial factor $\poly(n)$ depends on the norm, and represents the cost
of computing the two geometric primitives.
For example, for the $L_1$, $L_{\infty}$, and $L_1^{\skoro}$ norms, the computation involves linear programming, and 
for $L_2$ and $L_2^{\skoro}$, it requires solving multiple quadratic equations. 
Synergistically, our work also provides complete (and fully polynomial-time) algorithms
for computing the \frechet distance between curves in $\reals^n$ for
the five norms.

In many practical applications, one can restrict attention to retimings within a ``sliding window''
of constant size 
(\emph{e.g.}, we may want to match a point in one trace to a point at most 5 units away in the other trace).
We present a  window optimization that can be used to  gain further efficiency
(the decision problem can be then be solved in linear time assuming
a constant dimension for $\reals^n$).


While we consider the max norm modulo retimings in this work, an interesting question is to design
polynomial time algorithms for norms that involve summing the differences of two traces after retiming.

\bibliographystyle{alpha}
\bibliography{skoro}

\section{Appendix}
We show how to compute the parameters 
$\creach{\mya}_{i,j}^q, \creach{\myb}_{i,j}^q$ for $0\leq q\leq 3$, and $\cpoint_{i,j} $ 
 as follows.
We use the convention that when we we say
 a coordinate of
either $\creach{\mya}_{i,j}^q, \creach{\myb}_{i,j}^q$ has the value $d\in\reals$
 we mean the value $\tuple{d,\cdot}$;
we also let $\overline{d}$ stand for either $\tuple{d, \cdot}$ or
$\tuple{d, +}$.
We also define $\tuple{d, \psi} + \tuple{0, +} = \tuple{d, +}$ for 
$\psi\in\set{\cdot, +}$.
The value $\bot$ added to anything is $\bot$.
We use the expected ordering for $\set{\cdot, +}$; namely
$\tuple{d,\cdot} < \tuple{d,+}$.
We observe that for  $i>0, j>0$, there exists a monotone curve from $(0,0)$
 to the free space
boundaries of cell $i,j$ iff one of the following three hold:

\begin{enumerate}

\item 
There exists a free space monotone curve to the right edge of cell $i-1,j$ which
can continue onto the free space of cell $i,j$.
See Figure~\ref{figure:CellLeft}
\begin{figure}[h]
\strut\centerline{\setlength{\unitlength}{0.00034996in}
\begingroup\makeatletter\ifx\SetFigFont\undefined%
\gdef\SetFigFont#1#2#3#4#5{%
  \reset@font\fontsize{#1}{#2pt}%
  \fontfamily{#3}\fontseries{#4}\fontshape{#5}%
  \selectfont}%
\fi\endgroup%
{\renewcommand{\dashlinestretch}{30}
\begin{picture}(9387,6442)(0,-10)
\texture{55888888 88555555 5522a222 a2555555 55888888 88555555 552a2a2a 2a555555 
	55888888 88555555 55a222a2 22555555 55888888 88555555 552a2a2a 2a555555 
	55888888 88555555 5522a222 a2555555 55888888 88555555 552a2a2a 2a555555 
	55888888 88555555 55a222a2 22555555 55888888 88555555 552a2a2a 2a555555 }
\shade\path(4512,5707)(9012,5707)(9012,1207)
	(4512,1207)(4512,5707)
\path(4512,5707)(9012,5707)(9012,1207)
	(4512,1207)(4512,5707)
\path(1812,5977)(1812,5572)
\path(4512,937)(4512,1342)
\path(4512,937)(4512,1342)
\path(3612,5977)(3612,5572)
\path(6762,5977)(6762,5572)
\path(8112,5977)(8112,5572)
\path(9237,4807)(8787,4807)
\path(9192,3007)(8787,3007)
\path(6762,982)(6762,1387)
\path(6762,982)(6762,1387)
\path(4512,2107)(4062,2107)
\shade\path(12,5707)(4512,5707)(4512,1207)
	(12,1207)(12,5707)
\path(12,5707)(4512,5707)(4512,1207)
	(12,1207)(12,5707)
\whiten\path(4512,1207)(4512,3907)(6762,5707)
	(8112,5707)(9012,4807)(9012,3007)
	(6762,1207)(4512,1207)
\path(4512,1207)(4512,3907)(6762,5707)
	(8112,5707)(9012,4807)(9012,3007)
	(6762,1207)(4512,1207)
\whiten\path(12,3457)(1812,5707)(3612,5707)
	(4512,3907)(4512,2107)(12,2107)(12,3457)
\path(12,3457)(1812,5707)(3612,5707)
	(4512,3907)(4512,2107)(12,2107)(12,3457)
\path(4512,3907)(4062,3907)
\path(4422,3907)(4872,3907)
\put(4332,127){\makebox(0,0)[lb]{\smash{{\SetFigFont{9}{10.8}{\familydefault}{\mddefault}{\updefault}$\mya_{i,j}^3$}}}}
\put(4332,577){\makebox(0,0)[lb]{\smash{{\SetFigFont{9}{10.8}{\familydefault}{\mddefault}{\updefault}$\mya_{i,j}^0$}}}}
\put(9327,2962){\makebox(0,0)[lb]{\smash{{\SetFigFont{9}{10.8}{\familydefault}{\mddefault}{\updefault}$\mya_{i,j}^1$}}}}
\put(9372,4762){\makebox(0,0)[lb]{\smash{{\SetFigFont{9}{10.8}{\familydefault}{\mddefault}{\updefault}$\myb_{i,j}^1$}}}}
\put(8022,6112){\makebox(0,0)[lb]{\smash{{\SetFigFont{9}{10.8}{\familydefault}{\mddefault}{\updefault}$\myb_{i,j}^2$}}}}
\put(6627,6112){\makebox(0,0)[lb]{\smash{{\SetFigFont{9}{10.8}{\familydefault}{\mddefault}{\updefault}$\mya_{i,j}^2$}}}}
\put(6582,577){\makebox(0,0)[lb]{\smash{{\SetFigFont{9}{10.8}{\familydefault}{\mddefault}{\updefault}$\myb_{i,j}^0$}}}}
\put(1632,6112){\makebox(0,0)[lb]{\smash{{\SetFigFont{9}{10.8}{\familydefault}{\mddefault}{\updefault}$\creach{\mya}_{i-1,j}^2$}}}}
\put(3477,6112){\makebox(0,0)[lb]{\smash{{\SetFigFont{9}{10.8}{\familydefault}{\mddefault}{\updefault}$\creach{\myb}_{i-1,j}^2$}}}}
\put(3027,2332){\makebox(0,0)[lb]{\smash{{\SetFigFont{9}{10.8}{\familydefault}{\mddefault}{\updefault}$\creach{\mya}_{i-1,j}^1$}}}}
\put(5007,3907){\makebox(0,0)[lb]{\smash{{\SetFigFont{9}{10.8}{\familydefault}{\mddefault}{\updefault}$\myb_{i,j}^3$}}}}
\put(3027,4042){\makebox(0,0)[lb]{\smash{{\SetFigFont{9}{10.8}{\familydefault}{\mddefault}{\updefault}$\creach{\myb}_{i-1,j}^1$}}}}
\end{picture}
}}
 \caption{Cell $i,j$ and  $i-1,j$ in $\free_{\delta}(f,g)$.}
\label{figure:CellLeft}
\end{figure}
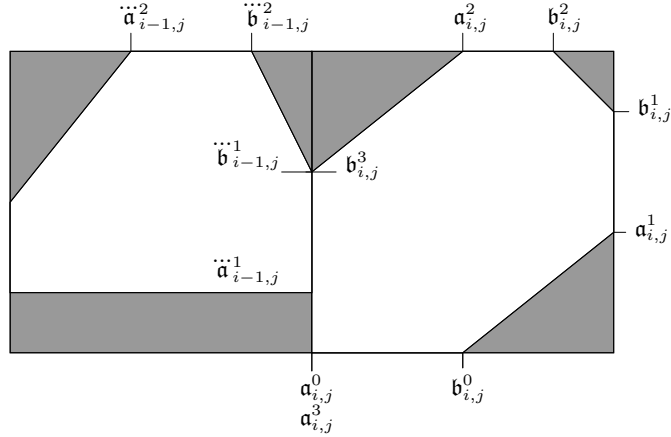
Note that if $\myline(\creach{\mya}_{i-1,j}^1, \creach{\myb}_{i-1,j}^1)$
is non-empty
 then $  \creach{\myb}_{i-1,j}^1 = {\myb}_{i,j}^3$.
This is because any strict monotone increasing curve from $(0,0)$
that can reach $\creach{\mya}_{i-1,j}^1$, passing through the interior of cell
$i-1, j$,  can be modified in cell $i-1, j$ to reach ${\myb}_{i-1,j}^1 (= {\myb}_{i,j}^3)$.

\item There exists a free space monotone curve to the top edge of cell $i,j-1$ which
can continue onto the free space of cell $i,j$.
See Figure~\ref{figure:CellBottom}
\begin{figure}[h]
\strut\centerline{\setlength{\unitlength}{0.00034996in}
\begingroup\makeatletter\ifx\SetFigFont\undefined%
\gdef\SetFigFont#1#2#3#4#5{%
  \reset@font\fontsize{#1}{#2pt}%
  \fontfamily{#3}\fontseries{#4}\fontshape{#5}%
  \selectfont}%
\fi\endgroup%
{\renewcommand{\dashlinestretch}{30}
\begin{picture}(4905,9702)(0,-10)
\texture{55888888 88555555 5522a222 a2555555 55888888 88555555 552a2a2a 2a555555 
	55888888 88555555 55a222a2 22555555 55888888 88555555 552a2a2a 2a555555 
	55888888 88555555 5522a222 a2555555 55888888 88555555 552a2a2a 2a555555 
	55888888 88555555 55a222a2 22555555 55888888 88555555 552a2a2a 2a555555 }
\shade\path(30,9012)(4530,9012)(4530,4512)
	(30,4512)(30,9012)
\path(30,9012)(4530,9012)(4530,4512)
	(30,4512)(30,9012)
\whiten\path(912,4503)(12,5403)(12,7203)
	(912,8553)(1812,9003)(3612,9003)
	(4512,8103)(4512,6303)(3612,4953)
	(2712,4503)(912,4503)
\path(912,4503)(12,5403)(12,7203)
	(912,8553)(1812,9003)(3612,9003)
	(4512,8103)(4512,6303)(3612,4953)
	(2712,4503)(912,4503)
\path(1830,9237)(1830,8832)
\path(1830,4602)(1830,4197)
\path(885,4782)(885,4377)
\path(2730,4332)(2730,4737)
\path(2730,4332)(2730,4737)
\path(3585,9237)(3585,8832)
\path(4755,8112)(4305,8112)
\path(4710,6312)(4305,6312)
\shade\path(30,4512)(4530,4512)(4530,12)
	(30,12)(30,4512)
\path(30,4512)(4530,4512)(4530,12)
	(30,12)(30,4512)
\whiten\path(1380,12)(1380,4512)(2730,4512)
	(4530,3162)(4530,1587)(3180,12)
	(1335,12)(1380,12)
\path(1380,12)(1380,4512)(2730,4512)
	(4530,3162)(4530,1587)(3180,12)
	(1335,12)(1380,12)
\put(1605,9372){\makebox(0,0)[lb]{\smash{{\SetFigFont{9}{10.8}{\familydefault}{\mddefault}{\updefault}$\mya_{i,j}^2$}}}}
\put(3540,9372){\makebox(0,0)[lb]{\smash{{\SetFigFont{9}{10.8}{\familydefault}{\mddefault}{\updefault}$\myb_{i,j}^2$}}}}
\put(4890,8067){\makebox(0,0)[lb]{\smash{{\SetFigFont{9}{10.8}{\familydefault}{\mddefault}{\updefault}$\myb_{i,j}^1$}}}}
\put(4845,6267){\makebox(0,0)[lb]{\smash{{\SetFigFont{9}{10.8}{\familydefault}{\mddefault}{\updefault}$\mya_{i,j}^1$}}}}
\put(750,4917){\makebox(0,0)[lb]{\smash{{\SetFigFont{9}{10.8}{\familydefault}{\mddefault}{\updefault}$\mya_{i,j}^0$}}}}
\put(2505,4917){\makebox(0,0)[lb]{\smash{{\SetFigFont{9}{10.8}{\familydefault}{\mddefault}{\updefault}$\myb_{i,j}^0$}}}}
\put(2730,4017){\makebox(0,0)[lb]{\smash{{\SetFigFont{9}{10.8}{\familydefault}{\mddefault}{\updefault}$\creach{\myb}_{i,j-1}^2$}}}}
\put(1425,4017){\makebox(0,0)[lb]{\smash{{\SetFigFont{9}{10.8}{\familydefault}{\mddefault}{\updefault}$\creach{\mya}_{i,j-1}^2$}}}}
\end{picture}
}}
 \caption{Cell $i,j$ and  $i,j-1$ in $\free_{\delta}(f,g)$.}
\label{figure:CellBottom}
\end{figure}
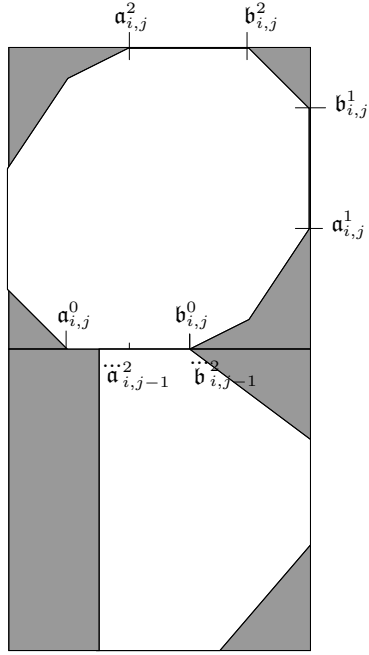
We observe that if $\myline(\creach{\mya}_{i,j-1}^2, \creach{\myb}_{i,j-1}^2)$
is non-empty 
 then $  \creach{\myb}_{i,j-1}^2 = {\myb}_{i,j}^0$.
This is because any strict monotone increasing curve from $(0,0)$
that can reach $\creach{\mya}_{i,j-1}^2$, passing through the interior of cell
$i, j-1$,  can be modified in cell $i, j-1$ to reach ${\myb}_{i,j-1}^2 (= {\myb}_{i,j}^0)$.

\item There exists a free space monotone curve to the point $i,j$ of cell $i-1,j-1$
 which can continue onto the free space of cell $i,j$.
See Figure~\ref{figure:CellCorner}.
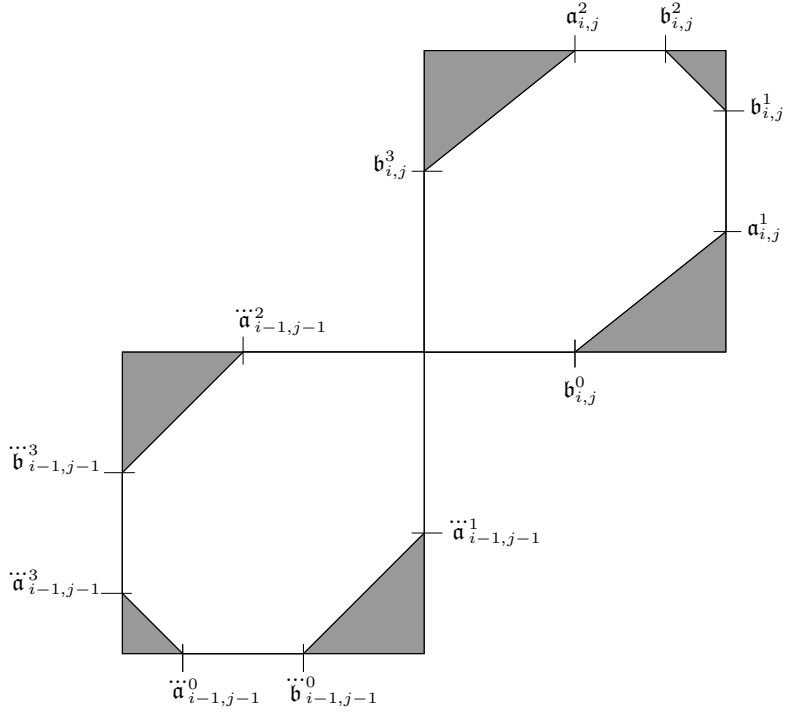
\begin{figure}[h]
\strut\centerline{\setlength{\unitlength}{0.00034996in}
\begingroup\makeatletter\ifx\SetFigFont\undefined%
\gdef\SetFigFont#1#2#3#4#5{%
  \reset@font\fontsize{#1}{#2pt}%
  \fontfamily{#3}\fontseries{#4}\fontshape{#5}%
  \selectfont}%
\fi\endgroup%
{\renewcommand{\dashlinestretch}{30}
\begin{picture}(11100,10537)(0,-10)
\texture{55888888 88555555 5522a222 a2555555 55888888 88555555 552a2a2a 2a555555 
	55888888 88555555 55a222a2 22555555 55888888 88555555 552a2a2a 2a555555 
	55888888 88555555 5522a222 a2555555 55888888 88555555 552a2a2a 2a555555 
	55888888 88555555 55a222a2 22555555 55888888 88555555 552a2a2a 2a555555 }
\shade\path(6225,9802)(10725,9802)(10725,5302)
	(6225,5302)(6225,9802)
\path(6225,9802)(10725,9802)(10725,5302)
	(6225,5302)(6225,9802)
\whiten\path(6225,5302)(6225,8002)(8475,9802)
	(9825,9802)(10725,8902)(10725,7102)
	(8475,5302)(6225,5302)
\path(6225,5302)(6225,8002)(8475,9802)
	(9825,9802)(10725,8902)(10725,7102)
	(8475,5302)(6225,5302)
\shade\path(1725,5302)(6225,5302)(6225,802)
	(1725,802)(1725,5302)
\path(1725,5302)(6225,5302)(6225,802)
	(1725,802)(1725,5302)
\whiten\path(1725,3502)(3525,5302)(6225,5302)
	(6225,2602)(4425,802)(2625,802)
	(1725,1702)(1725,3502)
\path(1725,3502)(3525,5302)(6225,5302)
	(6225,2602)(4425,802)(2625,802)
	(1725,1702)(1725,3502)
\path(3525,5527)(3525,5122)
\path(1455,3502)(1905,3502)
\path(1410,1702)(1860,1702)
\path(6045,2602)(6495,2602)
\path(2625,532)(2625,937)
\path(2625,532)(2625,937)
\path(4425,532)(4425,937)
\path(4425,532)(4425,937)
\path(8475,10027)(8475,9622)
\path(9825,10027)(9825,9622)
\path(10995,8902)(10545,8902)
\path(10950,7102)(10545,7102)
\path(8475,5077)(8475,5482)
\path(8475,5077)(8475,5482)
\path(6045,8002)(6495,8002)
\put(8340,10207){\makebox(0,0)[lb]{\smash{{\SetFigFont{9}{10.8}{\familydefault}{\mddefault}{\updefault}$\mya_{i,j}^2$}}}}
\put(9735,10207){\makebox(0,0)[lb]{\smash{{\SetFigFont{9}{10.8}{\familydefault}{\mddefault}{\updefault}$\myb_{i,j}^2$}}}}
\put(11085,8857){\makebox(0,0)[lb]{\smash{{\SetFigFont{9}{10.8}{\familydefault}{\mddefault}{\updefault}$\myb_{i,j}^1$}}}}
\put(11040,7057){\makebox(0,0)[lb]{\smash{{\SetFigFont{9}{10.8}{\familydefault}{\mddefault}{\updefault}$\mya_{i,j}^1$}}}}
\put(5460,8002){\makebox(0,0)[lb]{\smash{{\SetFigFont{9}{10.8}{\familydefault}{\mddefault}{\updefault}$\myb_{i,j}^3$}}}}
\put(8295,4627){\makebox(0,0)[lb]{\smash{{\SetFigFont{9}{10.8}{\familydefault}{\mddefault}{\updefault}$\myb_{i,j}^0$}}}}
\put(3390,5662){\makebox(0,0)[lb]{\smash{{\SetFigFont{9}{10.8}{\familydefault}{\mddefault}{\updefault}$\creach{\mya}_{i-1,j-1}^2$}}}}
\put(6585,2512){\makebox(0,0)[lb]{\smash{{\SetFigFont{9}{10.8}{\familydefault}{\mddefault}{\updefault}$\creach{\mya}_{i-1,j-1}^1$}}}}
\put(2400,127){\makebox(0,0)[lb]{\smash{{\SetFigFont{9}{10.8}{\familydefault}{\mddefault}{\updefault}$\creach{\mya}_{i-1,j-1}^0$}}}}
\put(4155,127){\makebox(0,0)[lb]{\smash{{\SetFigFont{9}{10.8}{\familydefault}{\mddefault}{\updefault}$\creach{\myb}_{i-1,j-1}^0$}}}}
\put(15,3592){\makebox(0,0)[lb]{\smash{{\SetFigFont{9}{10.8}{\familydefault}{\mddefault}{\updefault}$\creach{\myb}_{i-1,j-1}^3$}}}}
\put(15,1747){\makebox(0,0)[lb]{\smash{{\SetFigFont{9}{10.8}{\familydefault}{\mddefault}{\updefault}$\creach{\mya}_{i-1,j-1}^3$}}}}
\end{picture}
}}
 \caption{Cell $i,j$ and  $i-1,j-1$ in $\free_{\delta}(f,g)$.}
\label{figure:CellCorner}
\end{figure}
However, this case can be subsumed in the previous two, as the point
 $i,j$ of cell $i-1,j-1$ is also a point of cell $i-1, j$, \emph{i.e.}, the cell to
the left of cell  $i,j$.

\end{enumerate}

 Note that in contrast to the free space, the portion of the cell reachable by a 
 strict monotone increasing curve from $(0,0)$ need not be convex.
Based on the above analysis, we define  monotone curve reachable free space
cell boundaries as follows recursively. In the following, we assume
that the left, or bottom  free space cell boundaries of cell $i,j$ are non-empty. 
If they both are empty, then the $\creach{\mya}_{i,j}^q, \creach{\myb}_{i,j}^q$
parameters get the value $\bot$, as in this case no monotone increasing
curve can enter cell  $i,j$.
Moreover, for all $q$, if ${\mya}_{i,j}^q = \bot$ then
$\creach{\mya}_{i,j}^q = \bot$; and
if ${\myb}_{i,j}^q = \bot$ then
$\creach{\myb}_{i,j}^q = \bot$.

\begin{itemize}

\item[\textbf{Boundary Point.}]
  Intuitively, this set contains the left and bottommost boundary point if it
  is reachable from a strict monotone curve
  We define $\cpoint_{0,0}= \set{(0,0)}$.
  For $i>0, j>0$, we define   $\cpoint_{i,0} = \cpoint_{0,j} = \emptyset$.
  Note that $(i,0)$ and $(0,j)$ for $i>0, j>0$
  are not reachable by a monotone curve from $(0,0)$ since both coordinates must
  keep \emph{strictly} increasing.
  For $i>0, j>0$ we define 
  $\cpoint_{i,j} = \set{(i,j)}$ iff $\first( \creach{\myb}_{i-1,j-1}^2) = i$ or
  $ \second(\creach{\myb}_{i-1,j-1}^1) = j$. 
  See Figure~\ref{figure:CellCorner} which illustrates
  the case when $\cpoint_{i,j} = \set{(i,j)}$.

\item[\textbf{Left Subline.}] 
We define
$\creach{\mya}_{0,0}^3 =  \creach{\myb}_{0,0}^3 = 0 $ if
$(0,0) \in \free_{\delta}(f,g)$; otherwise these values are $\bot$.
For all $j>0$ we define
$\creach{\mya}_{0,j}^3 = \creach{\myb}_{i,0}^3 =  \bot$.
For  $i>0, j\geq 0$, we define
$\creach{\mya}_{i,j}^3 =\creach{\mya}_{i-1,j}^1$, and
$\creach{\myb}_{i,j}^3 = \creach{\myb}_{i-1,j}^1$.
That is, we reduce the computation of the reachable portion of the  left boundary
of cell $i,j$
to the computation of the reachable portion of the right boundary of cell
$i-1, j$.
See Figure~\ref{figure:CellLeft} which illustrates this case.

\item[\textbf{Bottom Subline.}] 
We define
 $\creach{\mya}_{0,0}^0 =  \creach{\myb}_{0,0}^0 = 0 $ if
$(0,0) \in \free_{\delta}(f,g)$; otherwise these values are $\bot$.
We define
$\creach{\mya}_{i,0}^0 = \creach{\myb}_{i,0}^0 =  \bot$ for all $i>0$.
For  $i\geq 0, j>0$, we define
$\creach{\mya}_{i,j}^0 =\creach{\mya}_{i,j-1}^2$,
and $\creach{\myb}_{i,j}^0 =\creach{\myb}_{i,j-1}^2$.
That is, we reduce the computation of the reachable portion of the  bottom boundary
of cell $i,j$
to the computation of the reachable portion of the upper boundary of cell
$i, j-1$.

See Figure~\ref{figure:CellBottom} which illustrates this case.

\item[\textbf{Right Subline.}] 
We observe the following facts.
\begin{enumerate}
\item 
For cell $i,j$, if a point $(d,j)$ with $d<i+1$  on the bottom boundary of the cell $i,j$
is reachable by a  free space monotone curve, then
all points on $\myline({\mya}_{i,j}^1, {\myb}_{i,j}^1)$ are reachable
by a free space monotone continuation of that curve.
\item 
For cell $i,j$, let a point $p$ be reachable by a  free space monotone curve
passing through the left boundary point $(i,d)$.
Then any other free space monotone curve passing through a lower
left boundary point $(i,d')$ with $d'<d$ can be extended to be  a
free space monotone curve to $p$. 
Thus, the lower the point on the left boundary, the  ``better'' it is 
for a monotone curve to reach the most points on the right boundary.
\end{enumerate}
Using the above facts, we obtain the parameters 
$\creach{\mya}_{i,j}^1,  \creach{\myb}_{i,j}^1$ as follows.
We assume $(0,0)$ belongs to the free space, if not, then all parameters are $\bot$.
\begin{enumerate}
\item 
We define
$\creach{\mya}_{0,0}^1 = {\mya}_{0,0}^1$ if  $\second({\mya}_{0,0}^1)>  0$,
otherwise $\creach{\mya}_{0,0}^1 = (1, \tuple{0,+})$.
Similarly, we also define $\creach{\myb}_{0,0}^1 = {\myb}_{0,0}^1$ if  
$\second({\myb}_{0,0}^1)> 0$,
otherwise $\creach{\myb}_{0,0}^1 = (1, \tuple{0,+})$.
The reason is that any point on the 
 line segment $\myline({\mya}_{0,0}^1, {\myb}_{0,0}^1)$ can
be reached from $(0,0)$ by a  strict monotone curve, except the point $(1, 0)$.

\item For $i>0, j=0$ (\emph{i.e}, the first row), we define $\creach{\mya}_{i,0}^1$ and $ \creach{\myb}_{i,0}^1$ 
as follows using $\creach{\mya}_{i,0}^3$ and $\creach{\mya}_{i,0}^3$;
thus, we obtain the parameters for the right boundary using the parameters for the
left bpundary of the same cell. 
Note that in this case a monotone curve can enter the free space of cell $i,0$ only
from the cell $i-1,0$ (the cell to the left).
\begin{itemize}
\item 
If $\second(\creach{\mya}_{i,0}^3) \geq \second({\myb}_{i,0}^1)$, then 
 $\creach{\mya}_{i,0}^1 =  \creach{\myb}_{i,0}^1=  \bot$ as in this case
a strictly increasing  monotone curve cannot reach the right boundary of the cell $i,0$
at all.
\item Otherwise if  $\second(\creach{\mya}_{i,0}^3) < \second({\myb}_{i,0}^1)$,
\begin{compactitem}
\item 
If $\second(\creach{\mya}_{i,0}^3) =  \second({\mya}_{i,0}^1)$,
then
$\creach{\mya}_{i,0}^1 = {\mya}_{i,0}^1  +  \tuple{+, 0}$ (we add a $+$
as the point  ${\mya}_{i,0}^1$ itself is unreachable by a strict monotone curve);
 and $\creach{\myb}_{i,0}^1 = {\myb}_{i,0}^1$.
If after this, we get $\creach{\mya}_{i,0}^1 > \creach{\myb}_{i,0}^1 $, we set both
to $\bot$.

\item 
Otherwise if $\second(\creach{\mya}_{i,0}^3) <  \second({\mya}_{i,0}^1)$,
then
$\creach{\mya}_{i,0}^1 = {\mya}_{i,0}^1$
and $\creach{\myb}_{i,0}^1 = {\myb}_{i,0}^1$.

\end{compactitem}

\end{itemize}

\item  For $0, j>0$ (the first column), 
 we define $\creach{\mya}_{0,j}^1$ and $ \creach{\myb}_{0,j}^1$ 
as follows. 
Note that in this case a monotone curve can enter the free space of cell $0,j$ only
from the  top end of cell $0,j-1$.
If $\myline(\creach{\mya}_{0,j}^0,  \creach{\myb}_{0,j}^0)$ is empty,
then both 
$\creach{\mya}_{0,j}^1$ and $ \creach{\myb}_{0,j}^1$ 
have the value $\bot$.
Otherwise,
\begin{compactitem}
\item 
If $\first(\creach{\mya}_{0,j}^0) =\tuple{1, \cdot}$, 
(in this case we also have $\first(\creach{\myb}_{0,j}^0) =\tuple{1, \cdot})$, then
$\creach{\mya}_{0,j}^1 = \creach{\myb}_{0,j}^1 = (1, \tuple{j, \cdot}$, as the
bottom point of the right edge is the only point allowed by a  strictly increasing
monotone curve.
\item $\first(\creach{\mya}_{0,j}^0) <\tuple{1, \cdot}$, 
then,
$\creach{\mya}_{0,j}^1 = {\mya}_{0,j}^1$ and
$\creach{\myb}_{0,j}^1 = {\myb}_{0,j}^1$.
This is because all point on the right bpundary
 $\myline( {\mya}_{0,j}^1 ,  {\myb}_{0,j}^1$ 
can be reached from the point on the bottom  boundary
$\creach{\mya}_{0,j}^0$.
\end{compactitem}

\item  For $i>0, j>0$,

  \begin{compactitem}
  \item if $\cpoint_{i,j}  = \set{(i,j)}$, we have 
    $\creach{\mya}_{i,j}^1 = \max({\mya}_{i,j}^1, (i+1, \tuple{j,+}))$, and
    $\creach{\myb}_{i,j}^1 = \max({\myb}_{i,j}^1, (i+1, \tuple{j,+}))$.
    This is because the entire line segment
    $\myline({\mya}_{i,j}^1,{\myb}_{i,j}^1)$ except for the point
    $ (i+1, \tuple{j,\cdot})$ can be reached by a strictly increasing monotone
    curve from the corner point $(i,j)$.

  \item  if  $\myline(\creach{\mya}_{i,j}^0,  \creach{\myb}_{i,j}^0)$
    is non-empty  and $\cpoint_{i,j}  = \emptyset$:
    \begin{compactitem}
    \item 
      If $\first(\creach{\mya}_{i,j}^0) <\tuple{i+1, \cdot}$, then
      $\creach{\mya}_{i,j}^1 = {\mya}_{i,j}^1$ and
      $\creach{\myb}_{i,j}^1 = {\myb}_{i,j}^1$.
      This corresponds to the case where the free space monotone curve enters
      cell $i,j$ through either the corner point $i,j$, or through the bottom edge 
      (through a point which is not $(i+1,j)$ of the cell.
    \item 
      If  $\first(\creach{\mya}_{i,j}^0) = \tuple{i+1, \cdot}$ and
      $\cpoint_{i,j}  = \emptyset$, then
       $\creach{\mya}_{i,j}^1 = \creach{\myb}_{i,j}^1 = (i+1\tuple{j, \cdot})$.
       \emph{intially}.
       We change these values if step $\dagger$ below adds any more reachable
       points on the right boundary.
       The point $(i+1, j)$ will be covered by the corner point of cell $i+1,j$
       \footnote{
         Here we see the utility of keeping track of the corner points.
         If we required only non-decreasing monotone curves, then the whole
         segment $\myline({\mya}_{i,j}^1,  {\myb}_{i,j}^1)$ would be reachable
         from the corner point  $(i+1, j)$, and there would not be a  need to treat corner
         points differently.}.
     \end{compactitem}

   \item 
     Otherwise ($\myline(\creach{\mya}_{i,j}^0,  \creach{\myb}_{i,j}^0) $
     and $\cpoint_{i,j}  $ are empty, and the analysis is
          as in the case $i>0, j=0$  done previously.
          (the monotone increasing curve enters the cell at the left edge):
          \begin{compactitem}
          \item 
            If $\second(\creach{\mya}_{i,j}^3) \geq \second({\myb}_{i,j}^1)$, then 
            $\creach{\mya}_{i,j}^1 =  \creach{\myb}_{i,j}^1=  \bot$ as in this case
            a strictly increasing  monotone curve cannot reach the right boundary of the cell $i,j$
            at all.
          \item Otherwise if  $\second(\creach{\mya}_{i,j}^3) < \second({\myb}_{i,j}^1)$,
            (\textbf{\textsf{Step }$\dagger$})
            \begin{compactitem}
            \item 
              If $\second(\creach{\mya}_{i,j}^3) =  \second({\mya}_{i,j}^1)$,
              then
              $\creach{\mya}_{i,j}^1 = {\mya}_{i,j}^1  +  \tuple{+, 0}$ (we add a $+$
              as the point  ${\mya}_{i,j}^1$ itself is unreachable by a strict monotone curve);
              and $\creach{\myb}_{i,j}^1 = {\myb}_{i,j}^1$.
              If after this, we get $\creach{\mya}_{i,j}^1 > \creach{\myb}_{i,j}^1 $, we set both
              to $\bot$.
              
            \item 
              Otherwise if $\second(\creach{\mya}_{i,j}^3) <  \second({\mya}_{i,j}^1)$,
              then
              $\creach{\mya}_{i,j}^1 =  {\mya}_{i,0}^1$
              and $\creach{\myb}_{i,j}^1 = {\myb}_{i,j}^1$.
              
            \end{compactitem}
            
          \end{compactitem}

\end{compactitem}
\end{enumerate}

\item[\textbf{Top Subline.}]
The analysis of this case is similar to that for the right subline.
We observe the following facts.
\begin{enumerate}
\item 
For cell $i,j$, if a point $(i,d)$ with $d<j+1$  on the left boundary of the cell
is reachable by a  free space monotone curve, then
all points on $\myline({\mya}_{i,j}^2, {\myb}_{i,j}^2)$ are reachable
by a free space monotone continuation of that curve.
\item 
For cell $i,j$, let a point $p$ be reachable by a  free space monotone curve
passing through the bottom boundary point $(d,j)$.
Then any other free space monotone curve passing through a ``more left''
bottom boundary point $(d', j)$ with $d'<d$ can be extended to be  a
free space monotone curve to $p$. 
Thus, the more left the  point on the bottom boundary, the  ``better'' it is 
for a monotone curve to reach the most points on the top boundary.
\end{enumerate}
Using the above facts, we obtain the parameters 
$\creach{\mya}_{i,j}^2,  \creach{\myb}_{i,j}^2$ as follows.
\begin{enumerate}
\item 
We define
$\creach{\mya}_{0,0}^2 = {\mya}_{0,0}^2$ if  ${\mya}_{0,0}^2> ( 0, 1)$,
otherwise $\creach{\mya}_{0,0}^1 = (\tuple{0,+}, \tuple{1,\cdot})$.
Similarly, we also define $\creach{\myb}_{0,0}^2 = {\myb}_{0,0}^2$ if  
${\myb}_{0,0}^2> ( 0,1)$,
otherwise $\creach{\myb}_{0,0}^2 = (\tuple{0,+}, \tuple{1,\cdot})$.
The reason is that any point on the 
line segment $\myline({\mya}_{0,0}^2, {\myb}_{0,0}^2)$ can
be reached from $(0,0)$ by a  strict monotone curve, except the point $(0,1)$.

\item For $i=0, j>0$, we define $\creach{\mya}_{i,0}^1$ and $ \creach{\myb}_{i,0}^1$ 
as follows. 
Note that in this case a monotone curve can enter the free space of cell $i,0$ only
from the cell $0,j-1$.
\begin{itemize}
\item 
If $\creach{\mya}_{0,j}^0 \geq {\myb}_{0,j}^2$, then 
 $\creach{\mya}_{0, j}^2 =  \creach{\myb}_{0,j}^2=  \bot$ as in this case
a strictly increasing 
 monotone curve cannot reach the top boundary of the cell $i,0$.
\item 
Otherwise if $\creach{\mya}_{0,j}^0 = {\mya}_{0,j}^2$, then 
(and $\creach{\mya}_{0,j}^0 < {\myb}_{0,j}^2$), then 
$\creach{\mya}_{0,j}^2 = {\mya}_{0,j}^2  +  \tuple{0,+}$ (we add a $+$
as the point  ${\mya}_{0, j}^2$ itself is unreachable by a strict monotone curve);
 and $\creach{\myb}_{0,j}^2 = {\myb}_{0,j}^2$.

\item 
Otherwise if $\creach{\mya}_{0,j}^0  \neq {\mya}_{0,j}^2$,  
(and the previous
two cases do not hold),  then
$\creach{\mya}_{0,j}^2 = \max\left( \creach{\mya}_{0, j}^2, {\mya}_{0, j}^2 \right)$;
and $\creach{\myb}_{0,j}^2 = {\myb}_{0,j}^2$.
\end{itemize}

\item  For $i>0, j=0$, 
 we define $\creach{\mya}_{i,0}^2$ and $ \creach{\myb}_{i, 0}^2$ 
as follows. 
Note that in this case a monotone curve can enter the free space of cell $i, 0$ only
from the  right  end of cell $i-1, 0$.
If $\myline(\creach{\mya}_{i,0}^3,  \creach{\myb}_{i, 0}^3)$ is non-empty,
then $\creach{\mya}_{i,0}^2 = {\mya}_{i, 0}^2$ and
$\creach{\myb}_{i, 0}^2 = {\myb}_{i,0}^2$; otherwise both 
$\creach{\mya}_{i,0}^2$ and $ \creach{\myb}_{i,0}^2$ 
have the value $\bot$.

\item  For $i>0, j>0$, 
\begin{itemize}
\item 
if  either $\myline(\creach{\mya}_{i,j}^3,  \creach{\myb}_{i,j}^3)$ is
non-empty, or $\cpoint_{i,j} = \set{(i,j)}$, then
$\creach{\mya}_{i,j}^3 = {\mya}_{i,j}^3$ and
$\creach{\myb}_{i,j}^3 = {\myb}_{i,j}^3$.
This corresponds to the case where the free space monotone curve enters
cell $i,j$ through either the corner point $i,j$, or through the left edge of the cell.
\item 
Otherwise, we have as in the case $i=0, j>0$ above
(the monotone increasing curve enters the cell at the bottom edge):
\begin{itemize}
\item 
If $\creach{\mya}_{i,j}^0 \geq {\myb}_{i,j}^2$, then 
 $\creach{\mya}_{i,j}^2 =  \creach{\myb}_{i,j}^2=  \bot$ as in this case
a strict monotone curve cannot reach the top boundary of the cell $i,j$.
\item 
Otherwise if $\creach{\mya}_{i,j}^0 =  {\mya}_{i,j}^2$ 
(and $\creach{\mya}_{i,j}^0 < {\myb}_{i,j}^2$),  then
$\creach{\mya}_{i,j}^2 = {\mya}_{i,j}^0  +  \tuple{+, 0}$ (we add a $+$
as the point  ${\mya}_{i,j}^2$ itself is unreachable by a strict monotone curve);
 and $\creach{\myb}_{i,j}^2 = {\myb}_{i,j}^2$.

\item 
Otherwise if $\creach{\mya}_{i,j}^0 \neq  {\mya}_{i,j}^2$ (and the previous
two cases do not hold),  then
$\creach{\mya}_{i,j}^2 = \max\left( \creach{\mya}_{i,j}^0, {\mya}_{i,j}^2\right)$;
and $\creach{\myb}_{i,j}^2 = {\myb}_{i,j}^2$.
\end{itemize}
Note that in this case a monotone curve can enter the free space of cell $i,j$ only
from the cell $i,j-1$.
\end{itemize}

\end{enumerate}

\end{itemize}

\end{document}